\documentclass{article}
\usepackage[%
journal=MSL,    %% Replace XXX by one of the above journal codes.
lang=american,   %% Change to `american' if you use American English.
]{ems-journal}

\usepackage{dsfont}
\usepackage{mathtools}
\usepackage{mathrsfs}
\usepackage{xfrac}

\usepackage{algorithm}
\usepackage{algpseudocode}

\usepackage{cite}

\newcommand{\Expe}{\mathbb{E}}
\newcommand{\Prob}{\mathbb{P}}
\newcommand{\Var}{\mathrm{Var}}
\newcommand{\R}{\mathbb{R}}
\newcommand{\C}{\mathbb{C}}

\newcommand{\N}{\mathbb{N}}
\newcommand{\indicator}[1]{\mathds{1}_{#1}}
\newcommand{\suchthat}{\; \vert \;}
\newcommand{\widesim}[2][1.5]{\mathrel{\underset{#2}{\scalebox{#1}[1]{$\sim$}}}}

\DeclareMathOperator*{\argmax}{arg\,max}
\DeclareMathOperator*{\argmin}{arg\,min}

\newcommand{\Pcorr}{\mathbb{P}^{\,\text{corr}}}
\newcommand{\Pcorrd}{\mathbb{P}^{\,\text{corr}}_d}
\newcommand{\Pcorrdplus}{\mathbb{P}^{\,\text{corr}}_{d+1}}
\newcommand{\Pcorrdminus}{\mathbb{P}^{\,\text{corr}}_{d-1}}
\newcommand{\Pind}{\mathbb{P}^{\,\text{ind}}}
\newcommand{\Pindd}{\mathbb{P}^{\,\text{ind}}_d}
\newcommand{\Pinddplus}{\mathbb{P}^{\,\text{ind}}_{d+1}}
\newcommand{\Pinddminus}{\mathbb{P}^{\,\text{ind}}_{d-1}}
\newcommand{\Pdisj}{\mathbb{P}^{\,\text{dis}}}
\newcommand{\Pdisjd}{\mathbb{P}^{\,\text{dis}}_d}

\newcommand{\Pdisjdminus}{\mathbb{P}^{\,\text{dis}}_{d-1}}
\newcommand{\Ecorrd}{\mathbb{E}^{\,\text{corr}}_d}
\newcommand{\Ecorrdplus}{\mathbb{E}^{\,\text{corr}}_{d+1}}
\newcommand{\Ecorr}{\mathbb{E}^{\,\text{corr}}}
\newcommand{\Eindd}{\mathbb{E}^{\,\text{ind}}_d}
\newcommand{\Einddplus}{\mathbb{E}^{\,\text{ind}}_{d+1}}
\newcommand{\Einddminus}{\mathbb{E}^{\,\text{ind}}_{d-1}}
\newcommand{\Eind}{\mathbb{E}^{\,\text{ind}}}
\newcommand{\Edisjdminus}{\mathbb{E}^{\, \text{dis}}_{d-1}}

\theoremstyle{plain}
	\newtheorem{theorem}{Theorem}[section]
	\newtheorem{lemma}[theorem]{Lemma}
	\newtheorem{corollary}[theorem]{Corollary}
	\newtheorem{proposition}[theorem]{Proposition}
\theoremstyle{definition}
	\newtheorem{definition}{Definition}[section]
	
	\newtheorem{remark}[theorem]{Remark}

\numberwithin{equation}{section}

\begin{document}

%------
% Insert the title of your paper and (if necessary)
% a short title for the running head.
%------
\title{Asymmetric graph alignment and the phase transition for asymmetric tree correlation testing}
\titlemark{Asymmetric graph alignment and tree correlation testing}

%------

%%%% Pls fill in all fields for each author
%%%% Label the authors by their position in the authors' list using {}
%%%% If you published any math paper ever, you have an MR Author ID.
%  Please look it up in three easy (and free) steps:
% 1. copy the bibliographic data of any published paper (co-)authored by you in the search field at https://mathscinet.ams.org/mathscinet/freetools/mref
% 2. Hit your name in the search result
% 3. Find your MR Author ID in the first row, copy it in the \mrid{} field
%%%% If you have not created your ORCID yet, you may like to do it now, pls copy it in the field \orcid{}
%%%% Abbreviate first names for the running head

\emsauthor{1}{
	\givenname{Jakob}
	\surname{Maier}
	\mrid{}
	\orcid{0000-0002-7697-8457}}{J.~Maier}
%%%% Repeat the same fields for each numbered author
\emsauthor{2}{
	\givenname{Laurent}
	\surname{Massoulié}
	\mrid{354412}
	\orcid{0000-0001-7263-0069}}{L.~Massoulié}

%%%% Please provide detailed address info for each author
%%%% Use the same numbering as for \emsauthor above
%%%% Please look up the ROR ID of your institute here: https://ror.org
\Emsaffil{1}{
	\department{INRIA Paris}
	\organisation{ENS - PSL Research University}
	\rorid{05a0dhs15}
	\address{48 Rue Barrault}
	\zip{75013}
	\city{Paris}
	\country{France}
	\affemail{jakob.maier@inria.fr}
}
%%%% Repeat the same fields for each numbered author
%%%% If some author has multiple affiliations, repeat the fields for each affiliation
%%%% Number the affiliations using {}
\Emsaffil{2}{
	\department{INRIA Paris}
	\organisation{ENS - PSL Research University}
	\rorid{05a0dhs15}
	\address{48 Rue Barrault}
	\zip{75013}
	\city{Paris}
	\country{France}
	\affemail{laurent.massoulie@inria.fr}
}
%------
% Add MSC 2020 codes according to https://zbmath.org/classification/.
% A unique primary MSC code (in curly brackets) is mandatory,
% while secondary MSC codes (in square brackets) are optional.
%------
\classification[62B10]{05C80}

%------
% Add a list of keywords.
%------

%------
% Insert your abstract.
%------
\begin{abstract}
	Graph alignment --- identifying node correspondences between two graphs --- is a fundamental problem with applications in network analysis, biology, and privacy research. While substantial progress has been made in aligning correlated Erdős--Rényi graphs under symmetric settings, real-world networks often exhibit asymmetry in both node numbers and edge densities. In this work, we introduce a novel framework for asymmetric correlated Erdős--Rényi graphs, generalizing existing models to account for these asymmetries. We conduct a rigorous theoretical analysis of graph alignment in the sparse regime, where local neighborhoods exhibit tree-like structures. Our approach leverages tree correlation testing as the central tool in our polynomial-time algorithm, \textsc{MPAlign}, which achieves one-sided partial alignment under certain conditions.
	
	A key contribution of our work is characterizing these conditions under which asymmetric tree correlation testing is feasible:  If two correlated graphs $G$ and $G'$ have average degrees $\lambda s$ and $\lambda s'$ respectively, where $\lambda$ is their common density and $s,s'$ are marginal correlation parameters, their tree neighborhoods can be aligned if $ss' > \alpha$, where $\alpha$ denotes Otter's constant and  $\lambda$ is supposed large enough. The feasibility of this tree comparison problem undergoes a sharp phase transition since $ss' \leq \alpha$ implies its impossibility. These new results on tree correlation testing allow us to solve a class of random subgraph isomorphism problems, resolving an open problem in the field.
\end{abstract}

\maketitle

\tableofcontents

%------
% INSERT THE BODY OF THE PAPER HERE (except
% acknowledgments, funding info and bibliography)
%------

\section{Introduction}

For a pair of simple graphs $(G, G')$ with randomly labeled node sets $(V, V')$ there are several prominent questions about the joint structure of $G$ and $G'$: 
\begin{itemize}
	\item Are they structurally identical in the sense that there is a bijective function between $V$ and $V'$ mapping edges to edges and non-edges to non-edges?\\This is the \textbf{graph isomorphism problem} (GIP).
	\item Can $G$ be embedded as a subgraph of $G'$, meaning that there is an injective mapping $V \hookrightarrow V'$  preserving the presence of edges? \\This is the \textbf{subgraph isomorphism problem} (SIP).
	\item Without any assumption on the graphs, which one-to-one mapping between their nodes maximizes edge overlap? \\This is the \textbf{graph alignment problem} (also called graph matching problem) formulated as an instance of the quadratic assignment problem (QAP):  Denoting the adjacency matrices of $G$ and $G'$ as $A$ and $A'$, this problem reads
	\begin{equation}\label{QAP}
		\sigma_\text{\tiny{QAP}} \in \argmax_{\substack{\sigma : V \rightarrow V',\\ \sigma \text{ injective}}} \, \, \, \sum_{i, j \in V} A_{i, j} A'_{\sigma(i), \sigma(j)}. \tag{{QAP}}
	\end{equation}
\end{itemize}
These three problems are increasingly difficult since solving one yields a solution for the ones listed above. While graph isomorphisms can be found in quasi-polynomial time \cite{babai2018group}, the subgraph isomorphism problem is NP-hard since it can detect a Hamiltonian path by taking $G$ to be the cycle of length $|V'|$. Even more difficult is ({QAP}), which solves the traveling salesman, $k$-clique, and subgraph isomorphism problems, among many others \cite{cela_quadratic_1998}.

While aligning graphs is NP-hard in the worst case, real-world questions are not necessarily of this complexity. Applications typically consider instances of the graph alignment problem somewhere on a spectrum between hard QAP and easy GIP, potentially allowing for polynomial time algorithms to recover node correspondences.

To model this interpolation between finding simple isomorphisms and solving worst-case QAP, a line of research initiated by \cite{pedarsani2011privacy} has considered $G$ and $G'$ correlated random graphs. This can be seen as an average case analysis of the graph alignment problem, allowing statistical tools to be used. In contrast to the QAP formulation, this approach supposes the existence of a true node correspondence $\sigma_*:V \to V'$ planted into the model. The goal of \textbf{planted graph alignment} is to find an estimator $\hat{\sigma}$ which coincides with $\sigma_*$ as much as possible.

The existing theory on correlated random graph alignment supposes that $G$ and $G'$ have identical node quantities and the same expected edge numbers. This symmetric framework is restrictive in many applications: Protein interaction graphs \cite{aladaug2013spinal} have different edge densities across species, brain graphs \cite{thual_aligning_2022} have different node numbers, and social network analysis requires the retrieval of subgraphs within a big graph \cite{fan_graph_2012}.

The present paper addresses all of these issues by introducing the novel, more general framework of \textbf{asymmetric correlated Erdős--Rényi graphs}, which allows for varying edge densities and differing node numbers. This model can be seen as an average-case interpolation between (QAP) on the one hand and the subgraph isomorphism problem (SIP) -- which is strictly more difficult than the GIP -- on the other hand. 

In this work, we consider \textbf{sparse graphs} with constant average degree, a regime where one cannot recover all node correspondences. We focus on the notion of \textbf{one-sided partial alignment} where the goal is to recover a significant portion of node matchings while making a vanishing fraction of errors.\\
A detailed discussion of the sparse setting is provided in Section~\ref{section:correlated_graph_model}, along with a formal description of the random graph model. For this introductory overview, we give the following definitions:
\begin{itemize}
	\item \textit{Erdős--Rényi model.}  A random graph $G \sim \mathrm{ER}(n, p)$ consists of $n$ nodes where each edge between distinct vertices exists independently with probability $p$. 
	\item \textit{Parameters.} The correlated asymmetric model is parameterized by a global density level $\lambda>0$ and correlation parameters $s,s'\in[0,1]$.
	\item \textit{Correlated asymmetric model (informal description).} Given an intersection graph $G_* \sim \mathrm{ER}(n_*, \lambda s s'/n_*)$, we add edges and nodes twice independently: Once to obtain $G \sim  \mathrm{ER}(n, \lambda s/n)$ and once to get $G' \sim \mathrm{ER}(n', \lambda s'/n')$. The vertex counts $n, n'$ are coupled through $n_*$, while the graphs $G, G'$ are correlated through $G_*$.
\end{itemize} 
This enables an informal presentation of our main results. 

\begin{theorem}[main results --- informal statement]\label{thm_informal}
	For a pair of sparse  asymmetric correlated Erdős--Rényi graphs $(G, G')$ with parameters $\lambda > 0$ and $s, s' \in [0,1]$ the following statements hold:
	\begin{enumerate}[(i)]
		\item The polynomial time algorithm MPAlign achieves one-sided partial alignment under the condition that asymmetric tree correlation testing is feasible (see Theorem \ref{thm:Algo_works}).
		\item There is a sharp phase transition for the feasibility of asymmetric tree correlation testing around $ \alpha \approx 0.338$ known as Otter's constant: Testing is possible if $ss' > \alpha$ and $ \lambda > \lambda_0(s, s')$; it is impossible if $ss' \leq \alpha$ (see Theorem \ref{thm:phase_transition}).
	\end{enumerate}
\end{theorem}
Taking $s < 1$ and $s' = 1$, our model allows $G$ to be a random subgraph of $G'$. To our knowledge, we contribute the first polynomial-time algorithm to solve this random instance of the SIP. The question about isomorphisms of random graphs is motivated by the work of Babai and Bollobás \cite{babai_random_1980, bollobas_random_2008} on efficient \textit{canonical labeling algorithms} that solve (GIP). Extending these results to \textbf{random subgraph isomorphisms is an open problem} raised in personal communication with Jiaming Xu in 2017 \cite{xu2017personal}. We provide an answer to one instance of that problem.

\subsection{Related work}

\paragraph{Applications.} Research on graph alignment is motivated by a diverse set of applications \cite{conte2004thirty}, and has gained popularity through the success in network de-anonymization and consequent privacy breaches during the late 2000s \cite{narayanan_robust_2008, narayanan2009anonymising}. Researchers in various disciplines have developed graph alignment algorithms for their specific use, such as aligning protein interaction graphs \cite{singh_global_2008, aladaug2013spinal}, enabling computer vision \cite{berg_shape_2005, scholkopf_balanced_2007}, natural language processing \cite{haghighi2005robust}, and cell biology \cite{chen2022one}.

\paragraph{Graph alignment theory.} The rigorous theoretical study of graph alignment algorithms was spearheaded by the introduction of correlated random graphs in \cite{pedarsani2011privacy}. The correlated Erdős--Rényi model has since then been thoroughly explored, essentially from two points of view: Information theoretically - \textit{under which parameters is alignment feasible?} \cite{cullina2016improved, cullina2017exact, mao2023exact, wu2022settling}, and computationally 
- \textit{which polynomial algorithm can retrieve an alignment?} \cite{zaslavskiy2008path, yartseva2013performance, ding2021efficient, mao2023random, fan2023spectral, araya2024graph}. All of these papers have in common that they consider either the \textit{dense regime} or the \textit{intermediate regime} where the average node degree across $G$ and $G'$ is growing with the number of nodes. In a notable recent advancement, \cite{ding2023polynomial} presents a polynomial-time alignment algorithm for the case with an average degree barely above constant. 

\paragraph{Asymmetry.} We present one of the first theoretical results on asymmetric graph alignment, but this question has sparked interest in related fields. One example comes from the optimal transport community where \textit{Unbalanced Gromov-Wasserstein alignment} allows to compare networks of different sizes \cite{thual_aligning_2022, sejourne2021unbalanced}. Another example comes from random graph theory and probabilistic subgraph matching \cite{schellewald2005probabilistic, brocheler2011probabilistic}. Since our work allows for $G$ to be a subgraph of $G'$, our results enable the matching of such graphs.

To our knowledge, only the authors of \cite{huang2024information} conduct a thorough theoretical analysis of graph alignment in the setting with differing node numbers. Their notion of \textit{partially correlated Erdős--Rényi graphs} resembles our correlated graph model but leaves out varying edge densities, and no polynomial matching algorithm is provided. 

\paragraph{Sparsity.} A branch of research spearheaded by \cite{cullina2020partial, hall2023partial, ding2023matching} has focused on the \textit{sparse regime}, where the average node degree remains constant. In this setting, Erdős--Rényi graphs are disconnected, making graph matching impossible outside the giant component \cite{ganassali2021impossibility}. However, the local tree-like structure of such graphs has led to intriguing questions about tree correlation testing for graph alignment \cite{ganassali2020tree}, which were largely resolved in \cite{luca, ganassali2022statistical}. Complementarily, \cite{mao2023random} proposes an algorithm which uses particular rooted-tree substructures called \textit{chandeliers}. This approach yields zero-mismatch partial alignment guarantees around Otter's constant in the sparse regime, while also being applicable to denser regimes.

\subsection{Main contributions}
We introduce the novel asymmetric correlated Erdős–Rényi model and extend the main results of \cite{luca} and \cite{ganassali2022statistical} to that setting. This extension is non-trivial and yields a substantially stronger result. Notably, the phase transition in (ii) of Theorem \ref{thm_informal} depends only on the \emph{product} $ss'$. This simplicity enables a broader interpretation: the density of one graph can compensate for the sparsity of the other, a regime strictly harder than those considered previously. Moreover, local methods, as used here, naturally accommodate asymmetry. This is unlike spectral or first-order optimization approaches, which typically require both graphs to have the same number of nodes.

The non-triviality of this generalization is already apparent in the likelihood ratio diagonalization formula (Theorem \ref{thm_diagonalization}), which differs meaningfully from \cite[Theorem 4]{ganassali2022statistical}: In our case, the functions $f$ are parametrized by $(\lambda s)$ and $(\lambda s')$, not just $\lambda$. This asymmetry propagates through the subsequent proofs, requiring considerable care and new work. We therefore make explicit many technical details omitted in \cite{luca} and \cite{ganassali2022statistical} and adjust the presentation for greater transparency. While these generalizations may seem natural in hindsight, executing them rigorously is a main contribution.

To further illustrate how this work goes beyond a direct generalization of \cite{luca} and \cite{ganassali2022statistical}, we list some major novelties:
\begin{enumerate}[label=(\alph*)]
    \item Theorem \ref{thm:Algo_works} uses concentration of node count measures, distinguishes be\-tween matchable and unmatchable nodes, and adjusts the ``dangling tree'' argument (see Lemma \ref{lemma:danglingTree}). None of these are necessary in the symmetric setting.
    \item Lemma \ref{lemma:disj_vs_ind} replaces the 9-line argument in Appendix C.3 of \cite{ganassali2022statistical} with a 7-page proof using graph exploration techniques, correcting an oversight in the earlier treatment of joint graph laws.
    \item Theorem \ref{thm_diagonalization} and Lemma \ref{gaussian_approx} carefully handle convergence of countable sums over unlabeled trees --- a crucial subtlety in the asymmetric setting.
\end{enumerate}
Preliminary results appeared in the workshop paper \cite{maier2023asymmetric}; this document presents a complete and significantly extended version.

\subsection{Organization of the paper}
Section \ref{section:setting_and_random_graph_model} introduces our asymmetric correlated Erd\H{o}s--R\'enyi model in the sparse\\ regime and formalizes the recovery target, namely one-sided partial alignment.\\
Section~\ref{section:From_graphs_to_trees_and_back} then reduces graph alignment to a local comparison of neighborhood trees and presents our polynomial-time algorithm \textsc{MPAlign}, establishing algorithmic guarantees in Theorem~\ref{thm:Algo_works}. This corresponds to point (i) of Theorem \ref{thm_informal}. \\
Finally, Section~\ref{section:tree_correl_testing} studies the underlying \emph{asymmetric tree correlation testing} problem: It derives the feasibility condition, proves the sharp phase transition stated in Theorem~\ref{thm:phase_transition}, and thereby identifies the parameter regime in which \textsc{MPAlign} succeeds. This corresponds to point (ii) of Theorem \ref{thm_informal}.\\
Section~\ref{section:conclusion} concludes and the Appendix \ref{appendix:main} collects deferred proofs.

\subsection{Notations}

Let $\N = \{0,1,2,3, ...\}$ and write \( \mathbb{N}_{>0}  = \N \setminus \{0\}\). For \( k \in \mathbb{N}_{>0} \),  define
\(
[k] := \{1, \dots, k\}
\)
and denote the sum over all injective mappings from \( [k] \) to \( [n] \) as
\(
\sum_{\sigma: [k] \hookrightarrow [n]}
\).\\
For \( a, b \in \mathbb{R} \), let \( a \wedge b := \min(a, b) \). The complement of a set $S$ is denoted $S^c$.\\
For a graph \( G \), denote its vertex set \( V(G) \) and write  \( i \in G \) as a shorthand for  \( i \in V(G) \). Summing over the nodes of a tree \( t \) is denoted
\(
\sum_{i \in t}
\). The distance between \( i, j \in G \) is 
\[
\mathrm{dist}_G(i, j) := \inf \big\{ \mathrm{length}(p)\,  \vert \, p \text{ is a path in } G \text{ connecting } i \text{ and } j \big\}.
\]
Given two probability measures \( \Prob,  \mathbb{Q} \) on a probablity space $(\Omega, \mathcal{F})$, denote their product measure \( \Prob \otimes \mathbb{Q} \) and their total variation distance
\(
d_{\text{TV}}(\Prob,  \mathbb{Q}) = \sup_{A \in \mathcal{F}} |\Prob(A) - \mathbb{Q}(A)|
\).\\
For \(n\in\mathbb N\), \(p\in[0,1]\), \(\mu>0\), let \(\mathrm{Bin}(n,p)\) and \(\mathrm{Poi}(\mu)\) denote the binomial and Poisson distributions. If \( B \sim \mathrm{Bin}(n,p) , X \sim \mathrm{Poi}(\mu) \), write $\pi_\mu(k):= \mathbb{P}(X = k)$ and recall
\[
 \pi_\mu(k) = e^{-\mu} \frac{\mu^k}{k!},  \quad  \mathbb{P}(B = k) = \binom{n}{k} p^k (1-p)^{n-k}, \quad \text{for all } k \in \N.
\]
To describe asymptotic behaviour of expressions $f,g:\N\to\R\setminus\{0\}$, we use
\[
f(N)\in \left\{
\begin{aligned}
	\mathcal{O}(g(N))\\
	\Theta(g(N))\\
	o(g(N))
\end{aligned}
\right\} \iff \left\{
 \begin{aligned}
&f(N)/g(N) \text{ is bounded above},\\
&f(N)/g(N) \text{ is bounded above and below},\\
&\lim\nolimits_{N\to \infty} f(N)/g(N) = 0.
\end{aligned}
\right.
\]

\section{Setting and random graph model}\label{section:setting_and_random_graph_model}
In this section, we introduce a novel model of correlated random graphs that allows for varying edge densities and different numbers of nodes. We subsequently introduce the \emph{overlap}, which quantifies agreement with the planted alignment, and the notion of \emph{one-sided alignment}, which captures our recovery objective.

\subsection{Asymmetric correlated Erdős--Rényi graphs}\label{section:correlated_graph_model}

The simplest way to describe the law of two asymmetric correlated Erdős--Rényi graphs $(G, G')$ is to algorithmically detail their sampling process.

For this, fix a node number $N \in \N$ and a density parameter $\lambda > 0$. Furthermore, define node correlation parameters $q, q' \in [0,1]$ and edge correlation parameters $r, r' \in (0,1]$. Two correlated random graphs $G, G'$ are sampled as follows:
\begin{enumerate}
	\item Let $n_* \sim \mathrm{Bin}(N, qq')$ and sample $G_* \sim \mathrm{ER}\Big(n_*, rr' \frac{\lambda}{N}\Big)$ which we refer to as the \textit{intersection graph}. 
	\item Take $n_+ \sim \mathrm{Bin}(N, q(1-q'))$ and add $n_+$ isolated nodes to $G_*$ which now has a total of $n := n_* + n_+$ nodes. Then, independently for each newly added node $v$ and every already present node $w \in [n]\setminus\{v\}$, add an edge connecting $v$ and $w$ with probability $rr'\lambda/N$. This process is called \textit{node augmentation}.
	\item Next, independently for each pair of unconnected nodes $v \neq w \in [n]$, connect them with probability $r(1-r')\lambda/N$. This process is called \textit{edge augmentation}. Call $G$ the graph resulting from steps 2 and 3.
	\item Repeat steps 2 and 3 independently with parameters $n'_+ \sim \mathrm{Bin}(N, (1-q)q')$ and edge augmentation probability $(1-r)r'\lambda/N$. Let $n' := n_* + n'_+$ and call the resulting graph $G'$. 
\end{enumerate}
All four steps are illustrated with an example in figure \ref{model_figure}.
\begin{figure}[t]
	\includegraphics[width=\textwidth]{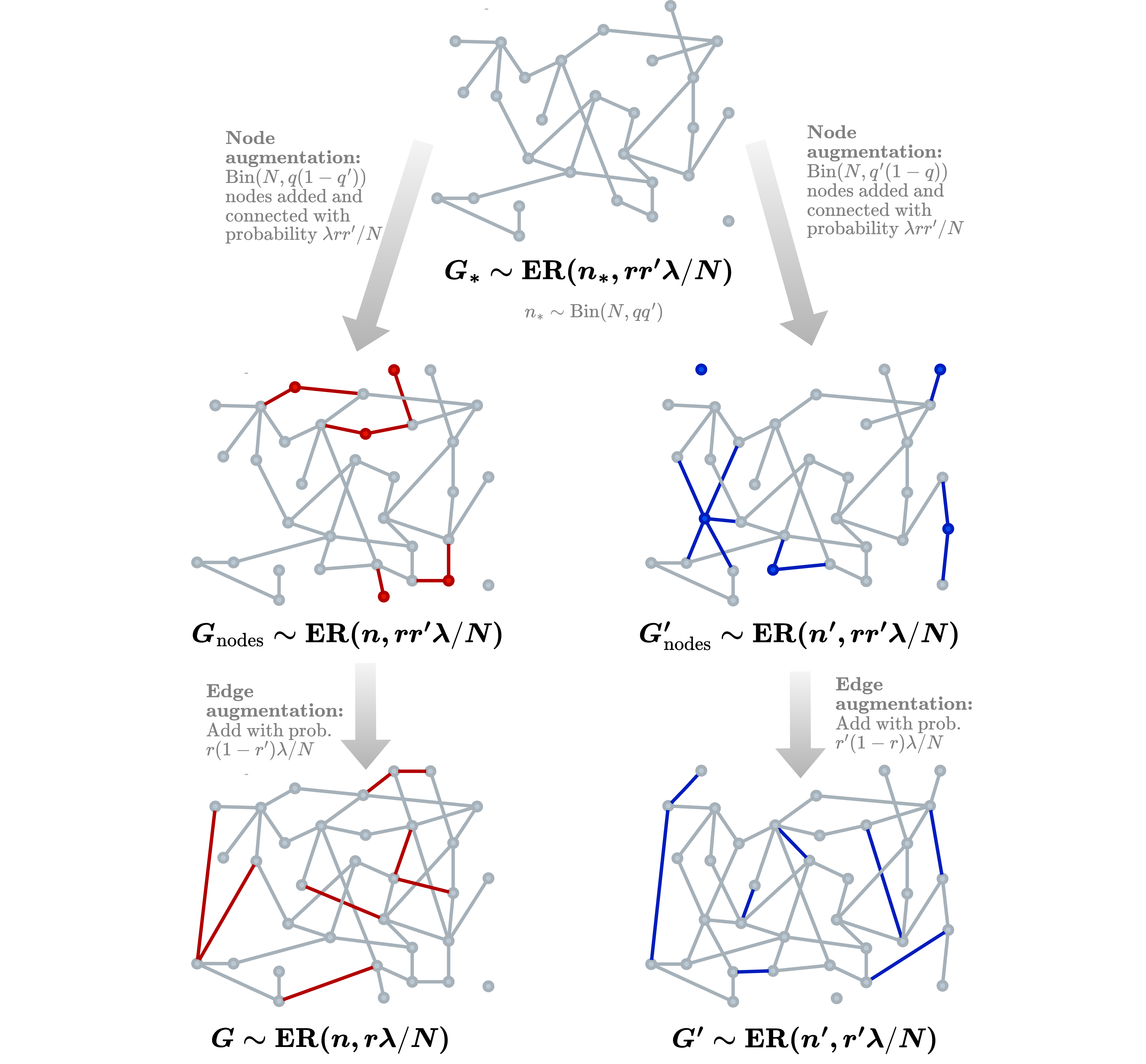}
	\caption{Sampling process of $(G, G') \sim \mathrm{CER}(N, \lambda, q, q', r, r')$. Nodes and edges added during the augmentations are red and blue, respectively.}
	\label{model_figure}
\end{figure}
This sampling process yields two asymmetric correlated Erdős--Rényi graphs $G$ and $G'$ whose joint law shall be denoted as $\mathrm{CER}(N, \lambda, q, q', r, r')$. We define their global correlation parameters as \[
s := qr \quad \text{and} \quad s':= q'r',
\]
which lets us write $(G, G') \sim \mathrm{CER}(N, \lambda, s, s')$. This notation hides how precisely the graphs are correlated, but we will later see that $(s,s')$ alone determines whether our alignment algorithm works, ignoring the precise value of $(q, q', r, r')$.

Since $\Expe[n] = qN$ and $\Expe[n'] = q'N$, the marginal distributions of both graphs are
\begin{equation}\label{marginals}
	G \mathrel{\sim} \mathrm{ER}\Big(n,  \frac{\lambda s}{\Expe[n]}\Big) \quad \text{and} \quad G' \mathrel{\sim} \mathrm{ER}\Big(n',  \frac{\lambda s'}{\Expe[n']}\Big)
\end{equation}
with expected average degrees $\Expe[n \lambda s/ \Expe[n] ] = \lambda s$ and $\Expe[n' \lambda s'/ \Expe[n'] ] = \lambda s'$. Since these average degrees are constant, $G$ and $G'$ are called \textit{sparse} Erdős--Rényi graphs: Their adjacency matrices contain a constant number of entries in each column.
\begin{remark}
	When $q = q' = 1$, both graphs have the same deterministic node number $n = n' = N$. Therefore, the only difference between $G$ and $G'$ are the average degrees $\lambda s$ and $\lambda s'$. This is what \cite{maier2023asymmetric} refer to as the \textit{varying edge densities} setting.
	
	In the case where $r = r' = 1$, there may be different node numbers in $G$ and $G'$, but each pair of nodes in both graphs is connected with the same probability $ \lambda r/N $. However, unless $q = q'$, their expected average degrees are different, equal to $\lambda q$ and $\lambda q'$, respectively. This is what  \cite{maier2023asymmetric} refer to as the \textit{differing node numbers} setting.
	
	Our definition of $\mathrm{CER}(N, \lambda, q, q', r, r')$ therefore combines both models from \cite{maier2023asymmetric} into one comprehensive asymmetric Erdős--Rényi model with differing node numbers and varying edge densities at the same time.
\end{remark}
\begin{remark}
	There are many ways of defining $ \mathrm{CER}(N, \lambda, q, q', r, r')$; here, we have chosen to start with an \textit{intersection graph} $ G_* $ before performing node and edge \textit{augmentations}. Note that an identically distributed pair of random graphs could equivalently be obtained by starting with a \textit{master graph} $\hat G \sim \mathrm{ER}(N, \lambda/N)$, followed by independent node and edge \textit{deletions}. We refer to \cite{maier2023asymmetric} for details on the deletion procedure; one can easily verify that both algorithms sample from the same distribution.
\end{remark}

\subsection{True node correspondences}\label{section:True_node_correspondences}

The planted graph alignment problem consists of retrieving correspondences between the nodes of two graphs. Given a pair $(G, G') \sim  \mathrm{CER}(N, \lambda, q, q', r, r')$ from our generative model, there exists a hidden partial mapping between the nodes of $G$ and $G'$, namely via the nodes inherited from $G_*$ that both graphs have in common. 

To formalize this, we randomly label the nodes of both $G$ and $G'$ with labels in $[n]$ and $[n']$ respectively (for instance, by choosing the node with number 1 uniformly at random, then picking a second node among the remaining ones, etc.). Then, there are two node subsets $V_* \subset [n]$, $V'_* \subset [n']$ of size $n_*$ and a bijective map $\sigma_*: V_* \to V'_*$ which relates all pairs of nodes shared by $G$ and $G'$. 
\begin{remark}
	In the setting $q = q' = 1$ where both graphs have $n$ nodes, one has $V_* = V'_* = [n]$ and $\sigma_*$ is a permutation of $[n]$. If we use the same node labeling for $G_*$ and $G$, ``forgetting'' these labels would formally consist of sampling $\sigma_*$ uniformly among all permutations of $[n]$ and then labeling the nodes of $G'$ as $\sigma_*(1), \sigma_*(2)$ and so on.
\end{remark}
\begin{remark}\label{remark_onlyedges}
	A key novelty of the  $ \mathrm{CER}(N, \lambda, q, q', r, r') $ model is step 2 in the sampling process, which allows for different node numbers in both graphs. Omitting this step by picking $q = q' = 1$, one recovers the correlated Erdős--Rényi model typically studied in the literature. Due to the deterministic node numbers in this setting, one can more easily define $G$ and $G'$ via their adjacency matrices $A$ and $A'$: Their correctly permuted entry pairs $(A_{ij}, A'_{
		\sigma_*(i) \sigma_*(j)})$ are independent tuples of Bernoulli variables with joint distribution
	\begin{equation*}
		(A_{ij}, A'_{\sigma_*(i) \sigma_*(j)}) = \begin{cases}
			(1, 1) \quad&\text{ with probability } (\lambda/n) s s'\\
			(1, 0) \quad&\text{ with probability } (\lambda/n) s (1-s')\\
			(0, 1) \quad&\text{ with probability } (\lambda/n) s' (1-s)\\
			(0, 0) \quad&\text{ with probability } 1 - (\lambda/n)(s + s' - ss').
		\end{cases}
	\end{equation*}
	For $q=q'=1$, if one additionally supposes $r = r'$,  we recover the ``standard'' correlated Erdős--Rényi model first introduced in \cite{pedarsani2011privacy} and studied extensively, for example in \cite{luca}. In the latter paper and most literature, the parameter $\lambda$ corresponds to $\lambda/r$ in our model.
\end{remark}

One speaks of \textit{planted graph alignment} since the objective is to recover $\sigma_*$, which, by construction, is metaphorically planted into the model, together with the intersection graph $G_*$. This objective requires us to define an alignment estimator consisting of two sets $\hat V \subset [n]$, $\hat V' \subset [n']$, which ideally contain the nodes of $ G_* $, together with a bijective map $ \hat \sigma: \hat V \to \hat V'$ matching nodes in $G$ with nodes in $G'$. Such an estimator, therefore, consists of the triplet $ (\hat \sigma, \hat V, \hat V') $, but we will often just write $\hat \sigma$ to denote all three objects.

We introduce the following performance indicator to measure if such a $\hat \sigma$ is close to the true matching $\sigma_*$.
\begin{definition}[overlap]
	For an alignment estimator $\hat \sigma : \hat V \to \hat V'$, define its overlap with the true matching $\sigma_*: V_* \to V'_*$ as
	\[
	\mathrm{ov}(\hat \sigma, \sigma_*) := \frac{1}{n_*} \sum_{i \in \hat V \cap V_*} \indicator{\hat \sigma(i) = \sigma_*(i)}.
	\]
\end{definition}
Since $n_* = |V_*|$, the overlap measures which proportion of the relevant nodes is correctly matched by $\hat{\sigma}$. This tells us how well $\hat \sigma$ acts on the true node set but is entirely agnostic to what the estimator does outside of $V_*$. Consequently, another quantity is required to measure how well $\hat \sigma$ identifies $V_*$ and how many errors it commits in a matching.
\begin{definition}[error fraction]
	For an alignment estimator $\hat \sigma: \hat V \to \hat V'$ , one defines its error fraction as
	\[
	\mathrm{err}(\hat \sigma, \sigma_*) := \frac{1}{n} \sum_{i \in \hat V} \indicator{ i \not \in V_* \text{ or } \hat \sigma(i) \neq \sigma_*(i)}.
	\]
\end{definition}
The goal of planted graph alignment is to find a $\hat \sigma$ that can be computed in polynomial time and maximizes the overlap while keeping the error fraction low. For our specific case, we formalize this goal as follows.

\subsection{One-sided partial alignment}
When seeking to align a pair $(G, G')$ of sparse correlated Erdős--Rényi graphs, it is impossible to recover $\sigma_*$ completely. The reason for this lies within the structure of $G_*$, the graph containing all the information about a potential matching between $G$ and $G'$. Since $G_* \sim \mathrm{ER}(n_*, ss'\lambda/\Expe[n_*])$, this is a sparse Erdős--Rényi graph with one giant connected component while the remaining graph mainly consists of small trees \cite{bollobas1998random}. This property has been exploited in \cite{ganassali2021impossibility} to prove that a matching between $G$ and $G'$, cannot be done outside this giant component. A more intuitive reason for the impossibility of retrieving $\sigma_*$ is isolated nodes: Since average node degrees are constant, there is a non-negligible amount of isolated nodes in $G_*$, which are all indistinguishable.

In the sparse setting, the giant component of $G_*$ contains strictly fewer nodes than the entire graph. Therefore, one can only hope to partially align $G$ with $G'$. A lower number of matchable nodes also means that every wrongly matched node significantly augments the relative error. This is why authors \cite{luca, cullina2020partial} have focused on alignment estimators that commit a \textit{negligible amount of errors} while \textit{recovering a significant proportion} of the possible graph matching. This is referred to as one-sided partial alignment:
\begin{definition}\label{def:one-sided_alignment}
	For a family of correlated graphs $(G_N, G'_N) \sim \mathrm{CER}(N, \lambda, s, s')$, a sequence of estimators $\hat \sigma_N : \hat{V}_N \to \hat{V}'_N$ obtained through the same algorithm is said to achieve one-sided partial alignment if there exists $\beta > 0$ such that
	\begin{equation*}
		\Prob\Big(\mathrm{ov}(\sigma_*, \hat{\sigma}_N) \geq \beta \Big) \xrightarrow[N\to\infty]{} 1 \quad \text{and}\quad \Prob\Big(\mathrm{err}(\sigma_*, \hat{\sigma}_N) = o(1) \Big) \xrightarrow[N\to\infty]{} 1.
	\end{equation*}
	We will often omit indexing by $N$ and write $G, G', \hat{\sigma}$ instead of $G_N, G'_N, \hat{\sigma}_N$. 
\end{definition}
\begin{remark}\label{equivalence:vanishesWHP_convergesINP}
	This definition translates the phrase \textit{``$\hat{\sigma}$ commits a negligible amount of errors''}  into the more precise \textit{``the error fraction of $ \hat\sigma $ is vanishing with high probability''} which in mathematical terms reads
	``$\Prob\big(\mathrm{err}(\sigma_*, \hat{\sigma}) = o(1) \big) \to 1$''. More rigorously, this expression signifies that there is a  sequence $(a_N)_N$ of $\R_{> 0}$ with $a_N \to 0$ and
	\[
	\Prob\big(\mathrm{err}(\sigma_*, \hat{\sigma}_N) \leq a_N \big) \xrightarrow[N\to\infty]{} 1 \quad \iff \quad \Prob\big(\mathrm{err}(\sigma_*, \hat{\sigma}_N) > a_N \big) \xrightarrow[N\to\infty]{} 0. 
	\]
	We claim that this is equivalent to $ \mathrm{err}(\sigma_*, \hat{\sigma}_N) \to 0 $ in probability. To prove the claim, let $\varepsilon > 0$ and take $N$ big enough so that $a_K < \varepsilon$ for all $K > N$. Then, for all such $K$, 
	\[
	\Prob\big(\mathrm{err}(\sigma_*, \hat{\sigma}_K) > \varepsilon \big) \leq \Prob\big(\mathrm{err}(\sigma_*, \hat{\sigma}_K) > a_K \big) \xrightarrow[K \to \infty]{} 0.
	\]
	On the other hand, if $ \mathrm{err}(\sigma_*, \hat{\sigma}_N)  $ converges in probability, one can take the constant sequence $a_N \equiv \varepsilon$ to obtain the opposite direction immediately. This lets us conclude that
	\[
	\Prob\Big(\mathrm{err}(\sigma_*, \hat{\sigma}_N) = o(1) \Big) \xrightarrow[N\to\infty]{} 1 \quad \iff \quad \mathrm{err}(\sigma_*, \hat{\sigma}_N) \xrightarrow[N \to \infty]{\text{in } \Prob} 0.
	\]
\end{remark}
To summarize, sparse Erdős–Rényi graphs admit at best partial alignment, and Definition \ref{def:one-sided_alignment} formalizes what we mean by a successful alignment estimator. In the next section, we show that such graphs are locally tree-like, and we leverage this structure to design an efficient alignment algorithm.

\section{From graphs to trees and back}\label{section:From_graphs_to_trees_and_back}
Local algorithms are a natural choice for aligning sparse Erdős--Rényi graphs, an intuitive reason being that these graphs have large diameters, and far-apart nodes contain negligible information about each other. Therefore, it instinctively makes sense to align nodes based on their local neighborhoods in the respective graphs. 

In this section, we describe the asymptotic structure of these neighborhoods, which are Galton--Watson trees. This will allow us to translate the global graph alignment problem to a local tree correlation test. While tree correlation testing deserves a whole section later, we will leave it as a black box here and focus on the algorithm that bridges the gap between graphs and trees, between global structure and local neighborhoods.

\paragraph{Notation for trees.}  Given $\mu > 0$, the law of a Galton--Watson tree \cite{watson1875probability} with $\mathrm{Poi}(\mu)$ offspring distribution is denoted as $\mathrm{GW}^{(\mu)}$. Next, for any rooted tree $t$ of potentially infinite size, we denote $t_d$ the \textit{pruning} of $t$ at depth $d \in \N$, i.e., the subtree containing all nodes whose distance to the root is at most $d$.\\
For $t \sim \mathrm{GW}^{(\mu)}$, we denote by $\mathrm{GW}_d^{(\mu)}$ the law of the pruned tree $t_d$. Furthermore, if we consider any fixed tree $\tau$, we denote by $\mathrm{GW}_d^{(\mu)}(\tau)$ the likelihood that the pruning $\tau_d$ is of law $ \mathrm{GW}_d^{(\mu)} $. Hence, the quantity $\mathrm{GW}_d^{(\mu)}(\tau)$ depends on $\tau$ only up to depth $d$ while ignoring any nodes which are further from the root. 

\subsection{Tree neighborhoods} 
As a first step towards understanding neighborhoods in sparse correlated Erdős--Rényi graphs, we describe the asymptotic local structure of a single graph $G \sim \mathrm{ER}(N, \mu/N)$.
For a node $i \in G$ and  $d \in \N_{>0}$, we denote the neighborhood of $i$ up to depth $d$ as
\[
\mathcal{N}_{G, d}(i) = \{j \in V \, : \, \mathrm{dist}_G(i,j) \leq d\}.
\]
While $ \mathcal{N}_{G, d}(i)$ can be seen as a ball of radius $d$ around $i$, we denote the sphere as
\[
\mathcal{S}_{G, d}(i) = \mathcal{N}_{G, d}(i) \setminus \mathcal{N}_{G, d-1}(i) = \{j \in V \, : \, \mathrm{dist}_G(i,j) = d\}.
\]
The following three lemmata, first presented as lemmata 2.1, 2.3, and 2.4 in \cite{ganassali2020tree}, describe the behaviors of these neighborhoods.
\begin{lemma}[Uniform bound on the neighborhood sizes]\label{lem_neighborhoodsize}
	Let $G = G_N \sim \mathrm{ER}(N, \mu/N)$ be a family of Erdős--Rényi graphs with $\mu > 1$.  Set $  d= d_N := \lfloor c \log(N)\rfloor$ for some $c$ fulfilling $c \, \log(\mu) < 1/2$.
	Then, for all $\varepsilon > 0$ there exists a constant $ C = C(\varepsilon) > 0$ such that
	\[
	\Prob\Big( \exists \,  i \in V,  t \in [d]  : \, | \mathcal{S}_{G, t}(i) | > C \log(N) \mu^t \Big) \in  \mathcal{O}(N^{-\varepsilon}).
	\]
\end{lemma}
\begin{lemma}[Different neighborhoods are asymptotically independent]\label{lem_independence_of_different_neighborhoods}
	Let $G, d, \mu$ and $c$ be like in Lemma \ref{lem_neighborhoodsize}.  Then there exists $\varepsilon > 0$ such that for any pair of nodes $i \neq j \in V$, one has
	\begin{equation*}
		d_{\mathrm{TV}} \Big( \mathrm{Law}\big(\mathcal{N}_{G,d}(i), \mathcal{N}_{G,d}(j) \big), \mathrm{Law}\big(\mathcal{N}_{G,d}(i)\big) \otimes  \mathrm{Law}\big(\mathcal{N}_{G,d}(j)\big) \Big) \in \mathcal{O}(N^{-\varepsilon}).
	\end{equation*}
	Consequently, there exists an optimal coupling (see, for instance, \cite[Theorem 4.1]{villani2009optimal}) between $ \big(\mathcal{N}_{G,d}(i), \mathcal{N}_{G,d}(j) \big) $ and $(\mathcal{N}, \mathcal{N}') \sim \mathrm{Law}\big(\mathcal{N}_{G,d}(i)\big) \otimes  \mathrm{Law}\big(\mathcal{N}_{G,d}(j)\big)$ fulfilling
	\[
	\Prob\Big(\, \big(\mathcal{N}_{G,d}(i), \mathcal{N}_{G,d}(j) \big) \neq (\mathcal{N}, \mathcal{N}') \Big)\in o(1).
	\]
\end{lemma}
\begin{lemma}[Neighborhoods look like Galton--Watson trees]\label{lem_likeGWtrees}
	Let $G, d, \mu$ and $c$ be like in Lemma \ref{lem_neighborhoodsize}. Then there exists $\varepsilon > 0$ such that for all nodes $i \in V$, 
	\[
	d_{\mathrm{TV}}\Big(\mathcal{N}_{G,d}(i), \, \mathrm{GW}^{(\mu)}_d \Big) \in \mathcal{O}(N^{-\varepsilon}).
	\]
\end{lemma}
Asymptotically for $N \to \infty$, Lemma \ref{lem_neighborhoodsize} allows to control the number of nodes in neighborhoods of size $\mathcal{O}(\log(N))$. Furthermore, Lemma \ref{lem_independence_of_different_neighborhoods} states that the neighborhoods of two different fixed nodes behave independently. Finally, Lemma \ref{lem_likeGWtrees} affirms that these neighborhoods look as if they were sampled like a Galton--Watson tree with Poisson offspring.
\begin{remark}
	Consider the more general case of a graph $ G \sim \mathrm{ER}(n, p) $ with random node number $n \sim \mathrm{Bin}(N, q)$ and edge probability $p = \lambda / \Expe[n]$. Since with high probability, $np = n \lambda / \Expe[n] \rightarrow \lambda$, one can easily adjust the proofs in \cite{ganassali2020tree} for the lemmata to remain valid in this random node number case. \\
	We later apply these modified versions of the lemmata to $(G, G') \sim \mathrm{CER}(N, \lambda, s, s')$. To do so, recall from (\ref{marginals}) that the marginal laws are $G \sim \mathrm{ER}(n, \lambda s /\Expe[n])$ and $G' \sim \mathrm{ER}(n', \lambda s' /\Expe[n'])$. The depth--$d$--neighborhoods of all nodes $i$ in $G$ will therefore converge towards $\mathrm{GW}^{(\lambda s)}_d$ and those of $j \in G'$ towards $\mathrm{GW}^{(\lambda s')}_d$.\\
	The same reasoning holds for $G_* \sim \mathrm{ER}(n_*, \lambda rr'qq'/(Nqq')) = \mathrm{ER}(n_*, \lambda ss'/\Expe[n_*])$ whose node neighborhoods asymptotically behave as $\mathrm{GW}^{(\lambda s s')}_d$.
\end{remark}

\subsection{Correlated Galton Watson trees}\label{section_correlated_GW_trees}
From the way we have defined correlated sparse Erdős--Rényi graphs in section \ref{section:correlated_graph_model}, it is straightforward to transpose the results about single graph neighborhoods to joint neighborhoods of graph couples. Recall that all information about the correlation of a pair $(G, G') \sim \mathrm{CER}(N, \lambda, s, s')$ is contained in the intersection graph $G_*$. Translating this to the local level, tree-like neighborhoods of the same node in $G$ and $G'$ are correlated via an \textit{intersection tree} corresponding to the neighborhood in $G_*$. This leads us to introduce the following model of correlated Galton--Watson trees.

\begin{definition}[Correlated Galton--Watson trees]\label{def:corr_GW_trees}
	For $\lambda > 0$ and $s, s' \in [0,1]$, define the law of correlated Galton--Watson trees $(t, t') \sim\Pcorr$ via the following sampling algorithm:
	\begin{enumerate}
		\item Sample a rooted tree $t_* \sim \mathrm{GW}^{(\lambda s s')}$ which we refer to as the intersection tree.
		\item Independently for each node $i$ of $t_*$, sample $c^+(i) \sim \mathrm{Poi}(\lambda s (1-s'))$ and attach $c^+(i)$ additional child nodes to $i$. 
		\item Then, for each child node $j$ added in step 2, independently sample a tree $t^+(j) \sim \mathrm{GW}^{(\lambda s)}$ and attach it, making $j$ its root.
		\item Call the resulting tree $t$. Then repeat steps 2 and 3 independently with a copy of $t_*$, but change the laws $ c^+(i) \sim \mathrm{Poi}(\lambda s' (1-s)) , \, \,   t^+(j) \sim \mathrm{GW}^{(\lambda s')} $ and call the result $t'$.
	\end{enumerate}
	Figure \ref{tree_sampling} illustrates this sampling process. While the procedure yields a pair of potentially infinitely deep trees $ (t, t') $, we will often consider their depth--$d$--truncated versions $(t_d, t'_d)$ whose law we denote as $\Pcorr_d$. 
\end{definition}
\begin{figure}[t]
	\includegraphics[width=0.7\textwidth]{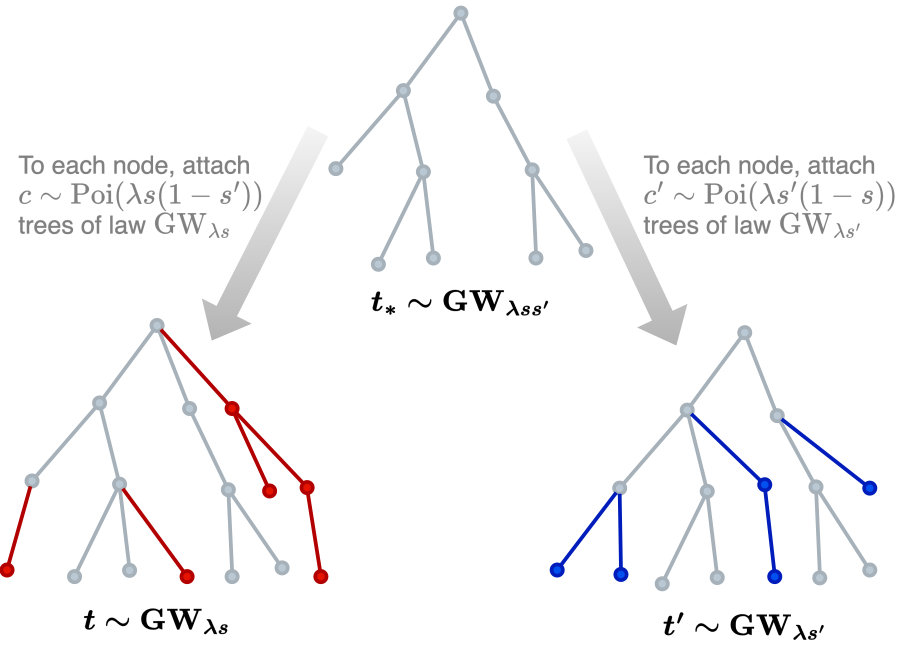}
	\caption{Sampling process of $  (t, t') \sim\Pcorr $ up to depth 3. The blue and red components are independently sampled and attached to the grey intersection tree.}
	\label{tree_sampling}
\end{figure}
For a pair of trees $(t, t') \sim\Pcorr_d$, independently assign random labels to their nodes, like we do for a pair of graphs in section \ref{section:True_node_correspondences}. The correlation between $t$ and $t'$ then relies on the embedded intersection tree $t_*$. It can be expressed via a mapping $\sigma_* : V_* \to V_*'$ where $V_*$ and $ V_*' $ are the node subsets of $t$ and $t'$ corresponding to the nodes of $t_*$.

The trees from Definition \ref{def:corr_GW_trees} are intimately related to sparse correlated Erdős--Rényi graphs. This is formalized in the following lemma, which is an adaptation of \cite[Lemma 6.4]{luca}, and its proof follows directly from Lemma \ref{lem_likeGWtrees}.
\begin{lemma}\label{lemma:correlated_graphs_converge_towards_P1}
	Given $(G, G') \sim \mathrm{CER}(N, \lambda, s, s')$, let $ (i, \sigma_*(i)) \in V_* \times V'_*$ be a pair of aligned nodes in the sense of section \ref{section:True_node_correspondences}. If furthermore $ d = \lfloor c \log(N) \rfloor $ for $c$ fulfilling $2 c \log( \lambda \max\{s, s'\}) < 1$, then there exists $\varepsilon > 0$ such that
	\begin{equation*}
		d_{\mathrm{TV}} \Big(\, \Big(\mathcal{N}_{G,d}(i), \mathcal{N}_{G',d}(\sigma_*(i)) \Big), \, \Pcorr_d \, \Big) \in \mathcal{O}(N^{-\varepsilon}).
	\end{equation*}
	Hence, there exists an optimal coupling \cite[Theorem 4.1]{villani2009optimal} between the neighborhood couple $ \Big(\mathcal{N}_{G,d}(i), \mathcal{N}_{G',d}(\sigma_*(i)) \Big) $ and a pair of trees $(t, t') \sim\Pcorr_d$ fulfilling 
	\[
	\Prob\Big(\, \Big(\mathcal{N}_{G,d}(i), \mathcal{N}_{G',d}(\sigma_*(i)) \Big) \neq (t,t') \Big)\in o(1).
	\]
\end{lemma}
\begin{remark}
	It may seem surprising that the global correlation parameters $s, s'$ are sufficient in Definition \ref{def:corr_GW_trees} to describe the local tree laws, ignoring the fact whether edges or nodes were added during the sampling process of $(G, G')$. The reason is that adding a new edge or a new node is equivalent in a tree: Both operations result in adding one leaf and one edge. Hence, only the products of edge and node correlation parameters matter, i.e., $s = r q$ and $ s' = r' q'$.
\end{remark}
Having described the neighborhoods of correctly matched nodes $(i, \sigma_*(i))$, the following section describes a way of distinguishing these from incorrect matches.

\subsection{One-sided hypothesis testing}\label{section:one-sided_hypo_testing}

Given two nodes $i \in G$ and $j \in G'$, there are two cases: Either they should be aligned, i.e., $i \in V_*$ and $\sigma_*(i) = j$, in which case Lemma \ref{lemma:correlated_graphs_converge_towards_P1} characterizes the joint law $\Pcorr$ of their neighborhoods. Or they should not be matched, i.e., $i \neq V_*$ or $ \sigma_*(i) \neq j $. Intuitively, if $i$ and $j$ are far apart in $G_*$, the neighborhood laws of $i$ and $j$ are roughly independent, in the sense that these neighborhoods are approximate samples of the following joint tree law: Let $t \sim \mathrm{GW}^{(\lambda s)}$ and $t' \sim \mathrm{GW}^{(\lambda s')}$ with $t, t'$ independent, and define
\begin{equation*}\label{def:independent_GW}
	\Pind:= \mathrm{Law}(t, t') = \mathrm{GW}^{(\lambda s)} \otimes \mathrm{GW}^{(\lambda s')}.
\end{equation*}
Furthermore, for a depth $d \in \N$ we denote by $\Pind_d $ the law of the truncated trees $ (t_d, t'_d) $.  

The approximate independence of the tree neighborhoods of  $i$  and  $j$ can be established regardless of their distance in  $G_*$  by employing the following technique:

\paragraph{Dangling tree trick} Instead of comparing the entire neighborhoods of two nodes, we focus on subsets of these neighborhoods that are \textit{pointing away} from the nodes. Formally, for two adjacent nodes $i$ and $i'$ of $G$, we define
\[
\mathcal{N}_{G, d}(i\to i') := \Big\{ j \in \mathcal{N}_{G, d}(i)\, : \, \mathrm{dist}_G(j,i') < \mathrm{dist}_G(j,i)\Big\}.
\]
If $  \mathcal{N}_{G, d}(i) $ is cycle-free, it can be viewed as a tree rooted in $i$. The subset $ \mathcal{N}_{G, d}(i\to i') $ then corresponds to the subtree rooted in $i'$ which has depth $d-1$. Visually speaking, $ \mathcal{N}_{G, d}(i\to i')$  describes the neighborhood rooted in $i'$ that is pointing away from $i$. 

\textbf{Why introduce $\boldsymbol{ \mathcal{N}_{G, d}(i\to i')  }$?} Take $(G, G') \sim \mathrm{CER}(N, \lambda, s, s')$ and recall that their common subgraph $G_*$ is embedded in $G$ and $G'$ via the node sets $V_* \subset V$ and $V_*'\subset V'$ respectively. The following lemma and definition let us bridge the gap between pointing neighborhoods and approximately independent trees:
\begin{lemma}\label{lemma:danglingTree}
	For $(G, G') \sim \mathrm{CER}(N, \lambda, s, s')$, let $i \in V, j \in V'$ and suppose that
	\begin{enumerate}[label=(\roman*)]
		\item either $i \not \in V_*$ or $\sigma_*(i) \neq j$,
		\item $\mathcal{N}_{G, 2d}(i)$ and $ \mathcal{N}_{G', 2d}(j)$ are cycle-free,
		\item both nodes have at least three neighbors, named $ i_1, i_2, i_3 $ and $j_1, j_2, j_3$.
	\end{enumerate}
	Then, there is at least one pair $(i_k, j_k)$ such that $\mathcal{N}_{G, d}(i\to i_k)\cap V_*$ and $\mathcal{N}_{G, d}(j\to j_k)\cap V_*'$ have disjoint node sets when embedded in $G_*$. 
\end{lemma}
\begin{definition}\label{def:disjTrees}
	A pair of trees $(t, t')$ of depth $d$ is said to follow the law $\Pdisjdminus$ if they are sampled as $\mathcal{N}_{G, d}(i\to i_k)\cap V_*$ and $\mathcal{N}_{G, d}(j\to j_k)\cap V_*'$ in the above lemma conditioned on $(i), (ii), (iii)$ and the fact that they have disjoint node sets when embedded in $G_*$. \\
	In other words, $\Pdisjdminus$ is the law of a tree-like neighborhood pair in correlated Erdős--Rényi graphs, conditioned to avoid each other in the intersection graph $G_*$. 
\end{definition}
\begin{remark}\label{remark:dangling_tree}
	This is called the dangling tree trick because Lemma \ref{lemma:danglingTree} guarantees the existence of disjoint subtrees in the neighborhoods  $\mathcal{N}_{G,d}(i)$ and $ \mathcal{N}_{G',d}(j)$. Since these neighborhoods are assumed to be trees, their pointing subtrees are figuratively \textit{dangling} from their respective roots $i$ and $j$. The usefulness of this trick stems from the similarity of the laws $\Pdisjd$ and $\Pindd$, which we have only hinted at so far and will show later in Lemma \ref{lemma:disj_vs_ind}.
\end{remark}

\begin{proof}[Proof of Lemma \ref{lemma:danglingTree}] 
	Let $\tau$ denote the subgraph of $G_*$ which is spanned by the nodes of $\mathcal{N}_{G, d}(i) \cap V_*$ and rooted in $ i_* \in \argmin_{k \in V_*} d_G(i, k) $. Since $\mathcal{N}_{G,2d}(i)$ is a tree, $\tau$ is also cycle-free. Define $\tau'$ similarly, as the subgraph with node set $\mathcal{N}_{G', d}(j) \cap V_*'$ rooted in $ j_* \in \argmin_{k \in V_*} d_{G'}(j, k)$.
	
	On the one hand, if $\tau$ and $\tau'$ do not share any nodes in $G_*$, they are disjoint subtrees of the same Erdős--Rényi graph. When augmenting both of these subgraphs to become $ \mathcal{N}_{G, d}(i) $ and $ \mathcal{N}_{G', d}(j) $ during the sampling procedure of $(G, G')$, they remain disjoint subtrees of those graphs, potentially changing roots from $i_*$ to $i$ and from $ j_*$ to $j$. Since these trees share no nodes, their subtrees do not share any either, and consequently all pairs $\big( \mathcal{N}_{G, d}(i\to i_m), \mathcal{N}_{G', d}(j\to j_{m}) \big)$ are of law $\Pdisjdminus$.
	
	On the other hand, if $\tau$ and $\tau'$ share a node in $G_*$, then there is exactly one path $\mathcal{P}_{i_* \leftrightarrow j_*}$ connecting $i_*$ with $j_*$ and contained in $ \big[\mathcal{N}_{G, d}(i) \cap V_*\big]  \cup  \big[ \mathcal{N}_{G', d}(j) \big) \cap V'_* \big]  $: If there was another such path $\mathcal{P}'$, then the union of $\mathcal{P}_{i_* \leftrightarrow j_*}$ and $\mathcal{P}'$ would be a cycle in $\mathcal{N}_{G, 2d}(i)$ which is a contradiction to (ii). \\
	Since $  \mathcal{N}_{G, d}(i) \cap V_*$ is a tree, all pointing subtrees $\mathcal{N}_{G, d}(i\to i_m) \cap V_*, \, m\in[3]$ are disjoint, and the same holds for $\mathcal{N}_{G, d}(j\to j_m) \cap V'_*, \, m\in[3]$.\\
	Consequently, $\mathcal{P}_{i_* \leftrightarrow j_*}$ can only be contained in two of the three sets $ \big[\mathcal{N}_{G, d}(i\to i_m) \cap V_* \big] \cup \big[ \mathcal{N}_{G', d}(j\to j_m) \cap V'_*\big], \, m \in [3]$. In other words, there is $m_* \in [3]$ such that $\mathcal{P}_{i_* \leftrightarrow j_*}$ is disjoint from $\big[ \mathcal{N}_{G, d}(i\to i_{m_*})\cap V_* \big] \cup \big[ \mathcal{N}_{G', d}(j\to j_{m_*}) \cap V'_*\big]$. Therefore, $ \mathcal{N}_{G, d}(i\to i_{m_*}) \cap V_*$ and $ \mathcal{N}_{G', d}(j\to j_{m_*}) \cap V'_*$ share no nodes.  Augmenting these two tree-neighborhoods in their respective graphs, one obtains that the couple $\mathcal{N}_{G, d}(i\to i_{m_*}), \,  \mathcal{N}_{G', d}(j\to j_{m_*}) $ follows the law $\Pdisjdminus$. 
\end{proof}
The dangling tree trick and the later Lemma \ref{lemma:disj_vs_ind} allow us to consider pairs of independent Galton--Watson trees and compare them to pairs of trees with law $\Pcorrd$. 

Given a pair of rooted trees $ (t, t') $ and their truncations $(t_d, t'_d)$, our goal is to determine whether they are correlated or independent. For any depth $d$, this translates to the hypothesis testing problem
\begin{equation}\label{testing_problem}
	\begin{aligned}
		H_0^{(d)}&: \quad (t_d, t'_d) \sim 	\Pindd, \quad \text{i.e. the trees are independent,}\\
		\text{vs. }H_1^{(d)} &: \quad (t_d, t'_d) \sim 	\Pcorrd, \quad \text{i.e. the trees are correlated.}
	\end{aligned}
\end{equation}
In the asymptotic regime where $d \to \infty$, the goal is to find a sequence of tests $\mathcal{T}_d$ to distinguish between $H_0^{(d)}$ and $ H_1^{(d)}$ at all depths $d$. It turns out that the correct performance measure for the test we are looking for is defined below; this definition uses the notation $\mathcal{X}_d := \{\text{rooted, unlabeled trees of depth } \leq d\}$ and we refer to section \ref{section_unlabeled_trees} for a rigorous introduction of $\mathcal{X}_d$.
\begin{definition}[one-sided test]\label{def:one-sided_Testing}
	We call a function $ \mathcal{T}_d : \mathcal{X}_d \times \mathcal{X}_d \rightarrow\{0,1\}$ a \textit{test at depth $d$} and let $\Pindd (\mathcal{T}_d = 0)  $ denote the probability that a pair of truncated independent trees $ (t_d, t'_d) \sim\Pindd$ is correctly identified by $\mathcal{T}_d$, which is to say $ \mathcal{T}_d(t_d, t'_d) = 0 $.\\
	An \textit{asymptotic one-sided test } is then defined as a sequence of functions $(\mathcal{T}_d)_{d\in\N} $ with the following two properties:
	\begin{itemize}
		\item $\Pindd (\mathcal{T}_d = 1) \xrightarrow[d \to \infty]{} 0$, i.e., the tests have vanishing type-I-error.
		\item There exists $\beta > 0$ such that  $\Pcorrd(\mathcal{T}_d = 1) \geq \beta$ for all but finitely many $d$, meaning that the tests have significant power $\beta$. In other words $\liminf_{d\to \infty} \Pcorrd(\mathcal{T}_d = 1) > 0$. 
	\end{itemize}
\end{definition}
The entirety of Section \ref{section:tree_correl_testing} is dedicated to the question of whether one-sided tests exist for different parameter choices of $(\lambda, s, s')$ and how to define $(\mathcal{T}_d)_{d\in \N}$. For now, we consider those results a black box and suppose that one-sided tests exist, allowing us to define a graph matching algorithm.

\subsection{An algorithm to align graphs using tree comparison}
Throughout this section, we assume that $\lambda, s$ and $s'$ are such that there is a sequence of tests $(\mathcal{T}_d)_{d\in \N}$ for the problem (\ref{testing_problem}) which fulfill $\lambda ss' > 1$ and
\begin{equation}\label{eq:one-sided-assumptions}
	\Pindd\big(\mathcal{T}_d= 1\big) \leq \exp\big(-C (\lambda ss')^d\big)  \quad  \text{and} \quad \beta :=  \liminf_{d\to \infty} \Pcorrd(\mathcal{T}_d = 1) > 0.
\end{equation}
Following our above discussion, the idea for a graph matching algorithm is straightforward: For all pairs of nodes $i \in G$ and $j \in G'$, whose neighborhoods are trees, compare their dangling subtrees using $\mathcal{T}_{d}$ and decide whether to match $i$ with $j$. The following algorithm has first been introduced \cite{luca}:
\begin{algorithm}[H]
	\caption{ $\, \, $ MPAlign (message-passing for sparse graph alignment)}\label{algo}
	\hspace{-3.5cm}\textbf{Input:} Two graphs $G, G'$, a depth parameter $d$, and a test $\mathcal{T}_d$.\\
	\textbf{Output:} An alignment estimator $ \hat{\sigma}: \hat{V} \to \hat{V}' $ where $\hat{V}, \hat{V}'$ are node subsets of $G, G'$.
	\vspace{4pt}
	\begin{algorithmic}
		\State \textbf{Initialise} $\hat{V} \leftarrow \emptyset, \,  \hat{V}' \leftarrow \emptyset $ and let $\hat{\sigma}$ be the empty map,
		\For {$i \in V(G)$ and $j \in V(G')$}
		\If { neither $\mathcal{N}_{G, 2d}(i)$ nor  $\mathcal{N}_{G', 2d}(j)$ contain cycles,}
		\If { there exist triplets $i_1, i_2, i_3  \in \mathcal{N}_{G, 1}(i)$ and $j_1, j_2, j_3 \in \mathcal{N}_{G', 1}(j)$ such that
			\vspace{-8pt}
			\[
			\mathcal{T}_{d-1}\Big(\mathcal{N}_{G, d}(i \rightarrow i_k), \mathcal{N}_{G', d}(j \rightarrow j_k)\Big) = 1 \, \text{ for all } k\in \{1,2,3\},
			\vspace{-8pt}
			\]
			\hspace{45pt}
		} $\hat{V} \leftarrow \hat{V} \cup \{i\}, \, \hat{V}' \leftarrow \hat{V}' \cup \{j\}$ and set $\hat{\sigma}(i) = j$.
		\EndIf
		\EndIf
		\EndFor\\
		\Return $ \hat{\sigma}: \hat{V} \to \hat{V}'$
	\end{algorithmic}
\end{algorithm}
\begin{remark}[Runtime of MPAlign ]\label{remark: algo-runtime}
	We choose $d \in \mathcal{O}(\log(N))$ and show in Proposition \ref{prop_likhood_recursive}, that $\mathcal{T}_d$ can be computed in polynomial time. Since MPAlign is essentially a \textbf{for}-loop over all $\mathcal{O}(N^2)$ pairs of nodes, this algorithm is therefore polynomial as well. This is the minimum requirement for an efficient graph alignment algorithm.
\end{remark}
The following theorem establishes that MPAlign achieves one-sided partial alignment.
\begin{theorem}\label{thm:Algo_works}
	Let $\lambda > 0$ and $s, s' \in [0,1]$ be such that there exist tests $(\mathcal{T}_d)_{d\in\N}$ satisfying the one-sided feasibility conditions in (\ref{eq:one-sided-assumptions}). Define $(G, G') \sim \mathrm{CER}(N, \lambda, s, s')$ and let $d = \lfloor c \log(N)\rfloor$, where $c$ satisfies $4 c \log(\lambda \max\{s, s'\}) < 1$. Then, Algorithm \ref{algo} with inputs $(G, G', d, \mathcal{T}_d)$ yields an estimator $\hat{\sigma}: \hat{V} \to \hat{V}'$ which achieves one-sided partial alignment. In other words there exists $ \tilde{\beta} > 0 $ such that
	\begin{equation}\label{property:significant_overlap}
		\Prob\Big( \mathrm{ov}(\sigma_*, \hat{\sigma}) \geq \tilde\beta \Big) \xrightarrow[N\to\infty]{} 1 \quad \text{and} 
	\end{equation}
	\begin{equation}\label{property:negligible_error}
		\Prob\Big( \mathrm{err}(\sigma_*, \hat{\sigma}) = o(1) \Big) \xrightarrow[N\to\infty]{} 1.
	\end{equation}
\end{theorem}
This theorem's proof builds on the ideas in \cite[Section 6]{luca}, providing additional details and extending the analysis to settings with random node numbers and unequal edge densities.\\
Although technically involved, the proof offers valuable insights into the interplay between graph- and tree-level structures in the asymmetric setting. This result concludes the chapter and sets the stage for the next section, where we justify the feasibility assumptions in (\ref{eq:one-sided-assumptions}).
\begin{proof}
	We start by recalling that the node numbers of $G_*, G$ and $G'$ marginally follow binomial distributions with means $\Expe[n_*] = Nqq$, $\Expe[n] = Nq$, and $\Expe[n'] = Nq'$. Throughout this proof, we consider these numbers to be sampled before the remaining sampling process from Definition \ref{def:corr_GW_trees}. To make this rigorous, define
	\[
	\gamma :=  \frac{1}{4} \,\Big(\frac12 -  c \log(\lambda \max\{s, s'\}) \Big) \quad \text{and} \quad \varepsilon := N^{ \sfrac{1}{2} \, +\,  \gamma},
 \]
 and condition on the event
	\begin{align*}
		\mathcal{B} := 
		\Big\{n_* > Nqq' - \varepsilon \Big\} \cap \Big\{ n -\varepsilon < Nq < n + \varepsilon \Big\} \cap \Big\{n' - \varepsilon < Nq' < n' + \varepsilon\Big\}.
	\end{align*}
	Note that $\gamma > 0$ because $2 c \log(\lambda \max\{s, s'\}) \leq 4 c \log(\lambda \max\{s, s'\}) < 1$ as assumed in Theorem \ref{thm:Algo_works}. Consequently, using Chernoff's inequality for binomial random variables \cite[Theorem 4]{chung2006concentration}, the event $\mathcal{B}$ occurs with high probability: 
	\begin{align}\label{inequality_measure_binomial}
		\Prob(\mathcal{B}) &\geq 1 - \Prob\Big(n_* - \Expe[n_*] \leq -\varepsilon \Big) -  \Prob\Big( \big \vert n - \Expe[n] \big \vert \geq  \varepsilon \Big) - \Prob\Big( \big \vert n' - \Expe[n'] \big\vert \geq  \varepsilon \Big) \nonumber \\
		&\geq 1 - e^{- \frac{N^{1 + 2 \gamma}}{2Nqq'}} - \Big( e^{- \frac{N^{1 + 2\gamma }}{2Nq} } + e^{-\frac{N^{1 + 2 \gamma}}{2Nq + (\sfrac{2}{3}) N^{\sfrac12+ \gamma}}} \Big) - \Big( e^{- \frac{N^{1 + 2\gamma }}{2Nq' } } + e^{-\frac{N^{1 + 2 \gamma}}{2Nq' + (\sfrac{2}{3}) N^{\sfrac12+ \gamma}}} \Big) \nonumber \\
		& = 1 - 3 \, e^{- \Theta(N^{2 \gamma})}  \xrightarrow[N \to \infty]{} 1.
	\end{align}
	Conditioning on $\mathcal{B}$ will be instrumental to conclude the proofs of both (\ref{property:significant_overlap}) and (\ref{property:negligible_error}).
	
	Next, recall that $ G $ and $G'$ are randomly labeled augmentations of an intersection graph $G_*  \sim \mathrm{ER}(n_*, rr' \lambda /N)$. The nodes of $G_*$ are embedded in $G, G'$, and we identify these embeddings with the node subsets $V_* \subset V, V_*'  \subset V'$ respectively. These subsets are related via the true graph alignment $\sigma_*: V_* \to V_*'$. 
	
	Let $\hat{\sigma}: \hat{V} \to \hat{V}'$ be the alignment estimator from algorithm \ref{algo} whose randomness stems from the randomness of  $(G, G') $. For all pairs $i \in V, j \in V'$, the \textbf{if}-clauses of algorithm \ref{algo} translate to the event
	\begin{align}\label{def:Mij}
		\mathcal{M}(i,j):=&\big\{ \text{MPAlign matches } i \text{ with } j \big\} \nonumber\\
		=& \big\{ \mathcal{N}_{G, 2d}(i) \text{ and }\mathcal{N}_{G', 2d}(j) \text{  are cycle-free} \big\} \nonumber\\
		&  \cap  \big\{i \text{ and }j \text{ have at least three neighbors, } (i_k)_{k=1}^3 \text{ and } (j_k)_{k=1}^3 \big\} \nonumber\\
		&  \cap \big\{ \mathcal{T}_{d-1} \Big(\mathcal{N}_{G, d}(i \rightarrow i_k), \mathcal{N}_{G', d}(j \rightarrow j_k)\Big) = 1 \, \text{ for } k \in [3] \big\}.
	\end{align}
	For all nodes $ i \in V $, define the event of a correct match as
	\begin{equation*}
		\mathcal{M}_*(i) :=  \{i \in  V_*, \,   \hat{\sigma}(i) = \sigma_*(i)\}.
	\end{equation*}
	Recalling that $\square^c$ denotes the complement of a set $\square$, this lets us rewrite
	\begin{equation*}
		\mathrm{ov}(\sigma_*, \hat{\sigma}) = \frac{1}{n_*} \sum_{i \in \hat{V}} \indicator{\mathcal{M}_*(i)} \quad \text{and} \quad \mathrm{err}(\sigma_*, \hat{\sigma}) = \frac{1}{n} \sum_{i \in \hat{V}} \indicator{\mathcal{M}_*(i)^c}.
	\end{equation*}
	
	\paragraph{Proof of (\ref{property:significant_overlap}).} We start by showing significant overlap. For this, set
	\[
	\tilde{\beta} := \beta \,  \underbrace{\Pcorrd(\text{the intersection tree } t_* \text{ from Definition \ref{def:corr_GW_trees} has root degree } \geq 3 )}_{=: \alpha}
	\]
	where $\beta > 0$ is the value from (\ref{eq:one-sided-assumptions}) and $\alpha$ is a strictly positive probability only depending on the parameters $\lambda, s$ and $s'$. Hence $\tilde{\beta}>0$. The reason to introduce $\tilde{\beta}$ is its relationship with $\Prob(\mathcal{M}_*(i))$:\\
	For every $i \in \hat{V}$, denote by $\mathcal{C}_i$  the event of optimal coupling between the neighborhoods $ \Big( \mathcal{N}_{G, 2d}(i), \mathcal{N}_{G', 2d}(\sigma_*(i)) \Big)$ and a tree pair $ (t, t') \sim \Prob^\text{ corr}_{2d}$. Applying Lemma \ref{lemma:correlated_graphs_converge_towards_P1} with parameter $2d$ instead of $d$ and exploiting the hypothesis $4 c \log(\lambda \max\{s, s'\}) < 1$, we get that $C_i$ has probability
	\[
	\Prob(\mathcal{C}_i) = \Prob\Big(\, \Big(\mathcal{N}_{G,2d}(i), \mathcal{N}_{G',2d}(\sigma_*(i)) \Big) =  (t,t') \Big) = 1 - o(1). 
	\]
	Since $\mathcal{M}_*(i) = \{i \in V_*\} \cap \mathcal{M}(i, \sigma_*(i))$, conditioning on $\mathcal{C}_i$ yields
	\begin{align}\label{probMstari}
		\Prob\big(\mathcal{M}_*(i)\big) &=  \Prob\big(\mathcal{C}_i^c, \mathcal{M}_*(i)\big) + \Prob\big(\mathcal{C}_i, i \in V_*,  \mathcal{M}(i,\sigma_*(i)) \big) \nonumber \\
		&\geq o(1) + \alpha \, \Pcorrdminus(\mathcal{T}_{d-1}= 1)).
	\end{align}
	The second inequality translates the decomposition of $\mathcal{M}(i,j)$ from (\ref{def:Mij}) conditioned on $ \mathcal{C}_i$: Under $\mathcal{C}_i$, one has that $ \mathcal{N}_{G, d}(i) $ and $ \mathcal{N}_{G', d}(\sigma_*(i)) $ are cycle-free, the event $\big\{i$ \textit{and} $\sigma_*(i)$ \textit{have at least three neighbors}$\big\}$ has probability $\alpha$ and is independent of what happens below the roots, in particular of the event $ \big\{ \mathcal{T}_{d-1} \big(\mathcal{N}_{G, d}(i \rightarrow i_k), \mathcal{N}_{G', d}(\sigma_*(i) \rightarrow \sigma_*(i)_k)\big) = 1 $ \textit{for} $k \in [3] \big\} $. The probability of this last event is bounded below by $ \Pcorrdminus(\mathcal{T}_{d-1} = 1)$, the probability that one of these correlated neighborhoods yields a positive test. Independence and this lower bound leads to (\ref{probMstari}).\\
	Next, since $\liminf_{d\to \infty} \Pcorrdminus(\mathcal{T}_{d-1} = 1) = \beta > 0$, one has $ \Pcorrdminus(\mathcal{T}_{d-1} = 1) = (1 - o(1)) \beta  $, implying
	\begin{align}\label{probmistarfinal}
		&\Prob\big(\mathcal{M}_*(i)\big)  \geq o(1) + \alpha  \, (1 - o(1)) \beta  \nonumber\\
		& \quad \quad \implies \tilde{\beta} = \alpha \beta \leq \frac{\Prob\big(\mathcal{M}_*(i)\big) - o(1)}{(1 - o(1))} \leq (1+o(1)) \Prob\big(\mathcal{M}_*(i)\big).
	\end{align}
	Finally, in order to show (\ref{property:significant_overlap}), we will prove 
	\begin{equation*}
		\Prob\Big(\mathrm{ov}(\sigma_*, \hat{\sigma}) < \tilde{\beta} - \delta\Big) \xrightarrow[N \to \infty]{} 0  \quad \text{ for all } \delta > 0. 
	\end{equation*}
	To do so, we use the conditioning on $\mathcal{B}$ from (\ref{inequality_measure_binomial}), followed by the tower property of conditional expectation, the result from (\ref{probmistarfinal}), and Markov's inequality to obtain the following sequence of upper bounds: 
	\begin{align*}\label{eq:overlap_prob_upper_bound}
		\Prob\Big(\mathrm{ov}&(\sigma_*, \hat{\sigma}) < \tilde{\beta} - \delta\Big)   \leq \Prob(\mathcal{B}^c) + \Prob\Big(\mathcal{B}, \, \mathrm{ov}(\sigma_*, \hat{\sigma}) < \tilde{\beta} - \delta \Big) \nonumber\\
		&= o(1) + \Expe\Big[ \indicator{\mathcal{B}} \, \Prob\Big( \mathrm{ov}(\sigma_*, \hat{\sigma}) < \tilde{\beta} - \delta \, \Big \vert \, n_*, n, n' \Big)\Big]\nonumber\\
		&  = o(1) + \Expe\Big[ \indicator{\mathcal{B}} \, \Prob \Big(\tilde \beta n_* -  \sum_{i \in \hat{V}} \indicator{\mathcal{M}_*(i)} > \delta n_* \Big) \Big] \nonumber \\
		& \leq o(1) +  \Expe\Big[ \indicator{\mathcal{B}} \, \Prob \Big( (1+o(1)) \sum_{i \in \hat{V}} \Big( \Prob\big(\mathcal{M}_*(i) \big) - \indicator{\mathcal{M}_*(i)}  \Big)> \delta n_* \Big) \Big] \nonumber \\
		& \leq o(1)  + \Expe\Big[ \indicator{\mathcal{B}} \, \frac{1+o(1)}{\delta^2 n_*^2} \Big( n \Var\big(\indicator{\mathcal{M}_*(1)} \big) + n(n-1) \mathrm{Cov}\big(\indicator{\mathcal{M}_*(1)}, \indicator{\mathcal{M}_*(2)}\big)\Big) \Big].
	\end{align*}
	Our goal is now to find upper bounds for the last line: The variance of $\indicator{\mathcal{M}_*(1)} \in [0,1]$ is bounded by  $ 1 $ and the first summand in the expectation vanishes because
	\begin{equation}\label{varbound}
		\Expe\Big[ \indicator{\mathcal{B}} \, \frac{1+o(1)}{\delta^2 n_*^2} n \Var\big(\indicator{\mathcal{M}_*(1)} \big) \Big] \leq \frac{1+o(1)}{\delta^2 (Nqq - N^{\sfrac12 + \gamma})^2} (Nq + N^{\sfrac12 + \gamma}) \xrightarrow[N \to \infty]{} 0. 
	\end{equation}
	For two distinct nodes, for instance $i=1$ and $j=2$, the term $ \mathrm{Cov}\big(\indicator{\mathcal{M}_*(1)}, \indicator{\mathcal{M}_*(2)}\big) \in [-1,1] $ equals $ 0 $ if $\big( \mathcal{N}_{G, d}(i), \mathcal{N}_{G', d}(\sigma_*(i)) \big)$ and $\big( \mathcal{N}_{G, d}(j), \mathcal{N}_{G', d}(\sigma_*(j)) \big)$ are independent. This independence can be obtained asymptotically by using Lemma \ref{lem_independence_of_different_neighborhoods} which leads to $\mathrm{Cov}(\indicator{\mathcal{M}_*(1)}, \indicator{\mathcal{M}_*(2)})  \in o(1)$. Combining this fact with a similar bound as in (\ref{varbound}) lets us conclude the proof of (\ref{property:significant_overlap}) since
	\[ \Prob\Big(\mathrm{ov}(\sigma_*, \hat{\sigma}) < \tilde{\beta} - \delta\Big) \leq o(1) + \Expe\Big[ \indicator{\mathcal{B}} \, \frac{1+o(1)}{\delta^2 n_*^2} n^2 \mathrm{Cov}\big(\indicator{\mathcal{M}_*(1)}, \indicator{\mathcal{M}_*(2)}\big) \Big]\xrightarrow[N \to \infty]{} 0.
	\]
	
	\paragraph{Proof of (\ref{property:negligible_error}).} By Remark \ref{equivalence:vanishesWHP_convergesINP}, we need to show $ \mathrm{err}(\sigma_*, \hat{\sigma}) \xrightarrow[N\to \infty]{} 0$ in probability. For this, let $ \delta > 0 $ and use (\ref{inequality_measure_binomial}) to obtain
	\[
	\Prob\Big(\mathrm{err}(\sigma_*, \hat{\sigma})> \delta \Big)  \leq  \Prob\Big(\mathrm{err}(\sigma_*, \hat{\sigma})> \delta, \mathcal{B}\Big) + \Prob(\mathcal{B}^c) = \Prob\Big(\mathrm{err}(\sigma_*, \hat{\sigma})> \delta, \mathcal{B}\Big) + o(1).
	\]
	Furthermore, Markov's inequality yields
	\begin{align*}
		\Prob\Big( \underbrace{\sum\nolimits_{i \in \hat{V}} \indicator{\mathcal{M}_*(i)^c}}_{=n \times \mathrm{ err}} > n \delta, \mathcal{B} \Big) & \leq \frac{1}{\delta n} \sum_{i \in V} \Prob\Big(\mathcal{M}_*(i)^c, i \in \hat{V}, \mathcal{B} \Big) \notag \\ &\leq \frac{1}{\delta} \max_{i \in V} \Prob\Big(\mathcal{M}_*(i)^c, i \in \hat{V}, \mathcal{B} \Big).
	\end{align*}
	Hence, in order to show $ \mathrm{err}(\sigma_*, \hat{\sigma}) \xrightarrow{\Prob} 0$ it suffices to check  $ \Prob(\mathcal{M}_*(i)^c, i\in\hat{V}, \mathcal{B}) \xrightarrow[]{N \to \infty} 0 $ for all $ i \in V$. 
	
	Fix any $i \in V$ and decompose the event $ \mathcal{M}_*(i)^c \cap \{i \in \hat{V}\}$ of a wrong match by MPAlign as follows:
	\begin{align*}
		\mathcal{M}_*(i)^c \cap \{ i \in \hat{V}\} & = \bigcup_{j \in \, G' \setminus \square}\mathcal{M}(i,j) \quad \text{where } \square = \begin{cases}
			\emptyset &\text{ if } i \in \hat{V} \setminus V_*,\\
			\{\sigma_*(i)\} & \text{ if } i \in \hat{V}, \hat{\sigma}(i) \neq \sigma_*(i).
		\end{cases}
	\end{align*}
	Next, define the event of exploding neighborhood sizes across $G$ and $G'$ as follows:
	\begin{align*}
		\mathcal{A} := &\Big\{ \exists i \in V, \, t \in [d] : \, | \mathcal{S}_{G, t}(i) | > C \log(n) (\lambda s)^t \Big\}\\ &\cup \Big\{ \exists j \in V', \, t' \in [d]  : \, | \mathcal{S}_{G', t'}(j) | > C \log(n') (\lambda s')^t \Big\}.
	\end{align*}
	The probability of $\mathcal{A}$ is vanishing according to Lemma \ref{lem_neighborhoodsize}. Consequently, 
	\begin{align*}
		\Prob\Big(\mathcal{M}_*(i)^c, i\in\hat{V}, \mathcal{B} \Big) &\leq \Prob(\mathcal{A}) + \Prob\Big( \mathcal{M}_*(i)^c, i\in\hat{V}, \mathcal{A}^c, \mathcal{B}\Big) \\ &= o(1) + \Prob\Big( \mathcal{M}_*(i)^c, i\in\hat{V}, \mathcal{A}^c, \mathcal{B}\Big)
	\end{align*}
	which leads us to work on the event $ \mathcal{M}_*(i)^c \cap \{i \in \hat{V}\} \cap \mathcal{A}^c \cap \mathcal{B}$ in the sequel. One can decompose this event using the union bound
	\begin{align*}
		\Prob\Big(\mathcal{M}_*(i)^c,  &i \in \hat{V}, \mathcal{A}^c, \mathcal{B}\Big) \leq \sum_{j \in  G'  \setminus \square} \Prob\Big(\mathcal{M}(i,j)  , \mathcal{A}^c, \mathcal{B}\Big) \,  = \, \sum_{j \in G' \setminus \square} \Expe\Big[ \indicator{\mathcal{A}^c\cap \mathcal{B}} \indicator{\mathcal{M}(i,j)} \Big] \\
		& \leq  \sum_{j \in G' \setminus \square} \Expe\Big[ \indicator{\mathcal{A}^c \cap\mathcal{B}} \indicator{\mathcal{T}_{d-1}(\mathcal{N}_{G, d}(i \to i_k), \mathcal{N}_{G', d}(j \to j_k))  = 1 \textit{ for some special } k } \Big].
	\end{align*}
	The last inquality uses the dangling tree trick: By ``\textit{some special} $k$'', we mean the index $k$ from Lemma \ref{lemma:danglingTree} which lets $\mathcal{N}_{G, d}(i\to i_k)\cap V_*$ and $\mathcal{N}_{G'\, d}(j\to j_k)\cap V_*'$ be disjoint. From this same lemma, we also know that $ \big( \mathcal{N}_{G, d}(i \to i_k), \mathcal{N}_{G', d}(j \to j_k)\big) \sim \Pdisjdminus$ which leads to 
	\begin{align}\label{almost_endspurt}
		\Prob(\mathcal{M}_*(i)^c, \, i \in \hat{V}, \mathcal{A}^c , \mathcal{B} )  &\leq \sum_{j \in G'  \setminus \square} \Edisjdminus\Big[\indicator{\mathcal{A}^c \cap \mathcal{B}} \indicator{\mathcal{T}_{d-1} = 1} \Big]  \nonumber \\
		&\leq  n' \, \Edisjdminus\big[\indicator{\mathcal{A}^c \cap \mathcal{B}} \indicator{\mathcal{T}_{d-1} = 1} \big].
	\end{align}
	For two trees $t$ and $t'$, define the likelihood ratio of $\Pdisjdminus$ and $\Pinddminus$ as
	\[
	R_{d-1}(t, t') :=\Pdisjdminus(t, t') / \Pinddminus(t, t')
	\]
	and write $R_{d-1}$ without arguments for the random variable if the trees follow some indicated law. This lets us write
	\begin{align}\label{endspurt}
		n' \, \Edisjdminus\big[\indicator{\mathcal{A}^c \cap \mathcal{B}} \indicator{\mathcal{T}_{d-1} = 1} \big] &= n' \, \Einddminus\big[ R_{d-1} \indicator{\mathcal{A}^c \cap \mathcal{B}} \indicator{\mathcal{T}_{d-1} = 1} \big]\nonumber \\
		&\leq n' \, \sqrt{ \Einddminus\big[ R_{d-1}^2 \indicator{\mathcal{A}^c \cap \mathcal{B}} \big] } \sqrt{\Pinddminus(\mathcal{T}_{d-1} = 1)}
	\end{align}
	where we have used the Cauchy--Schwarz inequality. At this point, we need to \textit{quantify the similarity between} $\Pdisjdminus$ \textit{and} $ \Pinddminus $ as announced in Remark \ref{remark:dangling_tree}. The following lemma bounds the second moment of their likelihood ratio:
	\begin{lemma}\label{lemma:disj_vs_ind}
		Under the parameters of Theorem \ref{thm:Algo_works} and using the above notations, one has
		\begin{equation*}
			\Einddminus\Big[ R_{d-1}^2 \indicator{\mathcal{A}^c \cap \mathcal{B}} \Big] \in \mathcal{O}(1).
		\end{equation*}
	\end{lemma}
	The proof of this lemma, which we defer to Appendix \ref{appendix:RdSecondMoment}, uses well-established techniques to describe the relationship between graph neighborhoods and branching processes. We refer to \cite[p. 161]{alon2016probabilistic} for a concrete example and mention \cite{bollobas2007phase} for an extension to inhomogeneous graphs. The proof presented in Appendix \ref{appendix:RdSecondMoment} considerably extends and clarifies its analog in the symmetric case, the proof of Lemma 6.1 in \cite{luca}.
	
	Finally, combining the hypothesis $\Pinddminus \big(\mathcal{T}_{d-1} =  1\big) \leq \exp\big(-C (\lambda ss')^{d-1}\big) $ from (\ref{eq:one-sided-assumptions}) with Lemma \ref{lemma:disj_vs_ind} and the equations (\ref{almost_endspurt}), (\ref{endspurt}), one obtains
	\[
	\Prob(\mathcal{M}_*(i)^c,  i \in \hat{V}, \mathcal{A}^c, \mathcal{B}  ) \leq n' \sqrt{\mathcal{O}(1)} \sqrt{\exp(-C (\lambda ss')^{d-1})}.
	\]
	Since $d = \lfloor c \log(N) \rfloor $, one has $ n' \leq N < e^{(d+1)/c}$, and hence the right hand side tends to $0$ as we let $d\to \infty$ (and a fortiori for $N \to \infty$). Together with our initial Markov bound on $\Prob(\mathrm{err} > \delta)$, this lets us conclude  $ \mathrm{err}(\sigma_*, \hat{\sigma}) \xrightarrow{\Prob} 0$ which completes the proof.
\end{proof}

\section{Tree correlation testing}\label{section:tree_correl_testing}
As the previous section shows, distinguishing between independent and correlated trees is the key building block for sparse graph alignment. Therefore, the statistical feasibility of tree correlation testing in the asymmetric setting is the central question remaining to be answered. This section presents a one-sided test for the following problem which was introduced in (\ref{testing_problem}):
\begin{equation}\label{testing_problem_2}
	H_0^{(d)} : \,  (t_d, t'_d) \sim 	\Pindd \quad \text{vs} \quad
	H_1^{(d)}: \, (t_d, t'_d) \sim 	\Pcorrd.
\end{equation}
We show that the existence of such one-sided tests achieves a sharp phase transition at $ss' = \alpha$ where $\alpha \approx 0.3383$. To do so, we present an analysis relying on advantageous properties of the likelihood ratio in the tree correlation testing problem, primarily based on \cite{luca, ganassali2022statistical}. This section is the theoretical backbone of this paper and presents our novel contribution to tree correlation testing.

We start by proposing a formalism for unlabeled trees before recalling an essential result on the number of such trees. Next, we introduce the likelihood ratio for (\ref{testing_problem_2}) and present two of its key properties. These will then serve as proof engines for the two main results of this section, Theorem \ref{thm_equivalences} and Theorem \ref{thm:phase_transition}.

\subsection{Unlabeled trees}\label{section_unlabeled_trees}
Since node labels are unknown in the sparse graph alignment problem, local neighborhoods are rooted unlabeled trees. Our goal is to describe the laws of these trees, which requires a rigorous definition of what we mean by the \textit{set of rooted unlabeled trees of depth at most} $d$, which we denote as $\mathcal{X}_d$. 

A first idea is to define $\mathcal{X}_d$ as the set of cycle-free, non-empty graphs $t$  with a distinct root node $\rho$ such that $\mathrm{dist}(\rho, i) \leq d$ for all nodes $i \in t$. For two rooted trees $(t, \rho)$ and $(t', \rho')$, we consider them equal if there is a graph isomorphism mapping $t$ to $t'$ and $\rho$ to $\rho'$. The notion of \textit{equality by isomorphism} translates our intuition of the word \textit{unlabeled} since all relabelings of structurally identical rooted trees are treated as a single element of $\mathcal{X}_d$. However, there is an equivalent definition of $\mathcal{X}_d$ which defines the same object but is better suited for some of our analysis. It defines unlabeled trees recursively as follows:
\begin{definition}[unlabeled rooted trees]\label{def:unlabeled_rooted_trees}
	Set $\mathcal{X}_0 := \{\bullet\}$ where $\bullet$ denotes the graph consisting of a single node. For all $d \in \N$,  recursively define 
	\[
	\mathcal{X}_{d+1}  = \Big\{ (N_\tau)_{\tau \in \mathcal{X}_d} \in \N^{\mathcal{X}_{d}} \, : \, \sum_{\tau \in \mathcal{X}_d} N_\tau < \infty \Big\}. 
	\]
	By this definition, we understand an element $ t \in \mathcal{X}_{d+1} $ as a sequence of natural numbers indexed by unlabeled trees of depth at most $d$, i.e., $ t = (N_\tau)_{\tau \in \mathcal{X}_d}$ where only finitely many $N_\tau$ are non-zero. For every $\tau \in \mathcal{X}_d$, the number $N_\tau$ denotes how many copies of $\tau$ are attached to the root of $t$. This is best understood by looking at Figure \ref{unlabeled_trees}.
	
	We furthermore set $\mathcal{X}_{d+1}^{{(n)}} := \{t = (N_\tau)_{\tau \in \mathcal{X}_d} \in \mathcal{X}_{d+1} \, : \,  |t| = n\}$ where $|t| = 1 + \sum_{\tau \in \mathcal{X}_d} N_\tau |\tau|$ denotes a tree's node number. In other words, $\mathcal{X}_{d+1}^{(n)} $ contains all unlabeled trees with $n$ nodes and depth bounded by $d+1$.  
\end{definition}
\begin{figure}[t]
	\includegraphics[width=\textwidth]{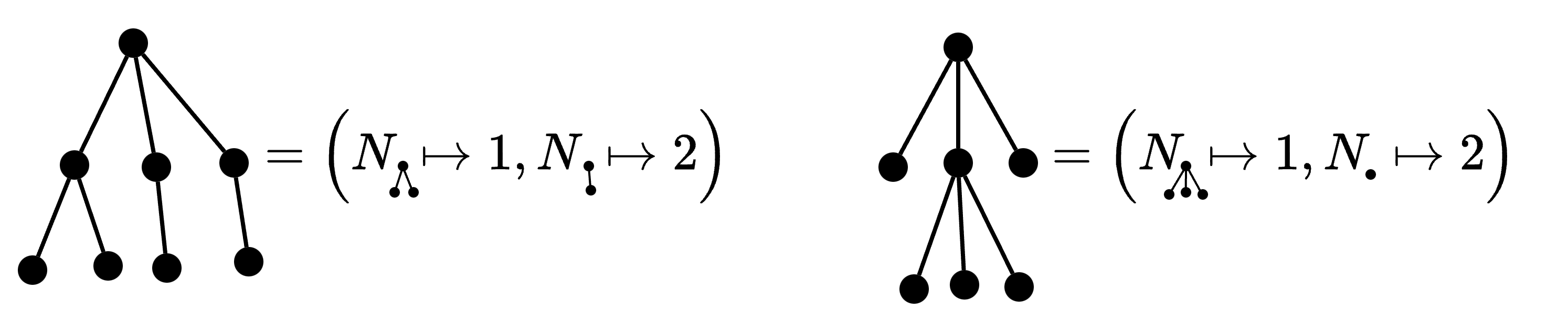}
	\caption{Two examples of unlabeled trees and their representation as sequences of natural numbers indexed by trees.}
	\label{unlabeled_trees}
\end{figure}
It is immediate by induction that $\mathcal{X}_d$ is countable and $\mathcal{X}_{d}^{(n)} $ is finite for all choices of $d, n \in \N$. Note that a tree of size $n$ can have a maximum depth of $n-1$ and hence $\mathcal{X}_{n-1}^{(n)} $ describes all unlabeled trees with $n$ nodes. The asymptotic cardinality of $\mathcal{X}_{n-1}^{(n)} $ is explicitly known thanks to the following result due to Otter \cite{otter1948number}.
\begin{proposition}\label{Otter}
	For $\alpha : =0.3383219...$ which is known as Otter's constant, there exists $D >0$ such that
	\begin{equation*}
		| \mathcal{X}_{n-1}^{(n)} | \widesim{n\to\infty} \frac{D}{n^{\sfrac{3}{2}}} \frac{1}{\alpha^n}.
	\end{equation*}
	This implies that the power series $ \Phi(x) = \sum_{n=1}^\infty | \mathcal{X}_{n-1}^{(n)} | \, x^n$
	has radius of convergence $\alpha > 0$ with $\Phi(\alpha) < \infty$. 
\end{proposition}
This proposition plays a key role when examining the likelihood ratio, as we explain in section \ref{section_likhood_ratio_properties}. Before delving into that, we first need to translate the laws of correlated Galton--Watson trees as introduced in section \ref{section_correlated_GW_trees} to the new formalism of $\mathcal{X}_{d+1}$. 

\subsection{Correlated Galton--Watson trees, revisited}
Given a pair of trees $ (t, t') \sim \Pcorrdplus $ represented by their subtree tuples $ t = (N_\tau)_{\tau \in \mathcal{X}_d} $ and $ t'= (N'_{\tau'})_{\tau' \in \mathcal{X}_d} $ as in Definition \ref{def:unlabeled_rooted_trees}, how can we describe their joint law using only $N_\tau$ and $N'_{\tau'}$? The following definition and lemma answer this question:
\begin{definition}\label{def:new_galton_watson_trees}
	A pair of trees $ t = (N_\tau)_{\tau}, \;  t'= (N'_{\tau'})_{\tau'} \in \mathcal{X}_{d+1}$ has joint law $ \widetilde\Pcorrdplus $ if for every $\tau, \tau' \in \mathcal{X}_d$, one has 
	\[
	N_\tau = \Delta_\tau + \sum_{\tau' \in \mathcal{X}_d} \Gamma_{\tau, \tau'} \quad \text{and} \quad N'_{\tau'} = \Delta'_{\tau'} + \sum_{\tau \in \mathcal{X}_d} \Gamma_{\tau, \tau'}
	\]
	where the independent random variables $\Delta, \Delta'$ and $\Gamma$ follow the laws
	\begin{align*}
		\Delta_\tau \sim \mathrm{Poi}\Big(\lambda s (1-s') \,\mathrm{GW}_d^{(\lambda s)}(\tau)\Big)&,\quad  \Delta'_{\tau'} \sim \mathrm{Poi}\Big(\lambda s' (1-s) \, \mathrm{GW}_d^{(\lambda s')}(\tau')\Big), \\ 
		\text{ and } \Gamma_{\tau, \tau'} &\sim \mathrm{Poi}\Big(\lambda s s' \, \Pcorrd(\tau, \tau')\Big).
	\end{align*}
\end{definition}
\begin{lemma}
	The laws $ \Pcorrdplus $ from Definition \ref{def:corr_GW_trees} and $ \widetilde\Pcorrdplus $ from Definition \ref{def:new_galton_watson_trees} are equal. Hence, we always write $ \Pcorrdplus $, also when working with the subtree tuple representation. 
\end{lemma}
\begin{proof}
	Let $(t, t') \sim \Pcorrdplus$ and recall that $t, t'$ are independent augmentations of a common intersection tree $t_* \sim \mathrm{GW}^{(\lambda s s')}$. Every node in $t$ adjacent to the root falls into one of two categories:
	\begin{enumerate}[label=(\alph*)]
		\item The node was present in $t_*$ before the augmentation, in which case the subtree attached to that node is shared between $t$ and $t'$, 
		\item The node was not in $t_*$, in which case the subtree attached to that node is of law $\mathrm{GW}_{d}^{(\lambda s)}$ and independent from $t'$. 
	\end{enumerate}
	For category (b), let $c_+ \sim \mathrm{Poi}(\lambda s(1-s'))$ be the number of non-shared children of the root in $t$.
	Conditionally on $c_+$, the attached subtrees are i.i.d. with law $\mathrm{GW}_d^{(\lambda s)}$, meaning that each child carries a mark
	$\tau \in \mathcal X_d$ with probability $\mathrm{GW}_d^{(\lambda s)}(\tau)$.
	By Poisson thinning for marked Poisson variables, the counts
	\[
	\Delta_\tau := \#\{\text{category (b) root-children in $t$ whose subtree equals }\tau\}
	\]
	are independent and satisfy $
	\Delta_\tau \sim \mathrm{Poi} \big(\lambda s(1-s')\,\mathrm{GW}_d^{(\lambda s)}(\tau)\big)$ for all $\tau\in\mathcal X_d$. 
	The same argument applies to $t'$ yielding $\Delta'_{\tau'}$.
	
	For category (a), let $k \sim \mathrm{Poi}(\lambda ss')$ be the number of shared root-children coming from the intersection tree $t_*$.
	Conditionally on $k$, the pairs of attached subtrees are i.i.d. with joint law $(\tau,\tau')\sim \Pcorrd$.
	By Poisson thinning, the pair-counts
	\[
	\Gamma_{\tau,\tau'} := \#\{\text{category (a) shared root-children with pair }(\tau,\tau')\}
	\]
	are independent and satisfy $\Gamma_{\tau,\tau'} \sim \mathrm{Poi} \big(\lambda ss'\,\Pcorrd(\tau,\tau')\big)$ for all pairs $(\tau,\tau')\in\mathcal X_d^2$.
	Hence the total number of category (a) root subtrees equal to $\tau$ in $t$ is $\sum_{\tau'} \Gamma_{\tau,\tau'}$. Analogously, $\sum_{\tau} \Gamma_{\tau,\tau'}$ counts the $\tau'$-subtrees in $t'$.
	
	Since $N_\tau$ and $N'_{\tau'}$ count the number of occurrences of $\tau$ and $\tau'$ as root subtrees in $t$ and $t'$ respectively, and because every root subtree is either in category (a) or (b), we obtain the expressions for $N$ and $N'$ in Definition \ref{def:new_galton_watson_trees}.
\end{proof}
The description of $\Pcorrdplus$ using subtree tuples was the last necessary ingredient before examining the likelihood ratio of the testing problem (\ref{testing_problem_2}).

\subsection{The likelihood ratio and its properties}\label{section_likhood_ratio_properties}

The Neyman-Pearson Lemma strongly suggests using the likelihood ratio when it comes to hypothesis testing. Therefore, we introduce this object in the context of unlabeled trees as a first step.

For two unlabeled trees $\tau, \tau' \in \mathcal{X}_d$ recall that $\tau = \tau'$ means that they are equal up to relabeling. With this in mind, we denote the likelihood of two given trees $t, t'$ being correlated up to depth $d \in \N$ as
\[ 
\Pcorrd(t, t') := \Prob_{(\tau, \tau')\sim  \Pcorr}\Big((\tau_d, \tau'_d) = (t_d, t'_d)\Big) \, = \, \Pcorrd\Big((\tau, \tau') = (t, t')\Big). 
\]
In the last expression and for the rest of this section, $(\tau, \tau')$ implicitly denotes the random variable issued from the indicated law ($\Pcorr_d $ in this case). Since $\Pcorr_d $ only considers trees up to depth $d$, our short notation $\Pcorrd(t, t') $ more accurately refers to $\Pcorrd(t_d, t'_d)$. We define $ \Pindd(t, t') $ analogously as the likelihood of $ (t_d, t'_d)$ following the distribution $\Pindd$.

With these notations, the likelihood ratio for the testing problem (\ref{testing_problem_2}) is
\[ 
L_d(t, t') := \dfrac{\Pcorrd(t, t')}{\Pindd(t, t')}.
\]
If $(t, t')$ is a pair of random trees, $L_d(t, t')$ becomes a random variable which we will often denote as $ L_d $ when the law of $ (t, t') $ is evident from the context.

\paragraph{Recursive property.} The reason for studying $L_d(t, t')$ is that it fully determines one-sided testability, as we show later, in ``(a) $\Leftrightarrow$ (b)'' of Theorem \ref{thm_equivalences}. Consequently, we will test for tree correlation using the likelihood ratio, making it crucial to find efficient ways of computing it. The following proposition solves this problem by providing a recursive formula for $L_d$.
\begin{proposition}[Recursive likelihood ratio formula]\label{prop_likhood_recursive}
	For $t, t' \in \mathcal{X}_d$, let $c, c'$ be the degrees of their respective root nodes. Denote by $t_{[1]}, \dots, t_{[c]}$ respectively $t'_{[1]}, \dots, t'_{[c']}$ the subtrees attached to these roots, labeled arbitrarily. Recall that $\sigma : [k] \xhookrightarrow{} [c]$ denotes an injective map from $[k]$ to $[c]$ and let $\sum_{\sigma:[k] \hookrightarrow [c]}$ be the sum over all such mappings.  Then, for all $d \in \N_{>0}$, the likelihood ratio at depth $d$ can be expressed recursively as
	\begin{equation*}
		L_d(t, t') = \sum_{k=0}^{c \wedge c'} \psi(k, c, c') \sum_{\substack{\sigma : [k] \xhookrightarrow{} [c],\\\sigma' : [k] \xhookrightarrow{} [c']}} \prod_{i=1}^k L_{d-1}\Big(t_{[\sigma(i)]}, t'_{[\sigma'(i)]}\Big)
	\end{equation*}
	where 
	\begin{equation*}
		\psi(k, c, c')= \exp(\lambda s s')\frac{1}{\lambda^k \, k!} \, (1-s')^{c - k} (1-s)^{c'-k}.
	\end{equation*}
\end{proposition}
In Appendix \ref{appendix:prop_likhood_recursive}, the reader can retrieve the proof of this result.
\begin{remark}
	Recall from Remark \ref{remark: algo-runtime} that the graph alignment algorithm MPAlign is polynomial if tree correlation tests can be computed in polynomial time. Since both $c$ and $c'$ are of order $\mathcal{O}(1)$ in the sparse regime, the recursive formula from Proposition \ref{prop_likhood_recursive} enables the desired efficient computation of $L_d$. With $d \in \mathcal{O}(N)$, this computation is even in $\mathcal{O}(\mathrm{polylog}(N))$. 
\end{remark}
One can repeatedly apply the recursive likelihood ratio formula to obtain an explicit lower bound for $L_{d+k}$, which will serve later in the proof of Theorem \ref{thm_equivalences}.
\begin{proposition}\label{prop_explicit_likhood}
	Let \( (t, t') \) be a pair of rooted trees sharing a common rooted subtree \( t_* \). Denote by 
	\( \sigma_*: t_* \hookrightarrow t \) and \( \sigma'_*: t_* \hookrightarrow t' \) 
	the embeddings of \( t_* \) into \( t \) and \( t' \), respectively. 
	
	For each \( i \in t \), let \( c_t(i) \) denote its number of child nodes in \( t \), and define 
	\( t_{[i]} \) as the subtree rooted at \( i \), containing all its descendants. 
	This notation extends naturally to \( t' \) and \( t_* \), and we recall that \( (t_*)_d \) denotes the depth-\( d \) truncation of \( t_* \).
	
	Then, for all \( d, k \in \mathbb{N}_{>0} \), the following inequality holds:
	\[
	L_{d+k}(t, t') \geq 
	\prod_{i \in (t_*)_{d-1}} \hspace{-8pt}
	\psi\Big(c_{t_*}(i), c_t(\sigma_*(i)), c_{t'}(\sigma'_*(i)) \Big) 
	\prod_{j \in (t_*)_d \setminus (t_*)_{d-1}} \hspace{-8pt}
	L_k(t_{[\sigma_*(j)]}, t'_{[\sigma'_*(j)]}).
	\]
\end{proposition}
The proof of this proposition is deferred to Appendix \ref{appendix:prop_explicit_likhood}.

The representations of $ L_d $ from these propositions have their worth outside of computational inquiries, for instance, when proving the next property of $ L_d $:

\paragraph{Martingale property.}
If we consider the subsequent likelihood ratios $(L_d(t, t'))_d$ as a random process, it is natural to suspect it to be a martingale with respect to $\Pind$ and the filtration $\mathcal{F}_d = \sigma((t_d, t'_d))$. The martingale properties are automatically true if the random variables in the likelihood ratio have Lebesgue densities \cite{hansen}. In our case, it can easily be checked by computation: Following the reasoning in Section 2.3 of \cite{luca}, it suffices to verify that the following expression equals 1:
\begin{align*}
	&\Eind\bigg[\sum_{k=0}^{c \wedge c'} \frac{\pi_{\lambda s s'}(k) \pi_{\lambda s (1-s')}(c-k) \pi_{\lambda s' (1-s)}(c'-k)}{\pi_{\lambda s}(c) \,\pi_{\lambda s'}(c')}\bigg] \\
	&= \sum_{k=0}^{\infty} \Eind\bigg[e^{\lambda ss'}\frac{1}{\lambda^k \, k!} \, (1-s')^{c - k} (1-s)^{c'-k} \dfrac{c!\,  c'!}{(c - k)!(c'-k)!} \, \indicator{k \leq c \wedge c'}\bigg]\\
	&= e^{\lambda ss'} \sum_{k=0}^\infty \frac{1}{\lambda^k k!} \Eind\bigg[(1-s')^{c-k}\frac{c!}{(c-k)!}\indicator{k \leq c}\bigg] \Eind\bigg[(1-s)^{c'-k}\frac{c'!}{(c'-k)!}\indicator{k \leq c'}\bigg]\\
	&= 1.
\end{align*}
The last two equalities follow from the independence of \( c \) and \( c' \) and their distributions, 
\( c \sim \mathrm{Poi}(\lambda s) \) and \( c' \sim \mathrm{Poi}(\lambda s') \). 
Due to \( \Eindd[L_d] = 1 \) and \( (L_d)_d \) is trivially adapted to the filtration, 
it follows that the likelihood ratio forms a martingale.

Since $L_d \geq 0$, the Martingale Convergence Theorem yields an almost sure limit $L_\infty \in L^1$. As a consequence, we obtain a fixed point formula for $ \Eind[L_\infty] $ from the recursive representation of Proposition \ref{prop_likhood_recursive}: Using that $c, c'$ are independent from $t_{[\sigma(i)]}, t'_{[\sigma'(i)]}$ in Galton--Watson trees, letting $d \to \infty$ and then taking $\Eind$, we obtain
\begin{align*}
	\Eind[L_\infty] &= \Eind\bigg[\sum_{k=0}^{c \wedge c'} \psi(k, c, c') \sum_{\substack{\sigma : [k] \xhookrightarrow{} [c],\\\sigma' : [k] \xhookrightarrow{} [c']}} \prod_{i=1}^k  \Eind[ L_\infty]\bigg] \\
	&= \Eind\bigg[\sum_{k=0}^{c \wedge c'} \frac{ \pi_{\lambda s (1-s')}(c-k) \pi_{\lambda s' (1-s)}(c'-k)}{\pi_{\lambda s}(c) \,\pi_{\lambda s'}(c')}\; \pi_{\lambda s s'}(k) \Eind[L_\infty]^k\bigg]\\
	& = \sum_{k = 0}^\infty \pi_{\lambda s s'}(k) \Eind[L_\infty]^k
\end{align*} 
where the last step follows from a similar computation as above. 

Consequently, $\Eind[L_\infty]$ is a fixed point of the probability generating function of a $ \mathrm{Poi}(\lambda s s') $ random variable. This lets us directly link $\Eind[L_\infty]$ to the theory of Galton--Watson processes with average offspring $ \lambda s s' $:\\
Recall that in the supercritical regime of such processes, their extinction probability is equal to the smallest solution to the same fixed point equation \cite{abraham2015introduction}. In particular, supposing $\lambda s s' > 1$ and $\Eind[L_\infty] <1$ implies that $\Eind[L_\infty]$ is equal to the extinction probability of a Galton--Watson tree with $\mathrm{Poi}(\lambda s s')$ offspring which we denote by $\Prob(\mathrm{Ext}(\mathrm{GW}^{(\lambda ss')}))$. This probability is also equal to the relative size of the giant component in a sparse Erdős--Rényi graph \cite{bollobas1998random}, illustrating another close link between these graphs and the trees we consider.

Aside from the connection to extinction probabilities of Galton--Watson trees, the condition $\Eind[L_\infty] < 1$ is equivalent to an additional property of $(L_d)_d$:
\begin{lemma}\label{lem_nonUI_equivalence}
	One has $\, \Eind[L_\infty] < 1$ if and only if $(L_d)_d$ is not uniformly integrable.
\end{lemma}
\begin{proof}
	Since martingales are uniformly integrable if and only if they converge in $ L^1 $, the equivalence to be shown can be reformulated as
	\[
	(L_d)_d \text{ converges in } L^1 \quad \iff \quad \Eind[L_\infty] = 1.
	\]
	Assuming $ L^1 $-convergence, one has $1 = \Eind[L_d] \to \Eind[L_\infty]$ yielding $ \Eind[L_\infty] = 1$. 
	
	On the other hand, suppose that $\Eind[L_\infty] = 1 = \lim_{d \to \infty} \Eind[L_d]$. Since $L_d \to L_\infty$ almost surely, Scheffé's Lemma also implies convergence in $ L^1 $, which concludes the proof.
\end{proof}

\paragraph{Diagonalization property.} With the recursion formula and the martingale property of $L_d$ at hand, we lack a final tool for a more fine-grained analysis of information-theoretic thresholds. The following stunning result generalizes \cite[Theorem 4]{ganassali2022statistical} to the asymmetric case, affirming that $ L_d $ can be \textit{diagonalized} over the space of unlabeled trees. We refer to this result as the likelihood ratio diagonalization formula since it resembles the Spectral Theorem for matrices, including the orthonormal eigenbasis.
\begin{theorem}\label{thm_diagonalization}
	There exists a set of functions $f^{(\mu)}_{d, \beta}:\mathcal{X}_d \to \R$ with parameters $\mu >0, d\in \N$ and indexed by trees $\beta \in \mathcal{X}_d$ such that for all choices of $s, s' \in [0,1]$, $\lambda >0$, the following \textbf{likelihood ratio diagonalization} formula holds:
	\begin{equation}\label{likhoodratio_bigformula}
		\forall t, t' \in \mathcal{X}_d\, : \, L_d(t, t') = \sum_{\beta\in\mathcal{X}_d} \sqrt{ss'}^{|\beta|-1} f^{(\lambda s)}_{d, \beta}(t) f^{(\lambda s')}_{d, \beta}(t').
	\end{equation}
	Furthermore, the $f^{(\mu)}_{d, \beta}$ have the following properties:
	\begin{itemize}
		\item Constant value for $\beta = \bullet$, the trivial tree:
		\begin{equation}\label{trivialtreevalue}
			\forall \mu >0, \,  \forall t \in \mathcal{X}_d: \quad f^{(\mu)}_{d, \bullet}(t)  = 1.
		\end{equation}
		\item \textbf{First orthogonality property} (w.r.t the Galton--Watson-measure):
		\begin{equation}\label{orthogonaliy_1_GW}
			\forall \mu >0, \, \forall \beta, \beta' \in \mathcal{X}_d: \quad \sum_{t \in \mathcal{X}_d} \mathrm{GW}_d^{(\mu)}(t) \, f^{(\mu)}_{d, \beta}(t) \, f^{(\mu)}_{d, \beta'}(t) \, =\, \indicator{\beta = \beta'}.
		\end{equation}
		\item \textbf{Second orthogonality property} (summing over $\beta$):
		\begin{equation}\label{orthogonaliy_2_betasum}
			\forall \mu >0, \, \forall t, t' \in \mathcal{X}_d: \quad \sum_{\beta \in \mathcal{X}_d} f^{(\mu)}_{d, \beta}(t) \, f^{(\mu)}_{d, \beta}(t') \,=\, \frac{\indicator{t = t'}}{\mathrm{GW}_d^{(\mu)}(t)}.
		\end{equation}
	\end{itemize}
\end{theorem}

This theorem, particularly the orthogonality properties, is central in our strategy for showing Theorem \ref{thm:phase_transition} later. Its proof is a computationally heavy induction, which we defer to appendix \ref{appendix:thm_diagonalization}. 

A first opportunity to see Theorem \ref{thm_diagonalization} in action is by relating the second moment of the likelihood ratio to the number of unlabeled trees: 
\begin{corollary}\label{second_moment}
	The second moment of the likelihood ratio under the independence assumption can be written as
	\[
	\Eindd[L_d^2] =  \sum_{n=1}^{\infty} \Big| \mathcal{X}_{d}^{(n)}\Big| (ss')^{n-1}.
	\]
\end{corollary}
\begin{proof}
	Using the likelihood ratio diagonalization formula (\ref{likhoodratio_bigformula}), we compute
	\begin{align*}
		\Eindd[L_d^2] &= \Eindd\Big[ \sum_{\beta, \beta' \in\mathcal{X}_d} \sqrt{ss'}^{|\beta| + |\beta'|-2} f^{(\lambda s)}_{d, \beta}(t) f^{(\lambda s)}_{d, \beta'}(t)  f^{(\lambda s')}_{d, \beta}(t') f^{(\lambda s')}_{d, \beta'}(t') \Big]\\
		&\overset{(a)}{=} \sum_{\beta, \beta' \in\mathcal{X}_d} \Eindd\Big[ \sqrt{ss'}^{|\beta| + |\beta'|-2}  f^{(\lambda s)}_{d, \beta}(t) f^{(\lambda s)}_{d, \beta'}(t)  f^{(\lambda s')}_{d, \beta}(t') f^{(\lambda s')}_{d, \beta'}(t')\Big]  \\
		&\overset{(b)}{=} \sum_{\beta, \beta' \in\mathcal{X}_d} \sqrt{ss'}^{|\beta| + |\beta'|-2} \Eindd\Big[ f^{(\lambda s)}_{d, \beta}(t) f^{(\lambda s)}_{d, \beta'}(t) \Big]  \Eindd\Big[ f^{(\lambda s')}_{d, \beta}(t') f^{(\lambda s')}_{d, \beta'}(t')\Big]  \\
		&\overset{(c)}{=} \sum_{\beta, \beta' \in\mathcal{X}_d} \sqrt{ss'}^{|\beta| + |\beta'|-2}\indicator{\beta = \beta'}^2  \\
        & = \sum_{\beta \in\mathcal{X}_d} (ss')^{|\beta| -1}\\
		&= \sum_{n=1}^{\infty} \Big| \mathcal{X}_{d}^{(n)} \Big|(ss')^{n-1}.
	\end{align*}
	The equation $(b)$ uses independence of $t$ and $t'$ under $\Pindd$, while $(c)$ exploits the first orthogonality property (\ref{orthogonaliy_1_GW}) in the form
	\[
	\sum_{t \in \mathcal{X}_d} \mathrm{GW}_d^{(\mu)}(t) \, f^{(\mu)}_{d, \beta}(t) \, f^{(\mu)}_{d, \beta'}(t) = \Expe_{t \sim \mathrm{GW}_d^{(\mu)}}\Big[f^{(\mu)}_{d, \beta}(t) \, f^{(\mu)}_{d, \beta'}(t)\Big]  = \indicator{\beta = \beta'}
	\]
	for parameter choices $\mu = \lambda s$ and $\mu = \lambda s'$.
	In $(a)$, we use both of these facts to apply Fubini's Theorem in order to exchange expectation and summation over the countable set $\mathcal{X}_d \times \mathcal{X}_d$. Using Fubini is possible because of $ss' < 1$ and the following application of the Cauchy-Schwartz inequality :
	\begin{align*}
		\Eindd&\Big[\Big|f^{(\lambda s)}_{d, \beta}(t) f^{(\lambda s)}_{d, \beta'}(t)  f^{(\lambda s')}_{d, \beta}(t') f^{(\lambda s')}_{d, \beta'}(t') \Big|\Big]\\ &= \Eindd\Big[\Big|f^{(\lambda s)}_{d, \beta}(t) f^{(\lambda s)}_{d, \beta'}(t) \Big|\Big] \Eindd\Big[\Big|f^{(\lambda s')}_{d, \beta}(t') f^{(\lambda s')}_{d, \beta'}(t')  \Big|\Big] \\
		&\leq \sqrt{\Eindd\Big[ f^{(\lambda s)}_{d, \beta}(t)^2 \Big] \Eindd\Big[ f^{(\lambda s)}_{d, \beta'}(t)^2 \Big] \Eindd\Big[ f^{(\lambda s')}_{d, \beta}(t')^2  \Big] \Eindd\Big[ f^{(\lambda s')}_{d, \beta'}(t')^2 \Big]} = 1.
	\end{align*}
	This concludes the corollary's proof.
\end{proof}
With a full toolbox of results about $L_d$, we are now ready to present our main theorems about tree correlation testing.	

\subsection{One-sided testability equivalences}

Motivated by the question whether there exist one-sided tests for the tree correlation testing problem (\ref{testing_problem_2}), we start by reformulating it in two steps: First, make the question uniquely about likelihood ratio tests; second, relate the feasibility to more amenable information-theoretic quantities. It turns out that testability is closely linked to the Kullback--Leibler divergence between $ \Pcorrd $ and $ \Pindd $ which is defined as
\begin{equation*}
	\mathrm{KL}_d := \mathrm{KL}(\Pcorrd \Vert \Pindd) = \Ecorrd[\log L_d].
\end{equation*}
This quantity is \textit{nondecreasing} as a function of the depth $d$. To see this, let $\mathcal{F}_d$ be the sigma-field containing information up to depth $d$ and apply Jensen's inequality with the convex function $x \mapsto x \log(x)$ as well as the martingale property $ \Eind_{d+1}\Big[ L_{d+1}\suchthat \mathcal{F}_d \Big] = L_{d}  $ to obtain
\begin{align*}
	\mathrm{KL}_{d+1} &= \Eind_{d+1}\Big[\Eind_{d+1}\Big[L_{d+1} \log(L_{d+1}) \suchthat \mathcal{F}_d \Big]\Big]\\ &\geq \Eind_{d+1} \, \Big[\Eind_{d+1} \Big[L_{d+1} \suchthat \mathcal{F}_d \Big]  \log\Big(\Eind_{d+1}\Big[ L_{d+1}\suchthat \mathcal{F}_d \Big]\Big)\, \Big] \\ &= \Eind_{d}[L_d \log(L_d)] = \mathrm{KL}_d.
\end{align*}
Consequently, the sequence $(\mathrm{KL}_d)_d$ has a (potentially infinite) limit $\mathrm{KL}_\infty \in [0, \infty]$.

The following theorem achieves the desired reformulation of one-sided testability by providing equivalent conditions that relate $\mathrm{KL}_d$, $L_d$, and the parameters $\lambda, s $ and $s'$. 

\begin{theorem}\label{thm_equivalences}
	In the tree correlation testing problem (\ref{testing_problem_2}), the following are equivalent:
	\begin{enumerate}[label=(\alph*)]
		\item There exist one-sided tests $\mathcal{T}_d$ to decide $\Pindd$ vs. $ \Pcorrd $,
		\item There is a sequence of thresholds $\theta_d \to \infty $ such that ${\Pindd(L_d > \theta_d) \to 0}$ and $\allowbreak{\liminf_{d\to\infty} \Pcorrd(L_d > \theta_d) > 0}$, 
		\item The $ \Pind $--martingale $(L_d)_d$ is not uniformly integrable,
		\item $\lambda s s' > 1$ and $KL_\infty = \infty$,
		\item $\lambda s s' > 1$ and $\,  \Pcorr \Big(\liminf_d (\lambda s s')^{-d} \log(L_d) \geq C\Big) \geq 1 - \Prob(\mathrm{Ext}(\mathrm{GW}^{(\lambda s s')})) $ for a constant $C>0$ only depending on $(\lambda, s, s')$ and $\mathrm{Ext}(\mathrm{GW}^{(\lambda s s')})$ denoting the extinction event of the Galton--Watson process with offspring $\mathrm{Poi}(\lambda ss')$.
	\end{enumerate}
\end{theorem}
This result was first presented in the symmetric model in \cite[Theorem 1]{luca} and was extended to asymmetric correlated graphs in \cite[Theorem 1]{maier2023asymmetric}. Here, we provide a complete and rigorous proof that was partially omitted in \cite{maier2023asymmetric}.

\begin{proof}
	We start by showing the equivalence (b) $\Leftrightarrow$ (c) followed by (a) $\Leftrightarrow$ (b) and concluding with the circular implications (b)  $\Rightarrow$ (d) $\Rightarrow$ (e) $\Rightarrow$ (b). We defer longer technical passages to Appendix \ref{appendix:thm_equivalences} for better readability. 
	
	\paragraph{(b) $\Leftrightarrow$ (c)} Given (b), there exists a sequence $\theta_d \to \infty$ such that $\liminf_{d \to \infty} \Pcorrd(L_d > \theta_d) > 0$. Reformulating this inequality, there exist $\varepsilon > 0$ and $d_0(\varepsilon) \in \N$ such that for all $d \geq d_0(\varepsilon)$,
	\begin{equation*}
		\Pcorrd(L_d > \theta_d) = \Eind\Big[L_d \indicator{L_d > \theta_d}\Big] \geq \varepsilon.
	\end{equation*}
	This contradicts the definition of uniform integrability of $(L_d)_d$. 
	
	For the opposite direction, suppose (c) to be true and invoke Lemma \ref{lem_nonUI_equivalence} to translate it to $\Eind[L_\infty] < 1$. From this, one can construct a sequence $(\theta_d)_d$ which goes to infinity and satisfies the desired properties, i.e.
	\begin{equation*}
		\liminf_{d\to\infty} \Pcorr(L_d > \theta_d)  > 0 \quad \text{and} \quad \Pind(L_d > \theta_d) \xrightarrow[d\to\infty]{} 0.
	\end{equation*}
	For details on the construction of $(\theta_d)_d$, we refer the reader to Appendix \ref{appendix:thm_equivalences}.
	
	\paragraph{(a) $\Leftrightarrow$ (b)} While (a) follows from (b) by taking $\mathcal{T}_d := \indicator{L_d > \theta_d}$, the other implication is more involved. Suppose that there exists a sequence of one-sided tests $ \mathcal{T}_d $ fulfilling
	\begin{equation*}
		\Pind(\mathcal{T}_d = 1) =: \alpha_d \to 0 \quad \text{and} \quad \liminf_{d \to \infty}\underbrace{\Pcorr(\mathcal{T}_d = 1)}_{=: \beta_d} > 0.
	\end{equation*}
	The Neyman-Pearson Lemma as presented in \cite{hallin2006neyman} then yields an optimal test sequence $\hat{\mathcal{T}}_d$ of the form
	\begin{align*}
		\hat{\mathcal{T}}_d = \begin{cases}
			1 &\text{ if } L_d > \hat{\theta}_d,\\
			\indicator{U < \gamma_d} & \text{ for } U \sim \mathrm{Unif}([0,1]) \text{ and } \gamma_d \in [0,1] \text{ if } L_d = \hat{\theta}_d, , \\
			0  &\text{ if } L_d \leq \hat{\theta}_d
		\end{cases}
	\end{align*}
	with type I error equal to that of $\mathcal{T}_d$, i.e.,
	\begin{equation}\label{eq_alphatozero}
		\alpha_d = \Pind(\hat{\mathcal{T}}_d = 1) = \Pind(L_d > \hat{\theta}_d) + \gamma_d \Pind(L_d = \hat{\theta}_d).
	\end{equation}
	We differentiate between two cases depending on whether $\hat{\theta}_d \to \infty$ or not.\\
	Beginning with the case $\hat{\theta}_d \not \to \infty$  we show in Appendix \ref{appendix:thm_equivalences} that
	\begin{equation}\label{implication_to_be_shown}
		(\hat\theta_d)_d \text{ is bounded}\implies \Eind[L_\infty] < 1.
	\end{equation}
	Since $\Eind[L_\infty] < 1$ is equivalent to (c) which we have already shown to imply (b), this concludes the first case.\\
	Assuming now that $\hat{\theta}_d \to \infty$, we set $\theta_d := \hat{\theta}_d - 1$ and use Markov's inequality to obtain
	\begin{equation*}
		\Pindd\Big(L_d > \theta_d\Big) \leq \frac{\Eindd[L_d]}{\theta_d} = \frac{1}{\hat\theta_d - 1} \xrightarrow[d \to \infty]{} 0.
	\end{equation*}
	This shows that the likelihood ratio test has vanishing type-I-error. It also has significant power since
	\begin{equation*}
		\Pcorrd(L_d \geq \theta_d) \geq \Pcorrd(L_d \geq \hat{\theta}_d) \geq \Pcorrd(\mathcal{T}_d = 1) = \beta_d \implies \liminf_{d \to \infty}\Pcorrd(L_d \geq \theta_d) > 0.
	\end{equation*}
	Hence, (b) holds true, concluding the equivalence we wanted to show.
	
	\paragraph{(b) $\Rightarrow$ (d)} Supposing that (b) is true, we refer to step 5 in the proof of Theorem 1 in \cite{luca} for a demonstration of $\mathrm{KL}_\infty = \infty$. This proof completely ignores the symmetry assumption; it is therefore independent of $s, s'$ and directly transferrable to our asymmetric setting.
	
	In order to show $\lambda s s' > 1$, we assume the contrary, i.e. $\lambda s s' \leq 1$. A standard result from Galton--Watson tree theory (see, for instance, \cite{abraham2015introduction}) states that in this case, $1$ is the only fixed point in the interval $[0,1]$ of the probability generating function of $ \mathrm{Poi}(\lambda s s')$. As seen in section \ref{section_likhood_ratio_properties} (``martingale property''), the value $\Eind[L_\infty]$ fulfills the same fixed point equation, leading to $\Eind[L_\infty] = 1$. Consequently, one has trivial convergence of the means $\Eind[L_d] = 1 \to 1 = \Eind[L_\infty]$, which we combine with the almost sure convergence $L_d \to L_\infty$ and Scheffé's Lemma to obtain $L^1$-convergence of $L_d$  to $L_\infty$. This, in turn, yields that $(L_d)_d$ is a flat martingale, hence uniformly integrable. Since we have already shown (b) $\Leftrightarrow$ (c), this contradicts one-sided testability and lets us conclude that the assumption  $\lambda s s' \leq 1$ must be wrong.
	
	\paragraph{(d) $\Rightarrow$ (e)}  This implication is long and technical, requiring additional results about Galton--Watson trees. We defer it to appendix \ref{appendix:thm_equivalences}.
	
	\paragraph{(e) $\Rightarrow$ (b)} The assumption $\lambda s s' > 1$ implies $ 1 - \Prob\Big( \mathrm{Ext}\Big(\mathrm{GW}^{(\lambda ss')}\Big)\Big) > 0$  and hence
	\begin{equation*}
		\Pcorr \Big(\liminf_{d \to \infty} (\lambda s s')^{-d} \log(L_d) \geq C\Big) \, > \, 0.
	\end{equation*}
	Setting $\theta_d := \exp\Big(C \, (\lambda ss')^d\Big)$, one has $\theta_d \to \infty$ and Markov's inequality implies vanishing type-I-error of the test $\indicator{L_d > \theta_d} \, :$
	\begin{equation*}
		\Pind(L_d \geq \theta_d) \leq \frac{\Eind[L_d]}{\theta_d} = \frac{1}{\theta_d} \xrightarrow[d\to\infty]{} 0.
	\end{equation*}
	To show significant power, we use Fatou's Lemma to obtain
	\begin{align*}
		\liminf_{d \to \infty} \Pcorr(L_d \geq \theta_d) &\geq  \Pcorr\Big(\liminf_{d \to \infty} \{L_d \geq \theta_d\}\Big) \\
		&= \Pcorr\Big(\exists D \, \forall d \geq D: L_d  \geq \exp\Big(C \, (\lambda ss')^d \Big)\Big) \\
		&= \Pcorr\Big(\liminf_{d \to \infty} (\lambda ss')^{-d}\log(L_d) \geq C\Big) > 0.
	\end{align*}
	This implies (b) and concludes the proof of all equivalences we wanted to show.
\end{proof}

\subsection{Sharp information-theoretic limits}
The previous section translates feasibility of tree correlation testing to concrete statistical properties. Especially (a) $\Leftrightarrow$ (d) of Theorem \ref{thm_equivalences} lets us focus on the Kullback--Leibler between $\Pcorrd$ and $\Pindd$. The following theorem relates the convergence of $(\mathrm{KL}_d)_d$ directly to the parameters $\lambda, s$ and $s'$, characterizing a sharp phase transition in asymmetric tree correlation testing around Otter's constant.
\begin{theorem}\label{thm:phase_transition}
	In the tree correlation testing problem (\ref{testing_problem_2}), parameterized by $s, s' \in (0,1]$ and $\lambda >0$, there is a feasibility phase transition around Otter's constant $\alpha \approx 0.3383$  which is characterized as follows:
	\begin{enumerate}[label = (\roman*)]
		\item If $s s' \leq \alpha$, then $ \sup_{d\in\N} \mathrm{KL}_d < \infty$ and tree correlation testing is impossible.
		\item If $ss' > \alpha$  and $\lambda > \lambda_0\, $ for a threshold $ \, \lambda_0 = \lambda_0(s, s')$ coming from a gaussian approximation argument, then $\mathrm{KL}_d \to \infty$ and tree correlation testing is possible.
	\end{enumerate}
\end{theorem}
The conclusions in (i) and (ii) about tree correlation testing are immediate consequences of Theorem \ref{thm_equivalences}. In particular, one needs to require $\lambda_0 ss' > 1$  for the conclusion in (ii) to be true, adding a requirement on the magnitude of $\lambda_0$.
\begin{remark}
	For large enough $\lambda$, the only parameter determining feasibility of the testing problem is the product $ss'$. This has the effect that one density parameter $s$ of the first tree can compensate for lower density $s'$ in the second tree. We refer to the Conclusion Section \ref{section:conclusion} for the implications of this result on graph alignment.
\end{remark}
\begin{remark}[On the size of $\lambda_0$]
	The sharpness statement in Theorem~\ref{thm:phase_transition} holds only under the additional condition $\lambda>\lambda_0$, which arises from the gaussian approximation argument later in this section. It is natural to conjecture that $\lambda s s' > 1$ is sufficient for $\mathrm{KL}_d \to \infty$, even more so in the light of Corollary \ref{second_moment} which states that $\Eind[L_d^2]$ only depends on $ss'$, not $\lambda$. However, the only link we can establish between $\mathrm{KL}_d$ and $\Eind[L_d^2]$ goes through gaussian approximation, and we do not have a principled reason to expect such a link to hold beyond that regime. Determining the minimal value of $\lambda_0$ therefore remains open.
\end{remark}

For the remainder of this section, we prove Theorem \ref{thm:phase_transition}. Showing (i) is straightforward due to our preparatory work, the proof of (ii) requires more involved techniques.
\begin{proof}[Proof of (i) in Theorem \ref{thm:phase_transition}.]
Let $ss' \leq \alpha$. We start by upper bounding the Kullback--Leibler by using Jensen's inequality:
	\begin{align*}
		\mathrm{KL}(\Pcorrd \, \Vert \, \Pindd )& = \Eindd\Big[L_d\log(L_d)\Big]  = \Ecorrd\Big[\log(L_d)\Big]\\
		&\leq \log\Big(\Ecorrd\Big[L_d\Big]\Big) = \log\Big(\Eindd[L_d^2]\Big).
	\end{align*}
	This second moment is characterized in Corollary \ref{second_moment}. Combining that result with $| \mathcal{X}_{d}^{(n)}|\leq | \mathcal{X}_{n-1}^{(n)}|$ yields
	\[
	\Eindd[L_d^2] = \sum_{n=1}^{\infty} \Big| \mathcal{X}_{d}^{(n)}\Big| (ss')^{n-1} \leq (ss')^{-1}\sum_{n=1}^{\infty} \Big| \mathcal{X}_{n-1}^{(n)}\Big| (ss')^{n}.
	\]
	Due to the asymptotic formula for $| \mathcal{X}_{n-1}^{(n)}| $ from Proposition \ref{Otter}, we have
	\[
	\sum_{n=1}^{\infty} \Big| \mathcal{X}_{n-1}^{(n)}\Big|  (ss')^{n} < \infty \quad \iff \quad \sum_{n=1}^{\infty} C \, n^{-\sfrac{3}{2}} \frac{(ss')^{n}}{\alpha^n} < \infty
	\]
	where finiteness on the right follows from $ss'\leq \alpha$. Consequently, $\mathrm{KL}(\Pcorrd \, \Vert \, \Pindd)$ is bounded above by the logarithm of a finite value, which concludes the proof. 
\end{proof}
The use of Jensen's inequality in the above proof prevents us from using the same argument to show  $\mathrm{KL}\Big(\Pcorrd \, \Vert \, \Pindd \Big) = \infty$ in the case where $ss' > \alpha$. This direction requires some additional tools, which we will introduce next.
\begin{proof}[Proof of (ii) in Theorem \ref{thm:phase_transition}.]
	Instead of being true for any value of $\lambda$, we recall that point (ii) of Theorem \ref{thm:phase_transition} only holds for all $\lambda$ surpassing a certain threshold $\lambda_0$. This phenomenon has its roots in our proof technique: We will first examine the asymptotic behavior of $\mathrm{KL}_d$ when letting $\lambda \to \infty$ and leaving $d$ fixed. This is done via a Gaussian approximation argument, which is introduced in the following lemma.\\
	This lemma will not be precisely enough to show the desired statement, which requires some refined analysis afterward. Nonetheless, it contains the key idea of the proof. 
	\begin{lemma}[gaussian approximation]\label{gaussian_approx}
		Let $d \in \N$ and recall that $\mathcal{X}_{d+1} = \N^{\mathcal{X}_d}$ by the subtree tuple identification. There exists a pair of functions
		\[
		(y, y'): \mathcal{X}_{d+1}\times \mathcal{X}_{d+1} \to \R^{\mathcal{X}_d} \times  \R^{\mathcal{X}_d}, \, (t, t')  \mapsto \Big(y_\beta (t), y'_\beta(t') \Big)_{\beta \in \mathcal{X}_d}
		\]
		with the property that the above map is affine and bijective. The law of the random vector $\big(y(t), y'(t')\big)$ shall be denoted either as $\mathcal{L}(y, y')$ in the case $(t, t') \sim \Pcorrdplus$, or as $\mathcal{L}(y) \otimes \mathcal{L}(y')$ in the case $(t, t') \sim \Pinddplus$.
		
		One has
		\[
		\mathrm{KL}(\Pcorrdplus \, \Vert \, \Pinddplus ) = \mathrm{KL}\Big(\mathcal{L}(y, y') \, \Vert \, \mathcal{L}(y) \otimes \mathcal{L}(y')\Big)
		\]
		and that $\mathcal{L}(y, y')$ converges weakly towards the joint law of a pair of random variables
		\[
		\mathcal{L}(y, y') \xrightarrow[\lambda \to \infty]{w} \mathcal{L}(z, z').
		\]
		The pair $(z, z') = \Big( (z_\beta)_{\beta \in \mathcal{X}_d}, (z'_\beta)_{\beta \in \mathcal{X}_d}\Big)$ is a gaussian vector of infinite dimension with zero mean and covariances defined by
		\[
		\Expe[z_\beta z_\gamma] = \Expe[z'_\beta z'_\gamma] = \indicator{\beta = \gamma}, \quad \Expe[z_\beta z'_\gamma] = \sqrt{ss'}^{|\beta|}\indicator{\beta = \gamma} \quad \text{for all } \beta, \gamma \in \mathcal{X}_d.
		\]
	\end{lemma}
The proof of this lemma is rather technical and exploits the orthogonality properties from Theorem \ref{thm_diagonalization}. The details are in Appendix \ref{appendix:gaussian_approx}, in particular how the vector $(y, y')$ is defined and how Levy's Continuity Theorem implies its weak convergence to $(z, z')$.
	
	For the remainder of this section, let $ss' \geq \alpha$ and recall that our goal is to show $\mathrm{KL}_d \xrightarrow[d\to\infty]{} \infty$ for sufficiently large $\lambda$. With Lemma \ref{gaussian_approx} at hand, we can take a first step towards this goal and establish a lower bound on $\mathrm{KL}$ for large $\lambda$:
	\begin{lemma}\label{liminf_lower_bd}
		For any depth $d \in \N$, one has the following asymptotic lower bound on $\mathrm{KL}_{d+1}$:
		\[
		\liminf_{\lambda \to \infty} \mathrm{KL}(\Pcorrdplus \, \Vert \, \Pinddplus ) \geq  \sum_{n=1}^{\infty} \Big| \mathcal{X}_{d+1}^{(n)}\Big| (ss')^{n-1}.
		\]
	\end{lemma}
	\begin{proof}
		The key idea for a lower bound on $\mathrm{KL}(\Pcorrdplus \, \Vert \, \Pinddplus )$ is the \textit{lower semi-continuity of the Kullback--Leibler divergence} (see, for instance, \cite{mackay2003information}): For two sequences of weakly converging random variables 
		\[
		X_\lambda \xrightarrow[\lambda \to \infty]{w} X \quad \text{and}\quad Y_\lambda \xrightarrow[\lambda \to \infty]{w} Y,
		\]
		one has
		\[
		\liminf_{\lambda \to \infty} \mathrm{KL}(\mathcal{L}(X_n) \, \Vert \, \mathcal{L}(Y_n) ) \, \geq \, \mathrm{KL}( \mathcal{L}(X) \, \Vert \, \mathcal{L}(Y) ).
		\]
		Using Lemma \ref{gaussian_approx}, we therefore obtain
		\begin{align*}
			\liminf_{\lambda \to \infty} \mathrm{KL}(\Pcorrdplus \, \Vert \, \Pinddplus ) &= \liminf_{\lambda \to \infty} \mathrm{KL}\Big(\mathcal{L}(y, y') \, \Vert \, \mathcal{L}(y) \otimes \mathcal{L}(y')\Big)\\ &\geq \mathrm{KL}\Big(\mathcal{L}(z, z') \, \Vert \, \mathcal{L}(z) \otimes \mathcal{L}(z')\Big).
		\end{align*}
		The right-hand side can be explicitly computed after some observations: Recall that the covariance structure of $(z, z')$ is given by
		\[
		\Expe[z_\beta z_\gamma] = \Expe[z'_\beta z'_\gamma] = \indicator{\beta = \gamma}, \quad \Expe[z_\beta z'_\gamma] = \sqrt{ss'}^{|\beta|}\indicator{\beta = \gamma} \quad \text{for all } \beta, \gamma \in \mathcal{X}_d.
		\]
		This leads us to observe, on the one hand, that an infinite Gaussian vector of law $\mathcal{L}(z) \otimes \mathcal{L}(z')$ has the infinite covariance matrix $\mathrm{diag}(1, 1, 1, ...)$. On the other hand, for a vector $Z$ of law $\mathcal{L}(z, z')$, we choose an arbitrary numbering of the trees in $\mathcal{X}_d$, namely $\beta_1, \beta_2, \beta_3,...$ and suppose that the coordinates are ordered as
		\[
		Z = (z_{\beta_1}, z'_{\beta_1}, z_{\beta_2}, z'_{\beta_2}, z_{\beta_3}, z'_{\beta_3}, \dots).
		\]
		Note that this ordering does not affect the Kullback--Leibler divergence since applying a permutation to the covariance matrix is a similarity transformation.\\
		Observe that the covariance of $Z$ is an infinite block diagonal matrix with the blocks
		\[
		\begin{pmatrix}
			1 & \sqrt{ss'}^{|\beta_i |}\\
			\sqrt{ss'}^{|\beta_i |} & 1
		\end{pmatrix}
		\quad \text{ for } i = 1, 2, 3, ... \, .
		\]
		This enables us to compute
		\begin{align*}
			\mathrm{KL}\Big(\mathcal{L}(z, z') \, \Vert \, \mathcal{L}(z) \otimes \mathcal{L}(z')\Big) &= \mathrm{KL}\Big( \mathcal{N}(0, \mathrm{Cov}(Z)) \, \Vert \, \mathcal{N}(0, I)\Big)\\
			&= -\frac{1}{2}\log\Big(\mathrm{det}(\mathrm{Cov}(Z))\Big)
		\end{align*}
		where the last equality is a simple computation, see for instance \cite{MathSE_MutualInfo_2024}.\\
		One can deduce $\mathrm{det}\big(\mathrm{Cov}(Z)\big) = \prod_{i=1}^\infty (1 - (ss')^{|\beta_i|})$ from the block diagonal structure of the covariance matrix. Consequently,
		\begin{align*}
			\liminf_{\lambda \to \infty} \mathrm{KL}(\Pcorrdplus \, \Vert \, \Pinddplus ) \geq - \frac{1}{2} \log \Big(\prod_{\beta \in \mathcal{X}_d} (1 - (ss')^{|\beta|})\Big) = \frac{1}{2} \log \Big(\prod_{\beta \in \mathcal{X}_d} \frac{1}{1 - (ss')^{|\beta|}}\Big).
		\end{align*}
		Since $ss' < 1$ and $\mathcal{X}_d$ is countable, we may develop the geometric series
		\begin{align*}
			\prod_{\beta \in \mathcal{X}_d} \frac{1}{1 - (ss')^{|\beta|}} = \prod_{\beta \in \mathcal{X}_d} \sum_{\gamma_\beta = 0}^\infty (ss')^{\gamma_\beta |\beta|}   = \sum_{(\gamma_\beta)_{\beta} \in \N^{\mathcal{X}_{d}}} (ss')^{\sum_\beta \gamma_\beta |\beta|} = \sum_{\gamma \in \mathcal{X}_{d+1}} (ss')^{|\gamma| - 1}.
		\end{align*}
		where for the last equality we recall that  $|\gamma|  = 1 + \sum_\beta \gamma_\beta |\beta|$ for any $\gamma = (\gamma_\beta)_{\beta \in \mathcal{X}_d} \in \mathcal{X}_{d+1}$. Developing the right-hand expression like in the proof of Corollary \ref{second_moment} yields
		\[
		\sum_{\gamma \in \mathcal{X}_{d+1}} (ss')^{|\gamma| - 1} = \sum_{n = 1}^\infty  \sum_{\substack{\gamma \in \mathcal{X}_{d+1},\\ |\gamma| = n}} (ss')^{n - 1} = \sum_{n=1}^{\infty} \Big| \mathcal{X}_{d+1}^{(n)}\Big| (ss')^{n-1}.
		\]
		This concludes the lemma's proof. 
	\end{proof}
	Letting $d \to \infty$ in Lemma \ref{liminf_lower_bd} yields
	\[
	\lim_{d\to\infty} \liminf_{\lambda \to \infty} \mathrm{KL}(\Pcorrdplus \, \Vert \, \Pinddplus ) \geq  \lim_{d\to\infty}\sum_{n=1}^{\infty} \Big| \mathcal{X}_{d}^{(n)}\Big|  (ss')^{n-1} = \sum_{n=1}^{\infty} \Big| \mathcal{X}_{n-1}^{(n)}\Big|  (ss')^{n-1}
	\]
	where the right hand side diverges due to $ss'\geq \alpha$ and the asymptotics of $\Big| \mathcal{X}_{n-1}^{(n)}\Big| $ from Proposition \ref{Otter}. If the limits in $d$ and $\lambda$ were swapped, this would give the desired results, but we are not allowed to do so. Instead, we need another proof strategy, which we sketch here before providing the details:
	\begin{enumerate}
		\item From the fact that $\liminf_\lambda \mathrm{KL}(\Pcorrdplus \, \Vert \, \Pinddplus )$ diverges for $d \to \infty$, one can derive that there is a depth $d_0$ and a test $\mathcal{T}_{d_0}$ fulfilling
		\begin{equation}\label{onesided_d0}
			\Pcorr_{d_0}(\mathcal{T}_{d_0}= 1) \geq \beta \quad \text{and} \quad \Pind_{d_0}(\mathcal{T}_{d_0} = 1) \leq \varepsilon.
		\end{equation}
		\item  This test can then be amplified in the next generation $d_0+1$, i.e., there exists a test $\mathcal{T}_{d_0+ 1}$ based on $\mathcal{T}_{d_0}$ such that for the same constants as before
		\begin{equation}\label{onesided_d0plus1}
			\Pcorr_{d_0 + 1}(\mathcal{T}_{d_0+1}= 1) \geq \beta \quad \text{and} \quad \Pind_{d_0 + 1}(\mathcal{T}_{d_0+ 1} = 1) \leq \frac{\varepsilon}{2}.
		\end{equation}
		\item The test amplification strategy is independent of the generation $d$, meaning we can propagate the bound arbitrarily often. This leads to a sequence of tests $(\mathcal{T}_{d})_d$ which share the lower bound $\beta>0$ on their power and have vanishing type-1-error since $\Pind_{d}( \mathcal{T}_{d}= 1) \leq \varepsilon/2^{-(d - d_0)}$. In other words, we obtain a sequence of one-sided tests.
		\item By Theorem \ref{thm_equivalences}, this is equivalent to the divergence of $\mathrm{KL}(\Pcorrdplus \, \Vert \, \Pinddplus )$, which concludes the proof of Lemma \ref{gaussian_approx}. 
	\end{enumerate}
	While points 3 and 4 of this sketch are rigorous, the points 1 and 2 require more detail. This is the content of the following two lemmata, which are proved afterward.
	\begin{lemma}[point 1 of the proof sketch]\label{first_final_lemma}
		Let $\beta \in (0, \sfrac{2}{15})$ and $\varepsilon > 0$. Then, there exists $\lambda_1 > 0$ and $d_0 \in \N$ such that for all $\lambda \geq \lambda_1$ there is a test $\mathcal{T}_{d_0}: \mathcal{X}_{d_0} \times \mathcal{X}_{d_0} \to \{0,1\}$ fulfilling
		\[
		\Pcorr_{d_0}(\mathcal{T}_{d_0} = 1) \geq \beta \quad \text{and} \quad \Pind_{d_0}(\mathcal{T}_{d_0} = 0) \leq \varepsilon.
		\]
	\end{lemma}
	\begin{lemma}[point 2 of the proof sketch]\label{second_final_lemma}
		Given $\beta \in (0, 1)$, there exist constants $\varepsilon > 0$ and $\lambda_0 > 0$ such that the following implication holds:\\
		If for all $\lambda \geq \lambda_0$ and any depth $d \in \N$, there exists a test $\mathcal{T}_{d}: \mathcal{X}_{d} \times \mathcal{X}_{d} \to \{0,1\}$ with
		\[
		\Pcorrd(\mathcal{T}_{d}= 1) \geq \beta \quad \text{and} \quad \Pindd(\mathcal{T}_{d} = 1) \leq \varepsilon,
		\]
		then one can define another test $\mathcal{T}_{d+1}: \mathcal{X}_{d+1} \times \mathcal{X}_{d+1} \to \{0,1\}$ fulfilling the amplified test conditions
		\[
		\Pcorrdplus(\mathcal{T}_{d+1}= 1) \geq \beta \quad \text{and} \quad \Pinddplus(\mathcal{T}_{d+1} = 1) \leq \frac{\varepsilon}{2}
		\]
		for all $\lambda \geq \lambda_0$. 
	\end{lemma}
	To see how these results imply points 1 and 2 from the above proof sketch, fix a constant $\beta \in (0, 2/15)$ and let $(\varepsilon, \lambda_0)$ be the corresponding quantities from Lemma \ref{second_final_lemma}. By Lemma \ref{first_final_lemma}, there exist $\lambda_1$ and $d_0$  such that (\ref{onesided_d0}) holds for all $\lambda \geq \lambda_1$. Hence, for $\lambda \geq \max\{\lambda_0, \lambda_1\}$, both lemmata can be applied. \\
	Since (\ref{onesided_d0}) is nothing but the implication assumption in Lemma \ref{second_final_lemma} for $d = d_0$, this lemma generates another test at generation $d_0 + 1$ fulfilling the amplified test conditions presented in (\ref{onesided_d0plus1}).  Point (3) of the proof sketch then iterates the utilization of Lemma \ref{second_final_lemma}, which concludes the proof of Theorem \ref{thm:phase_transition} (ii).
\end{proof}

All that remains to do now is show both lemmata \ref{first_final_lemma} and \ref{second_final_lemma}. 

The proof of Lemma \ref{first_final_lemma} can be copied line-by-line from \cite[Appendix B.1]{ganassali2022statistical}, up to notational differences and the fact that the initial condition changes from $s > \sqrt{\alpha}$ to the asymmetric counterpart $ss' > \alpha$. We refer the reader to the original source.

The same cannot be said about the proof of Lemma \ref{second_final_lemma}, which may be of particular interest since it shows how to obtain an amplified test $\mathcal{T}_{d+1}$ from a lower-depth test $\mathcal{T}_d$.
\begin{proof}[Proof of Lemma \ref{second_final_lemma}]
	Fix a constant $\beta \in (0,1)$ and a depth $d$. Next, suppose that there exist $\varepsilon > 0, \, \lambda_0 > 0$ and a test $\mathcal{T}_{d}: \mathcal{X}_{d} \times \mathcal{X}_{d} \to \{0,1\}$ such that for all $\lambda \geq \lambda_0$, 
	\[
	\Pcorrd(\mathcal{T}_{d}= 1) \geq \beta \quad \text{and} \quad \Pindd(\mathcal{T}_{d} = 1) \leq \varepsilon.
	\]
	To find an amplified test $\mathcal{T}_{d+1}$, we define the value
	\[
	Z(t, t'):= \sum_{\substack{(\tau, \tau') \in \mathcal{X}_d^2,\\ \mathcal{T}_d(\tau, \tau') = 1}} \Big(N_\tau - \lambda s \, \mathrm{GW}_d^{(\lambda s)} (\tau)\Big) \Big(N'_{\tau'} - \lambda s' \,\mathrm{GW}_d^{(\lambda s')} (\tau') \Big)
	\]
	where $t = (N_\tau)_\tau$ and $t' = (N'_{\tau'})_{\tau'}$ form a pair of trees in $\mathcal{X}_{d+1}$. For some threshold $\xi$ to be determined, we derive the test
	$
	\mathcal{T}_{d+1}(t, t') := \indicator{Z(t, t') \geq \xi}
	$
	which must fulfill the two requirements
	\begin{equation}\label{testreq_1_c}
		\Pcorrdplus(\mathcal{T}_{d+1}= 1)  = \Pcorrdplus(Z(t, t') \geq \xi) \geq \beta
	\end{equation}
	and
	\begin{equation}\label{testreq_2_eps}
		\Pinddplus(\mathcal{T}_{d+1}= 1)  = \Pinddplus(Z(t, t') \geq \xi) \leq \frac{\varepsilon}{2}.
	\end{equation}
	The technical tool for showing these equations is the following lemma, whose proof can be found in Appendix \ref{appendix:Z_moments}.
	\begin{lemma}[moments of $Z$]\label{Z_moments}
		The random variable $Z := Z(t, t')$ has the following moments under $\Pind$ and $\Pcorr$:
		\begin{enumerate}[label=(\alph*)]
			\item $\Einddplus\Big[Z\Big] = 0$,
			\item $\Ecorrdplus\Big[Z\Big] = \lambda s s' \,\Pcorr_d\Big(\mathcal{T}_d = 1\Big)$,
			\item $\Einddplus\Big[Z^4\Big] \leq 36 \, (\lambda^2 s s')^2 \,\Pind_d\Big(\mathcal{T}_d = 1\Big)^2 + 13 \lambda^3 s s' \, \Pind_d\Big(\mathcal{T}_d = 1\Big)$,
			\item $\mathrm{Var}_{d+1}^\text{corr}\Big[Z\Big] \leq \lambda^2 s s'\Big(1 + ss' \Big) \,\Pind_d\Big(\mathcal{T}_d= 1\Big)  +  \Ecorrdplus[Z] $.
		\end{enumerate}
	\end{lemma}
	\textbf{To show equation (\ref{testreq_2_eps})}, we simply combine Markov's inequality with point (c) and the hypothesis $\Pind_d\Big(\mathcal{T}_d = 1\Big) \leq \varepsilon$, which yields
	\[
	\Pinddplus(Z(t, t') \geq \xi) \leq \frac{1}{\xi^4} \Einddplus\Big[Z(t, t')^4\Big] \leq \frac{1}{\xi^4} \Big[36 \, (\lambda^2 s s')^2  \, \Pind_d\Big(\mathcal{T}_d = 1\Big) \, \varepsilon  + 13 \lambda^3 s s' \, \varepsilon\Big].
	\]
	Since $A + B \leq 2 \max\Big\{A, B\Big\}$, choosing $\xi^4 \geq 2 \, \max\Big\{ 72 (\lambda^2 s s')^2 \, \Pind_d\Big(\mathcal{T}_d = 1\Big),\,  26 \lambda^3 s s' \Big\}$ lets us bound the right hand side by $\varepsilon/2$ which implies $ \Pinddplus(Z(t, t') \geq \xi) \leq \varepsilon/2 $.\\
	Consequently, setting $\xi = \max\Big\{ 4\lambda \sqrt{s s'} \, \Pind_d\Big(\mathcal{T}_d = 1\Big)^{1/4},\,  3 \lambda^{3/4} \sqrt[4]{s s'} \Big\}$ implies (\ref{testreq_2_eps}). 
	
	\textbf{To show equation (\ref{testreq_1_c})}, we first assume $\xi \leq \Ecorrdplus\Big[Z\Big]  / 2$ and compute
	\begin{align*}
		\Pcorrdplus(Z \leq \xi) \,  & \leq \, \Pcorrdplus\Big(Z \leq \Ecorrdplus\Big[Z\Big]  / 2 \Big) \\
		& \leq \,  \Pcorrdplus\Big( \Big|Z - \Ecorrdplus\Big[Z\Big] \Big| \geq \frac{1}{2}\Ecorrdplus\Big[Z\Big] \Big)\\
		& \leq \, 4\,  \frac{\mathrm{Var}_{d+1}^\text{corr}\Big[Z\Big]}{\Ecorrdplus\Big[Z\Big]^2} \\
		& \leq \,  4 \, \frac{\lambda^2 s s'\Big(1 + ss' \Big) \,\Pind_d\Big(\mathcal{T}_d= 1\Big)  +  \lambda s s' \,\Pcorr_d\Big(\mathcal{T}_d = 1\Big)}{ \lambda^2 s^2 (s')^2\,\Pcorr_d\Big(\mathcal{T}_d = 1\Big)^2}\\
		& \leq \,  \frac{8\,\Pind_d\Big(\mathcal{T}_d= 1\Big) }{ss'\beta^2} + \frac{4}{\lambda ss' \, \beta}.
	\end{align*}
	We explain the inequalities one by one: The first one exploits that $\xi \leq \Ecorrdplus\Big[Z\Big]  / 2$ and the second line is true because  $0 \leq A \leq B$ implies $|A - 2B| \geq B$. Thirdly, we apply Markov's inequality; next, we use points (b) and (d) from Lemma \ref{Z_moments}, and the final inequality follows from $1+ ss' \leq 2$ as well as the hypothesis $ 	\Pcorrd(\mathcal{T}_{d}= 1) \geq \beta  $.
	
	Note that one has the implication $\Pcorrdplus(Z \leq \xi) \leq 1 - \beta \, \,  \implies \, \,  \Pcorrdplus(Z \leq \xi) \geq \beta $, where the latter expression corresponds to (\ref{testreq_1_c}). Hence, all that remains to be done is choosing $\lambda_0$ and $\varepsilon$ such that for all $\lambda \geq \lambda_0$, the following two equations hold:
	\begin{equation}\label{req_xi}
		\max\Big\{ 4\lambda \sqrt{s s'} \, \Pind_d\Big(\mathcal{T}_d = 1\Big)^{1/4},\,  3 \lambda^{3/4} \sqrt[4]{s s'} \Big\} \, = \xi \, \leq \, \Ecorrdplus\Big[Z\Big]  / 2, 
	\end{equation}
	\begin{equation}\label{req_c}
		\frac{8\,\Pind_d\Big(\mathcal{T}_d= 1\Big) }{ss'\beta^2} + \frac{4}{\lambda ss' \, \beta} \, \leq \, 1-\beta.
	\end{equation}
	Imposing $\lambda \geq \frac{8}{ss'(1-\beta)}$ and $\Pind_d\Big(\mathcal{T}_d= 1\Big) \leq \frac{(1-\beta)ss' \beta^2}{16} $ ensures (\ref{req_c}). Next, by combining $\Ecorrdplus\Big[Z\Big]  =  \lambda s s' \,\Pcorr_d\Big(\mathcal{T}_d = 1\Big) $ with $ 	\Pcorrd(\mathcal{T}_{d}= 1) \geq \beta  $ one observes that (\ref{req_xi}) is implied by
	\begin{equation*}
		4 \lambda \sqrt{s s'} \, \Pind_d\Big(\mathcal{T}_d = 1\Big)^{1/4} \leq \, \frac{1}{2}\, \lambda s s' \beta \quad \text{and} \quad 3 \lambda^{3/4} \sqrt[4]{s s'} \leq \, \frac{1}{2}\, \lambda s s'\beta.
	\end{equation*}
	This is equivalent to
	\begin{equation*}
		\Pindd\Big(\mathcal{T}_d = 1\Big) \leq \Big(\frac{ \sqrt{ss'} \beta}{8} \Big)^4\quad \text{and} \quad \frac{6^4}{\beta^4(ss')^3}\leq \lambda.
	\end{equation*}
	Finally, we set
	\begin{equation*}
		\lambda_0 := \max\Big\{\frac{6^4}{\beta^4(ss')^3}\, , \, \frac{8}{ss'(1-\beta)}\Big\} \quad \text{and} \quad \varepsilon := \min\Big\{\Big(\frac{ \sqrt{ss'} \beta}{8} \Big)^4\, , \, \frac{(1-\beta)ss' \beta^2}{16}\Big\}.
	\end{equation*}
	Note that both variables only depend on $\beta, s$ and $s'$. Consequently, we could have set these values from the start, and all arguments would work without extra assumptions, independently from the depth $d$. This concludes the proof.
\end{proof}

\section{Conclusion}\label{section:conclusion}
In this work, we extend the graph alignment problem to the asymmetric case with varying edge densities and different node numbers. This advances the applicability of graph alignment, but our analysis remains theoretical within the well-studied Erdős--Rényi setting. Section \ref{section:setting_and_random_graph_model} provides a rigorous model definition, adjusted performance metrics, and the idea of one-sided partial alignment in the asymmetric case.

Since our study focuses on sparse graphs, Section \ref{section:From_graphs_to_trees_and_back} translates the problem to correlation testing for local Galton--Watson trees. We generalize prior work to an asymmetric tree model and show that Algorithm \ref{algo} achieves alignment on the graph level by comparisons on the tree level.

Finally, Section \ref{section:tree_correl_testing} investigates asymmetric tree correlation testing in more detail. Leveraging advantageous properties of the likelihood ratio, particularly its diagonalization formula, we demonstrate a phase transition in the feasibility of tree correlation testing. The sharp threshold  $ss' = \alpha$  provides insights not only for tree correlation testing but also for graph alignment, as summarized in the following corollary:
\begin{corollary}[Combination of Theorem \ref{thm:Algo_works} and Theorem \ref{thm:phase_transition}.]\label{corollary:final}
	For $\lambda > 0$ large enough and $ss' > \alpha$ where $\alpha$ is Otter's constant, Algorithm \ref{algo} achieves one-sided recovery of the true alignment $\sigma_*$.\\
	In particular for $s = 1$ and $\alpha < s' < 1$, in which case $G'$ is a proper random subgraph of $G$, one recovers a significant portion of node correspondences while making negligible error.
\end{corollary}
This result has new theoretical implications for random subgraph matching. While the general subgraph isomorphism problem is NP-hard, the here presented random version is feasible in polynomial time. This resolves an open problem discussed in personal communication with Jiaming Xu \cite{xu2017personal} and naturally extends the work of \cite{bollobas_random_2008, babai_random_1980}.

\paragraph{Open questions.} While Corollary \ref{corollary:final} represents a positive result for asymmetric tree correlation testing, some questions remain:
\begin{itemize}
	\item We know from Theorem \ref{thm:phase_transition} that $ss' \leq \alpha$ implies impossibility of tree correlation testing. It remains an open question whether sparse graph alignment is also infeasible in this regime or if a non-local algorithm could function.
	\item MPAlign achieves one-sided recovery but its theoretical power is lowered by the dangling tree trick. Disregarding nodes with fewer than three neighbors limits the algorithm’s ability to align a significant portion of the giant component. A more refined analysis to extend the algorithm's scope is left for future research.
\end{itemize}
Aside from a refined analysis of the problem at hand, there are several directions in which the concept of asymmetric correlated graphs could enrich the alignment literature. Here are two examples:
\begin{itemize}
	\item Within the sparse regime, one natural generalization is to study the alignment of correlated configuration models with fixed degree sequences.
	\item In the denser settings of the correlated Erdős--Rényi model, numerous findings can be generalized to the asymmetric setting. We expect that some extend naturally, while others may require new techniques to accommodate varying node numbers and edge densities.
\end{itemize}
Many geometric graphs and real-world networks exhibit inherent asymmetry, necessitating alignment techniques beyond the classical symmetric setting. We hope to contribute a step towards deeper understanding of correlated networks beyond the theoretical playground of symmetrical pairs of Erdős--Rényi graphs.

\appendix

\section{Appendix: Deferred proofs.}\label{appendix:main}

\subsection{Proof of Lemma \ref{lemma:disj_vs_ind}}\label{appendix:RdSecondMoment}
We recall the lemma about the second moment of $R_d$ before moving to its proof. Here, we shift the index from $d-1$ to $d$ compared to Lemma \ref{lemma:disj_vs_ind}, which does not impact the result and simplifies notations. 
\begin{lemma}(copy of Lemma \ref{lemma:disj_vs_ind})
	Under the parameters of Theorem \ref{thm:Algo_works} and using the notations from the proof of Theorem \ref{thm:Algo_works}, one has
	\begin{equation*}
		\Eindd\Big[ R_ d^2 \indicator{\mathcal{A}^c \cap \mathcal{B}} \Big] \in \mathcal{O}(1).
	\end{equation*}
\end{lemma}
\begin{proof}[Proof]
	This proof uses vocabulary and ideas from techniques to establish relationships between graph neighborhoods and branching processes; c.f. \cite[p. 161]{alon2016probabilistic} and \cite{bollobas2007phase}.
	
	Let $(G, G') \sim \mathrm{CER}(N, \lambda, s, s')$ and let $i \in V$ as well as $j \in V'$. We describe the sampling process from $\Pdisj$ as a simultaneous exploration of $G$ and $G'$, starting at $i$ and $j$. This process is conditioned on the following hypotheses from Definition \ref{def:disjTrees}:
	\begin{enumerate}[label=(\alph*)]
		\item Either $i \not \in V_*$ or $\sigma_*(i) \neq j$,
		\item  $\mathcal{N}_{G, 2d}(i)$ and $\mathcal{N}_{G', 2d}(j) $ are trees,
		\item $\mathcal{N}_{G, d}(i)\cap V_*$ and $\mathcal{N}_{G', d}(j) \cap V'_*$ are disjoint as subgraphs of $G_*$.
	\end{enumerate}
	Note that this is not precisely the vocabulary of dangling trees since we are not talking about pointing neighborhoods: Considering $ \mathcal{N}_{G, d}(i)$ instead of $\mathcal{N}_{G, d}(i \to i_k)$ simplifies notation but does not influence this proof which could well be written with pointed neighborhoods. 
	
	Sampling from $\Pdisj$ can be seen as the following graph exploration algorithm: It describes the sampling process of $(\mathcal{N}_{G, d}(i), \mathcal{N}_{G', d}(j))$ conditioned on (a), (b), (c). 

	\paragraph{Initialisation.}  Initialise $ \mathcal{N}_{G, d}(i) := \{i\}$ and $ \mathcal{N}_{G', d}(j):=\{j\}$ as the trivial trees consisting only of their roots. These roots are assigned depth $0$ in their respective trees.  Next, we consider three types of node sets during the exploration process:
	\begin{itemize}
		\item The active sets $\mathcal{A}^*$, $\mathcal{A}$ and $ \mathcal{A}' $ which are subsets of $V_*$ (which is identified as $V'_*$ in $G'$), $V\setminus V_*$, and $V' \setminus V'_*$ respectively. If $i$ is a node of $G_*$, one initialises $\mathcal{A}^* = \{i\}, \mathcal{A} = \emptyset$; else if $i \in V \setminus V_*$ then $\mathcal{A}^* = \emptyset, \mathcal{A} = \{i\}$. Similarly, add $j$ to $\mathcal{A}^*$ if contained in $G_*$, or add it to $\mathcal{A}'$ otherwise. 
		\item The deactivated sets $\mathcal{D}^*, \mathcal{D}$, and $\mathcal{D}'$ which are all initialized as $\emptyset$. 
		\item The sets of unvisited nodes $\mathcal{U}^* := V_* \setminus \mathcal{A}^*$, $ \mathcal{U} := V \setminus \mathcal{A}$, and $ \mathcal{U}' := V'_* \setminus \mathcal{A}'$. 
	\end{itemize}
	Let $A^*$ denote the cardinality of $\mathcal{A}^*$, $D'$ the cardinality of $\mathcal{D}'$, and so forth for all sets. At all times during the sampling process, one has
	\[ 
	U^* = n_* - A^* - D^*, \, U = n_+ - A - D, \text{ and } U' = n'_+ - A' - D'.
	\] 
	
	\paragraph{Exploration of the unvisited nodes.} While at least one active set ($ \mathcal{A}^*, \mathcal{A} $, or $\mathcal{A}'$) is non-empty and contains nodes of depth $d$ or less (with respect to the root in their respective graph), execute the following steps:
	\begin{enumerate}
		\item Arbitrarily pick one node $k \in \mathcal{A}^* \cup \mathcal{A} \cup \mathcal{A}'$ with minimal depth among all such nodes, remove it from its respective active set, and add it to the corresponding deactivated set. For instance if $k \in \mathcal{A}'$ then move $k$ to $\mathcal{D}'$. 
		\item If $k$ is part of $ \mathcal{N}_{G, d}(i) $, then sample $ \mathrm{Bin}(U^*, \lambda r r'/N) $ child nodes from $\mathcal{U}^*$ and $ \mathrm{Bin}(U, \lambda r (1-r')/N) $  child nodes from $\mathcal{U}$. Remove these child nodes from their respective unvisited node sets and add them to the corresponding active sets. 
		\item Else if $k$ is part of $\mathcal{N}_{G', d}(j)$, then sample $ \mathrm{Bin}(U^*, \lambda rr'/N) $ child nodes from $\mathcal{U}^*$ and $ \mathrm{Bin}(U', \lambda r'(1-r)/N) $  child nodes from $\mathcal{U}'$. Again, remove these child nodes from their respective unvisited node sets and add them to the corresponding active sets. 
		\item Assign correct depths to all new active nodes: If $k$ was of depth $d_k$, then all children are assigned depth $d_k + 1$. 
	\end{enumerate}
	
	\paragraph{Computation of the likelihood under $\Pdisjd$.} Fix two trees $t$ and $t'$ of depth $d$ and impose an arbitrary breadth-first-search numbering across both trees: Assign the numbers $1$ and $2$ to the two roots, then number subsequent generations such that nodes of higher depth have strictly higher numbers. Let $K = |t| + |t'|$ be the total number of nodes across both trees and denote by $c_\kappa$ the number of children of each node $\kappa \in [K]$.
	
	We start by computing $\Pdisjd(t, t')$ which is independent of the arbitrary node numbering $\kappa \in [K]$. For this computation, one goes through the nodes of $t$ and $t'$, computing the likelihood that each child number $c_\kappa$ equals the binomial distributions from the exploration process.
	
	Note that in steps (2) and (3) of the exploration process, one needs to condition on the fact that edges do not connect active nodes amongst each other. This corresponds to the condition that $ \mathcal{N}_{G, d}(i) $ and $ \mathcal{N}_{G', d}(j) $ are non-intersecting and cycle-free. It translates to the event $\{\mathrm{Bin}(A^* - 1, \lambda s s') = 0\}$ for nodes of the intersection graph and similarly for other subgraphs. Subtracting the number $ 1 $ from $A^*$ compensates for the fact that at the beginning of each iteration, the current node $k$ is still in the active set. 
	
	We index the exploration steps by $\kappa \in [K]$ and use that same notation for the evolving set cardinalities $A^*_\kappa, U_\kappa, D'_\kappa$ etc. By ``$\kappa \in t$'', we denote an iteration over all exploration indices $ \kappa \in [K]$ that correspond to nodes of $t$. The likelihood of $ (t, t') $ to be sampled from $\Pdisjd$ can then be computed as follows:
	\begin{align*}
		\Pdisjd(t, t') &= \prod_{\kappa \in t} \Prob\Big( \mathrm{Bin}(U_\kappa^* + U_\kappa , \lambda r/N) = c_\kappa\Big) \Prob\Big(\mathrm{Bin}(A_\kappa^* + A_\kappa - 1, \lambda r /N) = 0 \Big)\\
		&\quad \quad \times  \prod_{\kappa  \in t'} \Prob\Big( \mathrm{Bin}(U_\kappa^* + U'_\kappa , \lambda r'/N) = c_\kappa\Big) \Prob\Big(\mathrm{Bin}(A_\kappa^* + A'_\kappa - 1, \lambda r' /N) = 0 \Big)\\
		&= \prod_{\kappa \in  t} \binom{U_\kappa^* + U_\kappa}{c_\kappa} \Big(\frac{\lambda r}{N}\Big)^{c_\kappa} \Big(1 -\frac{\lambda r}{N} \Big)^{U_\kappa^* + U_\kappa - c_\kappa} \, \Big(1 - \frac{\lambda r}{N}\Big)^{A_\kappa^* + A_\kappa - 1}\\
		&\quad \quad \times  \prod_{\kappa \in  t'} \binom{U_\kappa^* + U'_\kappa}{c_\kappa} \Big(\frac{\lambda r'}{N}\Big)^{c_\kappa} \Big(1 -\frac{\lambda r'}{N} \Big)^{U_\kappa^* + U'_\kappa - c_\kappa} \, \Big(1 - \frac{\lambda r'}{N}\Big)^{A_\kappa^* + A'_\kappa - 1}\\
		&= \prod_{\kappa  \in t} \binom{U_\kappa^* + U_\kappa}{c_\kappa} \Big(\frac{\lambda r}{N}\Big)^{c_\kappa} \Big(1 -\frac{\lambda r}{N} \Big)^{U_\kappa^* + U_\kappa + A_\kappa^* + A_\kappa  - c_\kappa - 1}\\
		&\quad \quad \times  \prod_{\kappa \in  t'} \binom{U_\kappa^* + U'_\kappa}{c_\kappa} \Big(\frac{\lambda r'}{N}\Big)^{c_\kappa} \Big(1 -\frac{\lambda r'}{N} \Big)^{U_\kappa^* + U'_\kappa + A_\kappa^* + A'_\kappa - c_\kappa - 1}.
	\end{align*}
	To simplify this expression, we track the quantities $A_\kappa, D'_\kappa, U^*_\kappa, ...$ throughout the sampling process. The initialization leads to 
	\begin{align*}
		A^*_1 &= \indicator{i \in V_*} + \indicator{j \in V'_*}, \, A_1 = \indicator{i \not \in V_*}, \, A'_1 = \indicator{j \not \in V'_*}, \\
		D^*_1 &= D_1 = D'_1 = 0.
	\end{align*}
	These quantities evolve as follows with $\kappa \in [K]$:
	\begin{align*}
		A_\kappa^* &= A^*_1 - D^*_\kappa + \sum_{\eta \in [\kappa]} c^*_\eta, \quad A_\kappa = A_1 - D_\kappa+ \sum_{\eta \in [\kappa]} c^+_\eta, \quad A'_\kappa= A'_1 - D'_\kappa+ \sum_{\eta \in [\kappa]} c'^{+}_\eta, \\
		D_\kappa^*& = \sum_{\eta \in [\kappa]} \indicator{\eta \in V_*},\quad D_\kappa = \sum_{\eta \in [\kappa]} \indicator{\eta \in V \setminus V_*}, \quad  D'_\kappa = \sum_{\eta \in [\kappa]} \indicator{\eta  \in V'\setminus V'_*}, \\
		U_\kappa^* &= n_* - A^*_\kappa - D^*_\kappa, \quad U_\kappa= n_+ - A_\kappa - D_\kappa, \quad U'_\kappa = n'_+ - A'_\kappa - D'_\kappa 
	\end{align*}
	where $c_\eta = c^*_\eta + c^+_\eta$ denotes the decomposition of children among those in the intersection graph $G_*$ and those in the added nodes $V \setminus V_*$ respectively $V' \setminus V'_*$. Since the product expansion of $ \Pdisjd(t, t') $ is only a function of the sums $A_\kappa^* + A'_\kappa, U_\kappa^* + U_\kappa$, $U_\kappa^* + U'_\kappa$ and $A_\kappa^* + A_\kappa$, we compute these sums directly:
	\begin{align*}
		A_\kappa^* + A_\kappa &= A^*_1 + A_1 - D^*_\kappa - D_\kappa + \sum_{\eta \in [\kappa]} (c^*_\eta+ c^+_\eta) \\
		 &= 1 + \indicator{j\in V'_*}- \sum_{\eta \in [\kappa]} \indicator{\eta \in V} + \sum_{\eta \in [\kappa], \eta \in V} c_\eta,
	\end{align*}
	\begin{align*}
		A_\kappa^* + A'_\kappa & =  A^*_1 + A'_1 - D^*_\kappa - D'_\kappa + \sum_{\eta \in [\kappa]} (c^*_\eta + c'^+_\eta) \\
		&= 1 + \indicator{i \in V_*} - \sum_{\eta \in [\kappa]} \indicator{\eta \in V'} + \sum_{\eta \in [\kappa], \eta \in V'} c_\eta ,
	\end{align*}
\begin{align*}
		U_\kappa^* + U_\kappa &= n_* +  n_+ - A^*_\kappa - A_\kappa - D^*_\kappa  - D_\kappa\\
		&= n - 1 -\indicator{j \in V'_*} - \sum_{\eta \in [\kappa], \eta \in V} c_\eta,\\
		U_\kappa^* + U'_\kappa &= n_* +  n'_+ - A^*_\kappa - A'_\kappa - D^*_\kappa  - D'_\kappa\\
		& n' - 1 -\indicator{i \in V_*} - \sum_{\eta \in [\kappa], \eta \in V'} c_\eta. 
	\end{align*}
	Introducing the notations $\mathbb{I} := 1 + \indicator{j \in V'_*}$ and $\mathbb{I}' := 1 + \indicator{i \in V_*}$, this simplifies to 
	\begin{align*}
		U_\kappa^* + U_\kappa & = n - \mathbb{I} - \sum_{\eta \in [\kappa] \cap V} c_\eta, \\
		U_\kappa^* + U'_\kappa &= n' - \mathbb{I}' - \sum_{\eta \in [\kappa] \cap V'} c_\eta, \\
		U_\kappa^* + U_\kappa + A_\kappa^* + A_\kappa &=  n - \sum_{\eta \in [\kappa]} \indicator{\eta \in V}, \\
		U_\kappa^* + U'_\kappa + A_\kappa^* + A'_\kappa & =  n' - \sum_{\eta \in [\kappa]} \indicator{\eta \in V'}.
	\end{align*}
	Plugging these into the expression for $\Pdisjd(t, t')$ yields
	\begin{align*}
		&\Pdisjd(t, t') = \prod_{\kappa \in t} \binom{n - \mathbb{I} - \sum_{\eta \in [\kappa] \cap V} c_\eta}{c_\kappa} \Big(\frac{\lambda r}{N}\Big)^{c_\kappa} \Big(1 -\frac{\lambda r}{N} \Big)^{n - 1 - c_\kappa - \sum_{\eta \in [\kappa]} \indicator{\eta \in V}}\\
		&\qquad \qquad \qquad \times  \prod_{\kappa \in t'}  \binom{n' - \mathbb{I}' - \sum_{\eta \in [\kappa] \cap V'} c_\eta}{c_\kappa} \Big(\frac{\lambda r'}{N}\Big)^{c_\kappa} \Big(1 -\frac{\lambda r'}{N} \Big)^{n' - 1 - c_\kappa - \sum_{\eta \in [\kappa]} \indicator{\eta \in V'}}\\
		& = \prod_{\kappa \in t} \frac{\big( n - \mathbb{I} - \sum_{\eta \in [\kappa] \cap V} c_\eta \big)!}{c_\kappa ! \, \big( n - \mathbb{I} - \sum_{\eta \in [\kappa - 1] \cap V} c_\eta\big)!} \Big(\frac{\lambda s}{Nq}\Big)^{c_\kappa}  \Big(1  - \frac{\lambda s}{Nq}\Big)^{n - 1 - c_\kappa - \big \vert [\kappa] \cap V\big \vert}   \\
		& \quad \quad \times \prod_{\kappa \in t'} \frac{\big( n' - \mathbb{I}' - \sum_{\eta \in [\kappa] \cap V'} c_\eta \big)!}{c_\kappa ! \, \big( n' - \mathbb{I}' - \sum_{\eta \in [\kappa - 1] \cap V'} c_\eta\big)!} \Big(\frac{\lambda s'}{Nq'}\Big)^{c_\kappa}  \Big(1  - \frac{\lambda s'}{Nq'}\Big)^{n' - 1 - c_\kappa - \big \vert [\kappa] \cap V'\big \vert}.
	\end{align*}
   	Recall that $\mathcal{B} \subset \{ n - \varepsilon \leq Nq \leq n + \varepsilon\} \cap \{n' - \varepsilon \leq Nq' \leq n' + \varepsilon\}$ where 
	\[
	\varepsilon = N^{\sfrac12 + \gamma } \quad \text{and} \quad \gamma =  \frac{1}{8} - \frac{1}{4} c \log(\lambda \max\{s, s'\})  > 0.
	\]
	Conditioning on $\mathcal{B}$, one can upper bound $\Pdisjd(t, t')$ as follows:
	\begin{align*}
		\Pdisjd(t,& t') \indicator{\mathcal{B}} \leq \prod_{\kappa \in t} \frac{\big( n - \mathbb{I} - \sum_{\eta \in [\kappa] \cap V} c_\eta \big)!}{c_\kappa ! \, \big( n - \mathbb{I} - \sum_{\eta \in [\kappa - 1] \cap V} c_\eta\big)!} \frac{(\lambda s)^{c_\kappa} }{(n - \varepsilon)^{c_\kappa} } \Big(1  - \frac{\lambda s}{n+\varepsilon}\Big)^{n - c_\kappa - \big \vert [\kappa] \cap V\big \vert}   \\
		& \times  \prod_{\kappa \in t'} \frac{\big( n' - \mathbb{I}'- \sum_{\eta \in [\kappa] \cap V'} c_\eta \big)!}{c_\kappa ! \, \big( n' - \mathbb{I}' - \sum_{\eta \in [\kappa - 1] \cap V'} c_\eta\big)!} \frac{(\lambda s')^{c_\kappa} }{(n' - \varepsilon)^{c_\kappa} } \Big(1  - \frac{\lambda s'}{n'+\varepsilon}\Big)^{n' - c_\kappa - \big \vert [\kappa] \cap V'\big \vert}.
	\end{align*}
	In the fraction $ ( n - \mathbb{I} - \sum_{\eta \in [\kappa] \cap V} c_\eta \big)!/( n - \mathbb{I} - \sum_{\eta \in [\kappa - 1] \cap V} c_\eta\big)!$, the numerator is equal to the denominator plus $c_\kappa$. Hence, using that the numerator is a product of factors bounded by $n$, one has  
	\[
	\frac{\big( n - \mathbb{I} - \sum_{\eta \in [\kappa] \cap V} c_\eta \big)!}{ \big( n - \mathbb{I} - \sum_{\eta \in [\kappa - 1] \cap V} c_\eta\big)!} \times \frac{1}{n^{c_\kappa}} \leq 1
	\]
	which still holds for sufficiently large $n$ if one replaces $1/n^{c_\kappa}$ by $1/(n-\varepsilon)^{c_\kappa}$ because $\varepsilon \in o(N)$. Consequently, there exists $N_0 \in \N$ such that for all $N \geq N_0$,
	\begin{align}\label{asym_upper_bound_Pdis}
		\Pdisjd(t, t') \indicator{\mathcal{B}} &\leq \prod_{\kappa \in t} \frac{1}{c_\kappa !} (\lambda s)^{c_\kappa} \Big(1  - \frac{\lambda s}{n+\varepsilon}\Big)^{n - c_\kappa - \big \vert [\kappa] \cap V\big \vert}.  \nonumber \\
		&\quad \quad \times  \prod_{\kappa \in t'} \frac{1}{c_\kappa !} (\lambda s')^{c_\kappa} \Big(1  - \frac{\lambda s'}{n'+\varepsilon}\Big)^{n' - c_\kappa - \big \vert [\kappa] \cap V'\big \vert}.
	\end{align}
	For the rest of this proof, we assume $N \geq N_0$ in order for (\ref{asym_upper_bound_Pdis}) to hold. 
	
	\paragraph{Computation of $R_d$, the likelihood ratio of $\Pdisjd$ and $\Pindd$.} It is quite simple to compute $\Pindd(t, t')$ because of the independence in Galton--Watson trees: 
	\begin{align*}
		\Pindd(t, t') &= \prod_{\kappa \in t} \Prob\Big(\mathrm{Poi}(\lambda s) = c_\kappa\Big) \times \prod_{\kappa \in t'} \Prob\Big(\mathrm{Poi}(\lambda s') = c_\kappa\Big)\\
		&= \prod_{\kappa \in t} e^{-\lambda s }  (\lambda s)^{c_\kappa}  \frac{1}{c_\kappa !} \prod_{\kappa \in t'} e^{-\lambda s'}  (\lambda s')^{c_\kappa}  \frac{1}{c_\kappa !}.
	\end{align*}
	Comparing this to the upper bound on $\Pdisjd(t, t')$ from (\ref{asym_upper_bound_Pdis}), one obtains
	\begin{align}\label{R_d_init_upper_bound}
		&R_d(t, t')\indicator{\mathcal{B}} = \Big(\Pdisjd(t, t')/ \Pindd(t, t')\Big) \indicator{\mathcal{B}}\\
		& \leq \prod_{\kappa \in t} e^{\lambda s} \Big(1  - \frac{\lambda s}{n+\varepsilon}\Big)^{n - c_\kappa - \big \vert [\kappa] \cap V\big \vert} \times  \prod_{\kappa \in t'}e^{\lambda s'}  \Big(1  - \frac{\lambda s'}{n'+\varepsilon}\Big)^{n' - c_\kappa - \big \vert [\kappa] \cap V'\big \vert}  \indicator{\mathcal{B}}.\nonumber
	\end{align}
	Moving forward, we condition on the event $\mathcal{A}^c$ which we recall to be defined as
	\begin{align*}
		\mathcal{A}^c &= \Big\{ \forall i \in V, \, \rho \in [d] : \, | \mathcal{S}_{G, \rho}(i) |  \leq C \log(n) (\lambda s)^\rho \Big\}\\ & \quad \quad \cup \Big\{ \forall j \in V', \, \rho' \in [d]  : \, | \mathcal{S}_{G', \rho'}(j) | \leq  C \log(n') (\lambda s')^\rho \Big\}\\
		& \subset \Big \{ |t| \leq C \log(n) \sum_{\rho = 0}^d (\lambda s)^\rho  \Big\} \cup \Big \{   |t'| \leq C \log(n') \sum_{\rho' = 0}^d (\lambda s')^{\rho'}    \Big\}\\
		& \leq  \Big \{ |t| \leq \tilde{C} \log(n) (\lambda s)^{d}\Big\} \cup \Big \{   |t'| \leq \tilde{C} \log(n') (\lambda s')^{d} \Big\},
	\end{align*}
	where we have introduced the constant $\tilde{C} := C \frac{\lambda \max\{s, s'\}}{\lambda \min\{s, s'\} - 1} > 0$. Under $\mathcal{A}^c$, we have
	\begin{align*}
		\sum_{\kappa \in t} \big \vert [\kappa] \cap V\big \vert \, =  \, 1  + 2 + ... + |t| \, = \frac{|t| \big (|t| + 1 \big)}{2}  \, \leq \,  |t|^2 \, \leq \,  \tilde{C}^2 \log(n)^2 (\lambda s)^{2d}.
	\end{align*}
	Plugging the bounds on $|t|$ and $\sum_\kappa |[\kappa] \cap V|$ into the first product ``$\prod_{\kappa \in  t}$'' in (\ref{R_d_init_upper_bound}), we obtain
	\begin{align}\label{bounding_product_kappa}
	&\prod_{\kappa \in t} e^{\lambda s} \Big(1  - \frac{\lambda s}{n+\varepsilon}\Big)^{n - c_\kappa - \big \vert [\kappa] \cap V\big \vert} \indicator{\mathcal{A}^c}\\
	&\quad =  \exp\Big[ \lambda s \, \vert t \vert +  \Big(n   \vert t \vert -  \sum_{\kappa \in t} c_\kappa  - \sum_{\kappa \in t} \big \vert [\kappa] \cap V\big \vert \Big) \, \log \Big( 1 - \frac{\lambda s}{n + \varepsilon} \Big)  \Big] \indicator{\mathcal{A}^c} \nonumber\\
	&\quad \overset{(\star)}{\leq} \exp\Big[ \lambda s \, \vert t \vert -  \frac{\lambda s}{n + \varepsilon}\Big(n   \vert t \vert -  \sum_{\kappa \in t} c_\kappa  - \sum_{\kappa \in t} \big \vert [\kappa] \cap V\big \vert \Big)  \Big] \indicator{\mathcal{A}^c} \nonumber\\
	& \quad = \exp\Big[  \Big (1  -  \frac{n}{n + \varepsilon} \Big) \lambda s \, \vert t \vert  + \frac{1}{n+ \varepsilon} \lambda s \, \frac{ \vert t \vert \big(\vert t \vert + 1 \big)}{2}  \Big]  \exp\Big[\frac{\lambda s}{n + \varepsilon} \sum_{\kappa \in t} c_\kappa\Big] \indicator{\mathcal{A}^c} \nonumber\\
	&\quad \leq \exp\Big[ \frac{\varepsilon}{n + \varepsilon} \tilde{C} \log(n) (\lambda s)^{d+1} \, + \,  \frac{1}{n+ \varepsilon} \tilde{C}^2 \log(n)^2 (\lambda s)^{2d + 1} \Big] \exp\Big[\frac{\lambda s}{n + \varepsilon} \sum_{\kappa \in t} c_\kappa\Big] \nonumber
	\end{align}
	where the inequality $(\star)$ uses $\log(1 - x) \leq -x$ for $x < 1$. Recalling that $Nq \leq n+ \varepsilon$ and $\varepsilon = N^{\sfrac12 + \gamma}$, we have 
	\[
	\frac{\varepsilon}{n+\varepsilon} \leq \frac{N^{\sfrac12 + \gamma}}{Nq} \quad \text{and} \quad \frac{1}{n+\varepsilon} \leq \frac{1}{Nq}.
	\]
	We further bound $\log(n) \leq \log(N)$ and use the assumption $d \leq c \log(N)$ from Theorem \ref{thm:Algo_works}, which leads to
	\begin{equation}\label{bound_on_lambda_s_d}
			(\lambda s)^d \leq (\lambda s)^{c \log(N)} = N^{c \log(\lambda s)} \leq N^{c \log(\lambda \max\{s, s'\})} \leq N^{\sfrac12 - 2 \gamma} \quad 
	\end{equation}
	where the last inequality follows from $c \log(\lambda \max\{s, s'\}) \leq \frac{1}{4}< \frac{1}{2}$ and
	\[
	\frac{1}{2} - 2 \gamma  \geq \frac12 - \frac14 + \frac12 c \log(\lambda \max\{s, s'\}) > c \log(\lambda \max\{s, s'\}.
	\]
	Going back to (\ref{bounding_product_kappa}), we can thus bound the first term of the last line as follows:
	\begin{align*}
		&\exp\Big[ \frac{\varepsilon}{n + \varepsilon} \tilde{C} \log(n) (\lambda s)^{d+1} \, + \,  \frac{1}{n+ \varepsilon} \tilde{C}^2 \log(n)^2 (\lambda s)^{2d + 1} \Big]\\
		&\qquad \leq \exp\Big[ \frac{N^{\sfrac12 + \gamma}}{Nq} \tilde{C} \log(N) \, \lambda s\, N^{\sfrac12 - 2\gamma} + \frac{1}{Nq} \tilde{C}^2 \log(N)^2 \, \lambda s \, N^{1 - 4 \gamma}  \Big]\\
		& \qquad = \exp\Big[ \frac{\tilde{C}\, \lambda s}{q}  \log(N) N^{-  \gamma  } + \frac{\tilde{C}^2 \lambda s}{q} \log(N)^2 N^{-4\gamma} \Big].
	\end{align*}
	Since $\gamma > 0$, the last line is in $\mathcal{O}(1)$ and we note that its value is deterministic.
	
	The same argumentation which we started in (\ref{bounding_product_kappa})  can be used verbatim by replacing $s, n, q, t$ with $s', n', q', t'$ in order to show
	\begin{align*}
		&\prod_{\kappa \in t'} e^{\lambda s'} \Big(1  - \frac{\lambda s'}{n'+\varepsilon}\Big)^{n' - c_\kappa - \big \vert [\kappa] \cap V'\big \vert} \indicator{\mathcal{B} \cap \mathcal{A}^c}\\
		&\ \leq \exp\Big[ \frac{\varepsilon}{n' + \varepsilon} \tilde{C} \log(n') (\lambda s')^{d+1} +  \frac{1}{n'+ \varepsilon} \tilde{C}^2 \log(n')^2 (\lambda s')^{2d + 1} \Big] \exp\Big[\frac{\lambda s'}{n' + \varepsilon} \sum_{\kappa \in t'} c_\kappa\Big]\\
		&\ \leq \mathcal{O}(1) \times \exp\Big[\frac{\lambda s'}{n' + \varepsilon} \sum_{\kappa \in t'} c_\kappa\Big].
	\end{align*}
	Using the thus obtained bounds on both products ``$\prod_{\kappa \in t}$'' and ``$\prod_{\kappa \in t'}$'' in (\ref{R_d_init_upper_bound}) yields
	\begin{align*}
	R_d(t, t')\indicator{\mathcal{B} \cap \mathcal{A}^c} \leq  \mathcal{O}(1) \times \exp\Big[\frac{\lambda s}{n + \varepsilon} \sum_{\kappa \in t} c_\kappa\Big] \exp\Big[\frac{\lambda s'}{n' + \varepsilon} \sum_{\kappa \in t'} c_\kappa\Big].
	\end{align*}
	Squaring this inequality and using $\frac{1}{n+\varepsilon} \leq \frac{1}{Nq}$ on $\mathcal{B}$, we get
	\begin{align}\label{upper_bound_final_on_Rd2}
		R_ d^2(t, t')\indicator{\mathcal{B} \cap \mathcal{A}^c} \leq  \mathcal{O}(1) \times \exp\Big[\frac{2 \lambda s}{Nq} \sum_{\kappa \in t} c_\kappa\Big] \exp\Big[\frac{2 \lambda s'}{Nq'} \sum_{\kappa \in t'} c_\kappa\Big].
	\end{align}
	
	\paragraph{Computing the expectation.} To conclude the proof of Lemma \ref{lemma:disj_vs_ind}, one needs to compute the expectation of $R_d^2$ under $\Pindd$. The assumption that $(t, t') \sim \Pindd$ is equivalent to $c_\kappa \sim \mathrm{Poi}(\lambda s)$ for $\kappa \in t$ and $c_\kappa \sim \mathrm{Poi}(\lambda s')$ for $\kappa \in t'$ with all $c_\kappa$ being independent. Hence, taking expectations in (\ref{upper_bound_final_on_Rd2}) yields
	\begin{align}\label{holadrio}
		\Eindd\Big[ R_ d^2 \indicator{\mathcal{B} \cap \mathcal{A}^c}  \Big]  \leq  \mathcal{O}(1) \times\prod_{\kappa \in t}  \Eindd\big[e^{\frac{2\lambda s}{Nq} c_{\kappa}}\big]   \prod_{\kappa \in t'}  \Eindd\big[e^{\frac{2\lambda s'}{Nq'} c_{\kappa}}\big] \indicator{ \mathcal{A}^c}
	\end{align}
	Since the generating function of $c_\kappa$ equals $\mu \mapsto \Expe[e^{\mu c_\kappa}] = \exp(\lambda s(e^\mu - 1))$, we obtain
	\begin{align*}
		\prod_{\kappa \in t} \Eindd\big[ e^{\frac{2\lambda s}{Nq} c_{\kappa}} \big] \indicator{\mathcal{A}^c} &= \exp\Big(\lambda s (e^{\frac{2\lambda s}{Nq}} -1)\Big)^{|t|}  \indicator{\mathcal{A}^c} \\
		& \leq  \exp\bigg(\lambda s (e^{\frac{2\lambda s}{Nq}}  - 1) \tilde{C} \log(n) (\lambda s)^{d} \bigg).
	\end{align*}
	A first order approximation yields $e^{\frac{2\lambda s}{Nq}} - 1  \in \mathcal{O}(N^{-1})$ and we recall from (\ref{bound_on_lambda_s_d}) that $(\lambda s)^d \leq N^{\sfrac12 - 2\gamma}$. Hence,
\[
 (e^{\frac{2\lambda s}{Nq}}  - 1) \log(n) (\lambda s)^{d} \leq \mathcal{O}\Big( N^{-1} \log(N) N^{\sfrac12 - 2\gamma}\Big) \subset o(1)
\] 
which leads to $\exp\bigg(\lambda s (e^{\frac{2\lambda s}{Nq}}  - 1) \tilde{C} \log(n) (\lambda s)^{d} \bigg) \in \mathcal{O}(1)$. Plugging this result into (\ref{holadrio}), we conclude that $d \mapsto \Eindd\Big[ R_ d^2 \indicator{\mathcal{A}^c} \indicator{\mathcal{B}} \Big]$ is bounded. This finishes the proof.
\end{proof}

\subsection{Proof of Proposition \ref{prop_likhood_recursive}}\label{appendix:prop_likhood_recursive}
We recall the recursive likelihood ratio property before moving to its proof.
\begin{proposition}[copy of Proposition \ref{prop_likhood_recursive}]
	For $t, t' \in \mathcal{X}_d$, let $c, c'$ be the degrees of their respective root nodes. Denote by $t_{[1]}, \dots, t_{[c]}$ respectively $t'_{[1]}, \dots, t'_{[c']}$ the subtrees attached to the roots, labeled randomly. Recall that $\sigma : [k] \xhookrightarrow{} [c]$ denotes an injective map from $[k]$ to $[c]$ and let $\sum_{\sigma:[k] \hookrightarrow [c]}$ be the sum over all such mappings.  Then, for all $d \in \N_{>1}$, the likelihood ratio at depth $d$ can be expressed recursively as
	\begin{equation*}
		L_d(t, t') = \sum_{k=0}^{c \wedge c'} \psi(k, c, c') \sum_{\substack{\sigma : [k] \xhookrightarrow{} [c],\\\sigma' : [k] \xhookrightarrow{} [c']}} \prod_{i=1}^k L_{d-1}\Big(t_{[\sigma(i)]}, t'_{[\sigma'(i)]}\Big).
	\end{equation*}
	where 
	\begin{equation*}
		\psi(k, c, c')= \exp(\lambda s s')\frac{1}{\lambda^k \, k!} \, (1-s')^{c - k} (1-s)^{c'-k}.
	\end{equation*}
\end{proposition}
\begin{proof} For a rooted tree $t$ , denote its root node by $\rho_t$.
	The values $c = \mathrm{deg}(\rho_t)$ and $c' = \mathrm{deg}(\rho_{t'})$ represent the  number of root children in $t$ and $t'$ respectively. Since $t$ and $t'$ are deterministic, both $c$ and $c'$ are fixed integers.\\
	In the following, $(\tau, \tau')$ denotes a pair of random trees following the law $ \Pcorrd$. Hence, they can be seen as independent augmentations of an intersection tree denoted as $\tau_*$ as explained in Definition \ref{def:corr_GW_trees}. Note that $(\tau, \tau') = (t, t')$ implies the event $D:= \{\mathrm{deg}(\rho_\tau) = c, \mathrm{deg}(\rho_\tau') = c'\}$ which lets us compute
	\begin{align*}
		\Pcorrd(t, t') &\overset{(a)}{=\joinrel=} \sum_{k = 0}^{c\wedge c'} \Pcorrd\Big((\tau, \tau') = (t, t'), \, D, \,  \mathrm{deg}(\rho_{\tau_*}) = k\Big) \\ 
		&= \sum_{k = 0}^{c\wedge c'} \Pcorrd\Big(D, \,\mathrm{deg}(\rho_{\tau_*}) = k\Big)  \; \underbrace{\Pcorrd \Big((\tau, \tau') = (t, t')\suchthat D, \, \mathrm{deg}(\rho_{\tau_*}) = k\Big)}_{=: \boldsymbol{\xi}} \\
		&\overset{(b)}{=\joinrel=} \sum_{k = 0}^{c\wedge c'} \pi_{\lambda s s'}(k) \pi_{\lambda s (1-s')}(c-k) \pi_{\lambda s' (1-s)}(c'-k) \, \boldsymbol{\xi}.
	\end{align*}
	The equality $(a)$ uses $ \mathrm{deg}(\rho_{\tau_*}) \leq \mathrm{deg}(\rho_\tau) \wedge \mathrm{deg}(\rho_{\tau'})$ while $(b)$ relies on the construction of $(\tau, \tau')$ as augmentations of $\tau_*$. From the construction of these augmentations, one has $\deg(\rho_{\tau_*}) \sim \mathrm{Poi}(\lambda ss'), \, \deg(\rho_{\tau}) -  \deg(\rho_{\tau_*})  \sim \mathrm{Poi}(\lambda s(1-s'))$, and $ \deg(\rho_{\tau'}) -  \deg(\rho_{\tau_*})  \sim \mathrm{Poi}(\lambda s'(1-s))$ with all three variables being independent.\\
	To compute $\boldsymbol{\xi}$, randomly assign the numbers $\{1, \dots, c\}$ to the child nodes of $\rho_t$ and denote by $t_{[i]}:= \mathcal{N}_{t, d}(\rho_t \to i)$ the (unlabeled) subtree rooted in $i \in [c]$, using the notation from section \ref{section:one-sided_hypo_testing}.
	Define $t'_{[1]}, \dots, t'_{[c']}$ analogously and enumerate the child nodes of $\tau, \tau'$ in a way that $(\tau_{[1]}, \tau'_{[1]}), \dots, (\tau_{[k]}, \tau'_{[k]})$ correspond to the root children in $\tau_*$ with $k = \mathrm{deg}(\rho_{\tau_*})$. With this notation, the event $ (\tau, \tau') = (t, t') $ implies that permutations $\sigma, \sigma'$ are correctly mapping the random numberings to the corresponding nodes. Denoting by $\mathfrak{S}_c$ the symmetric group of size $c$, this leads to
	\begin{align*}
		&\boldsymbol{\xi}  = \\
		&\sum_{\substack{\sigma \in \mathfrak{S}_c\\ \sigma' \in \mathfrak{S}_{c'}}} \frac{1}{c! c'!}\, \Pcorrd \Bigg( (\tau, \tau') = (t, t')\, \Big| \, D, \, \mathrm{deg}(\rho_{\tau_*}))= k, \begin{cases}
			\tau_i = t_{[\sigma(i)]} \text{ for } i \in [c],\\
			\tau'_j = t'_{[\sigma'(j)]} \text{ for } j \in [c']
		\end{cases} \hspace{-10pt}\Bigg) \\
		&= \sum_{\sigma, \sigma'} \frac{1}{c!\, c'!}\, \Pcorrdminus\Big((\tau_{[i]}, \tau'_{[i]}) = (t_{[\sigma(i)]}, t'_{[\sigma'(i)]})\text{ for } i \in [k]\Big) \\
		& \qquad \qquad \quad \quad\quad \times \Prob\Big( t_{[\sigma(i)]} = \mathrm{GW}_{d-1}^{(\lambda s)} \text{ for } i \in [k+1:c] \Big) \\
		& \qquad \qquad \quad \quad\quad\times  \Prob\Big( t'_{[\sigma'(i)]} = \mathrm{GW}_{d-1}^{(\lambda s')} \text{ for } i \in [k+1:c'] \Big)\\
		&= \sum_{\sigma, \sigma'} \frac{1}{c!\, c'!} \Big(\prod_{i = 1}^{k} \Pcorrdminus(t'_{[\sigma(i)]}, t'_{[\sigma'(i)]})\Big) \\
		& \qquad \qquad \quad \quad\quad \times \Big(\prod_{i=k+1}^{c}\mathrm{GW}_{d-1}^{(\lambda s)}(t_{[\sigma(i)]})\Big)\Big(\prod_{i=k+1}^{c'}\mathrm{GW}_{d-1}^{(\lambda s')}(t'_{[\sigma'(i)]})\Big)
	\end{align*}
	where we recall that $ \mathrm{GW}_{d-1}^{(\mu)}(t) $ denotes the likelihood of $t$ being a GW-tree of depth $d-1$ and offspring distribution $\mathrm{Poi}(\mu)$.\\
	Under the independence hypothesis, the computation is more immediate:
	\begin{align*}
		\Pindd(t, t') &= \mathrm{GW}_{d}^{(\lambda s)}(t) \, \mathrm{GW}_{d}^{(\lambda s')}(t')\\
		&= \pi_{\lambda s}(c) \Big(\prod_{i=1}^c \mathrm{GW}_{d-1}^{(\lambda s)}(t_{[i]})\Big) \pi_{\lambda s'}(c') \Big(\prod_{i=1}^{c'} \mathrm{GW}_{d-1}^{(\lambda s')}(t'_{[i]})\Big).
	\end{align*}
	Putting both results together, we obtain
	\begin{align}\label{lahaut}
		&L_d(t, t') = \nonumber\\
		& \sum_{k=0}^{c \wedge c'} \frac{\pi_{\lambda s s'}(k) \pi_{\lambda s (1-s')}(c-k) \pi_{\lambda s' (1-s)}(c'-k)}{\pi_{\lambda s}(c) \,\pi_{\lambda s'}(c') \, c! \,c'!} \sum_{\substack{\sigma \in \mathfrak{S}_c,\\\sigma' \in \mathfrak{S}_{c'}}} \Big(\prod_{i = 1}^{k} \Pcorrdminus(t'_{[\sigma(i)]}, t'_{[\sigma'(i)]})\Big) \nonumber\\
		&\qquad \qquad  \Big(\prod_{i=1}^{k}\mathrm{GW}_{d-1}^{(\lambda s)}(t_{[\sigma(i)]})\Big)^{-1}\Big(\prod_{i=1}^{k}\mathrm{GW}_{d-1}^{(\lambda s')}(t'_{[\sigma'(i)]})\Big)^{-1}\nonumber\\
		=& \sum_{k=0}^{c \wedge c'} \frac{\pi_{\lambda s s'}(k) \pi_{\lambda s (1-s')}(c-k) \pi_{\lambda s' (1-s)}(c'-k)}{\pi_{\lambda s}(c) \,\pi_{\lambda s'}(c') \, c! \,c'!} \sum_{\substack{\sigma \in \mathfrak{S}_c,\\\sigma' \in \mathfrak{S}_{c'}}} \prod_{i=1}^k L_{d-1}(t_{[\sigma(i)]}, t'_{[\sigma'(i)]})\nonumber\\
		&= \sum_{k=0}^{c \wedge c'} \frac{\psi(k, c, c')}{(c-k)!(c'-k)!} \sum_{\substack{\sigma \in \mathfrak{S}_c,\\\sigma' \in \mathfrak{S}_{c'}}} \prod_{i=1}^k L_{d-1}(t_{[\sigma(i)]}, t'_{[\sigma'(i)]})
	\end{align}
	where the term $\psi$, which is implicitly defined in the last equation, can be simplified to
	\begin{align*}
		&\psi(k, c, c') \overset{\text{def.}}{=\joinrel=} \frac{\pi_{\lambda s s'}(k) \pi_{\lambda s (1-s')}(c-k) \pi_{\lambda s' (1-s)}(c'-k)}{\pi_{\lambda s}(c) \,\pi_{\lambda s'}(c') } \, \frac{(c-k)!(c'-k)!}{c! \,c'!}\\
		&= \exp\Big(-\lambda s s' -\lambda s (1-s') -\lambda s' (1-s) + \lambda s + \lambda s'\Big) \frac{1}{k!} \lambda^{-k} (1-s')^{c - k} (1-s)^{c'-k}\\
		&= \exp(\lambda s s')\frac{1}{\lambda^k \, k!} \, (1-s')^{c - k} (1-s)^{c'-k}.
	\end{align*}
	One can obtain a more compact version the equation (\ref{lahaut}) by realizing that $|\mathfrak{S}_c| = (c-k)! \times |\{\sigma: [k] \to [c] \, : \, \sigma \text{ injective}\}|$. Denoting by $ \sum_{\sigma:[k] \hookrightarrow [c]}$ the sum over all injective mappings from $[k]$ to $[c]$, we therefore obtain the desired result:
	\begin{equation*}
		L_d(t, t') = \sum_{k=0}^{c \wedge c'} \psi(k, c, c') \sum_{\substack{\sigma:[k] \hookrightarrow [c],\\ \sigma':[k] \hookrightarrow [c']}} \prod_{i=1}^k L_{d-1}(t_{[\sigma(i)]}, t'_{[\sigma'(i)]}).
	\end{equation*}
\end{proof}

\subsection{Proof of Proposition \ref{prop_explicit_likhood}}\label{appendix:prop_explicit_likhood}
We recall the lower bound on $ L_{d+k} $ before moving to its proof.
\begin{proposition}[copy of Proposition \ref{prop_explicit_likhood}]
	Let \( (t, t') \) be a pair of rooted trees sharing a common rooted subtree \( t_* \). Denote by 
	\( \sigma_*: t_* \hookrightarrow t \) and \( \sigma'_*: t_* \hookrightarrow t' \) 
	the embeddings of \( t_* \) into \( t \) and \( t' \), respectively. 
	
	For each \( i \in t \), let \( c_t(i) \) denote its number of child nodes in \( t \), and define 
	\( t_{[i]} \) as the subtree rooted at \( i \), containing all its descendants. 
	This notation extends naturally to \( t' \) and \( t_* \), and we recall that \( (t_*)_d \) denotes the depth-\( d \) truncation of \( t_* \).
	
	Then, for all \( d, k \in \mathbb{N}_{>0} \), the following inequality holds:
	\[
	L_{d+k}(t, t') \geq 
	\prod_{i \in (t_*)_{d-1}} \hspace{-8pt}
	\psi\Big(c_{t_*}(i), c_t(\sigma_*(i)), c_{t'}(\sigma'_*(i)) \Big) 
	\prod_{j \in (t_*)_d \setminus (t_*)_{d-1}} \hspace{-8pt}
	L_k(t_{[\sigma_*(j)]}, t'_{[\sigma'_*(j)]})
	\]
\end{proposition}
\begin{proof}
	We start by iteratively applying the recursive likelihood formula from Proposition \ref{prop_likhood_recursive} to obtain
	\begin{equation}\label{explicit_likelihood}
		L_d(t, t') = \sum_{\tau \in \mathcal{X}_d} \sum_{\substack{\sigma: \tau \hookrightarrow t,\\\sigma' : \tau \hookrightarrow t'}} \prod_{i \in \tau_{d-1}} \psi\Big(c_\tau(i), \, c_t(\sigma(i)), \, c_{t'}(\sigma'(i))\Big).
	\end{equation}
	For a rigorous proof of this equation, we refer the reader to  \cite[Lemma 2.1]{luca} whose proof translates word by word to the asymmetric tree correlation setting. 
	
	Next, let $t, t' \in \mathcal{X}_{d+k}$ be arbitrary trees that share a common subtree $t_* \in \mathcal{X}_d$ up to depth $d$. Denote by $\sigma_*: t_* \hookrightarrow t$ and $\sigma'_*: t_* \hookrightarrow t'$ the respective subtree embeddings. One has the following sequence of computations; each step will be explained right after:
	\begin{align*}
		&L_{d+k}(t, t') = \sum_{\tau \in \mathcal{X}_{d+k}}  \sum_{\substack{\sigma: \tau \hookrightarrow t,\\\sigma' : \tau \hookrightarrow t'}} \prod_{i \in \tau_{d+k-1}} \psi\Big(c_\tau(i), \, c_t(\sigma(i)), \, c_{t'}(\sigma'(i))\Big)\\
		& \ \geq \sum_{\substack{\tau \in \mathcal{X}_{d + k}\\ \tau_d = t_*}} \quad \sum_{\substack{\sigma: \tau \hookrightarrow t, \, \sigma = \sigma_* \text{ on } \tau_d \\ \sigma' : \tau \hookrightarrow t', \, \sigma' = \sigma'_* \text{ on } \tau_d}} \quad \prod_{i \in (t_*)_{d-1}} \psi\Big(c_{t_*}(i), \, c_t(\sigma_*(i)), \, c_{t'}(\sigma'_*(i))\Big) \\
		& \ \quad \quad \quad \quad \times \prod_{j \in (t_*)_{d} \setminus (t_*)_{d-1}} \prod_{m \in (\tau_{[j]})_{k-1}} \psi \Big(c_\tau (m), \, c_{t}(\sigma(m)), \, c_{t'}(\sigma'(m))\Big)\\
		&\ = \bigg( \prod_{i \in (t_*)_{d-1}} \psi\Big(c_{t_*}(i), \, c_t(\sigma_*(i)), \, c_{t'}(\sigma'_*(i))\Big) \bigg) \sum_{ (\tilde{\tau}_\ell)_\ell \in (\mathcal{X}_k)^{(t_*)_{d} \setminus (t_*)_{d-1}} }\\
		&\  \qquad  \sum_{
			\substack{
				(\tilde{\sigma}_\ell)_\ell, \, (\tilde{\sigma}'_\ell)_\ell\, :\\
				\tilde{\sigma}_\ell: \tilde{\tau}_\ell \hookrightarrow t_{[{\sigma}_*(\ell)]} \\ \tilde{\sigma}'_\ell: \tilde{\tau}_\ell \hookrightarrow t'_{[{\sigma}'_*(\ell)]}
			}
		}  \ \prod_{j \in (t_*)_{d} \setminus (t_*)_{d-1} } 
		\quad  \prod_{m \in (\tilde{\tau}_j)_{k-1} } \psi\Big(c_{\tilde{\tau}_j}(m), \, c_{t}(\tilde{\sigma}_j(m)), \, c_{t'}(\tilde{\sigma}'_j(m))\Big)\\
		&\  = \bigg( \prod_{i \in (t_*)_{d-1}} \psi\Big(c_{t_*}(i), \, c_t(\sigma_*(i)), \, c_{t'}(\sigma'_*(i))\Big) \bigg) \\
		&\ \quad \quad \times \prod_{j \in (t_*)_{d} \setminus (t_*)_{d-1}} \sum_{\tilde{\tau} \in \mathcal{X}_k} \sum_{\substack{
				\tilde{\sigma}: \tilde{\tau} \hookrightarrow t_{[\sigma_*(j)]} \\ \tilde{\sigma}': \tilde{\tau} \hookrightarrow t'_{[\sigma'_*(j)]}
		}} \prod_{m \in \tilde\tau_{k-1}} \psi\Big(c_{\tilde\tau}(m), \, c_{t}(\tilde{\sigma}(m)), \, c_{t'}(\tilde{\sigma}'(m))\Big)\\
		&\ = \prod_{i \in (t_*)_{d-1}} \psi\Big(c_{t_*}(i), \, c_t(\sigma_*(i)), \, c_{t'}(\sigma'_*(i))\Big)  \, 
		\prod_{j \in (t_*)_{d} \setminus (t_*)_{d-1}} L_k(t_{[\sigma_*(j)]}, t'_{[\sigma'_*(j)]}).
	\end{align*}
	Here are the explanations for every step of this computation:
	\begin{itemize}
		\item The first equality is an application of (\ref{explicit_likelihood}) at depth $d+k$.
		\item The inequality follows from $\psi \geq 0$ and the omission of some terms in the two sums.
		\item For the next equality, note that the sum over $ \{\tau \in \mathcal{X}_{d+k} \suchthat \tau_d = t_*\} $ is equivalent to summing over all possible combinations of subtrees $(\tilde{\tau}_\ell)_\ell$ which are elements of $\mathcal{X}_k$ and rooted in their index nodes $\ell \in (t_*)_d \setminus (t_*)_{d-1}$. The injective mappings decompose across these subtrees.
		\item Next up is a swap of product and sum that simplifies the multinomial expression from before.
		\item Finally, the last equality is another application of (\ref{explicit_likelihood}) to the second product.
	\end{itemize}
	The final expression is what we wanted to show.
\end{proof}

\subsection{Proof of theorem \ref{thm_diagonalization}}\label{appendix:thm_diagonalization}
We recall the likelihood ratio diagonalization formula before moving to its proof.
\begin{theorem}[copy of Theorem \ref{thm_diagonalization}]
	There exists a set of functions $f^{(\mu)}_{d, \beta}:\mathcal{X}_d \to \R$ with parameters $\mu >0, d\in \N$ and indexed by trees $\beta \in \mathcal{X}_d$ such that for all choices of $s, s' \in [0,1]$, $\lambda >0$, the following \textbf{likelihood ratio diagonalization} formula holds:
	\begin{equation}\label{appendix:likhoodratio_bigformula}
		\forall t, t' \in \mathcal{X}_d\, : \, L_d(t, t') = \sum_{\beta\in\mathcal{X}_d} \sqrt{ss'}^{|\beta|-1} f^{(\lambda s)}_{d, \beta}(t) f^{(\lambda s')}_{d, \beta}(t').
	\end{equation}
	Furthermore, the $f^{(\mu)}_{d, \beta}$ have the following properties:
	\begin{itemize}
		\item Constant value for $\beta = \bullet$, the trivial tree:
		\begin{equation}\label{appendix:trivialtree}
			\forall \mu >0, \,  \forall t \in \mathcal{X}_d: \quad f^{(\mu)}_{d, \bullet}(t)  = 1.
		\end{equation}
		\item \textbf{First orthogonality property} (w.r.t the Galton--Watson-measure):
		\begin{equation}\label{appendix:orthogonaliy_1_GW}
			\forall \mu >0, \, \forall \beta, \beta' \in \mathcal{X}_d: \quad \sum_{t \in \mathcal{X}_d} \mathrm{GW}_d^{(\mu)}(t) \, f^{(\mu)}_{d, \beta}(t) \, f^{(\mu)}_{d, \beta'}(t) \, =\, \indicator{\beta = \beta'}.
		\end{equation}
		\item \textbf{Second orthogonality property} (summing over $\beta$):
		\begin{equation}\label{appendix:orthogonal_2_betasum}
			\forall \mu >0, \, \forall t, t' \in \mathcal{X}_d: \quad \sum_{\beta \in \mathcal{X}_d} f^{(\mu)}_{d, \beta}(t) \, f^{(\mu)}_{d, \beta}(t') \,=\, \frac{\indicator{t = t'}}{\mathrm{GW}_d^{(\mu)}(t)}.
		\end{equation}
	\end{itemize}
\end{theorem}
\begin{proof}
	We prove the theorem by induction on the depth $d$.
	
	\underline{$d = 0$:}\\	
	Since $\mathcal{X}_0 = \{\bullet\}$ only contains the trivial tree, we have $\Pcorr_0(\bullet, \bullet) = 1$, $\Pind_0(\bullet, \bullet) = 1$ and therefore $L_0(\bullet, \bullet) = 1$. Setting $f_{0, \bullet}^{(\mu)}(\bullet) := 1$, which naturally yields the property of constant value for the trivial tree, we have
	\[
	\sum_{\beta\in\mathcal{X}_0} \sqrt{ss'}^{|\beta|-1} f^{(\lambda s)}_{0, \beta}(\bullet) f^{(\lambda s')}_{0, \beta}(\bullet) = \sqrt{ss'}^{\, 0} f^{(\lambda s)}_{0, \bullet}(\bullet) f^{(\lambda s')}_{d, \bullet}(\bullet) = 1 = L_0(\bullet, \bullet).
	\]
	This proves the likelihood ratio diagonalization formula (\ref{appendix:likhoodratio_bigformula}) for $d = 0$. Further note that the first orthogonality property(\ref{appendix:orthogonaliy_1_GW}) is trivial since $\mathcal{X}_0$ only contains one element.
	
	\underline{$d \mapsto d+1$:}\\
	Suppose that the likelihood ratio diagonalization formula (\ref{appendix:likhoodratio_bigformula}) holds true for $d \in \N_{\geq 0}$. Fixing a pair of trees $(t, t') \in \mathcal{X}_{d+1}$, our goal is to prove that same equation at $d+1$, which can be rewritten as
	\begin{equation}\label{goal_induction}
		\Pcorrdplus(t, t') = \underbrace{\mathrm{GW}^{(\lambda s)}_{d+1}(t) \; \mathrm{GW}^{(\lambda s')}_{d+1}(t') }_{=\Pinddplus(t, t')}\, \times \, \underbrace{\sum_{\beta\in\mathcal{X}_{d+1}} \sqrt{ss'}^{|\beta|-1} f^{(\lambda s)}_{d+1, \beta}(t) f^{(\lambda s')}_{d+1, \beta}(t')}_{ =L_{d+1}(t, t')}
	\end{equation}
	where $f_{d+1, \beta}^\mu$ will be defined in the process. \\
	In the following, we identify the fixed tree pair $(t, t')$ with their subtree tuples as in Definition \ref{def:unlabeled_rooted_trees}. This encoding shall be denoted as $t = (N_\tau)_{\tau \in \mathcal{X}_d}$ and $t' = (N'_\tau)_{\tau \in \mathcal{X}_d}$. Since $t$ and $t'$ are finite trees, the maps $\tau \mapsto N_\tau$ and $\tau \mapsto N'_\tau$ are non-zero only for finitely many different subtrees $\tau$. Let $p, p' \in \N$ denote the respective numbers of non-zero coordinates.
	
	Our strategy for showing (\ref{goal_induction}) is to Fourier-invert the characteristic function of $\Pcorrdplus$. In the following, we denote by $(r, r') \sim \Pcorrdplus$ a random pair of correlated trees which we also identify with their subtrees tuples $ (r, r') = (M_\beta, M'_\beta)_{\beta \in \mathcal{X}_d} $. Note that $t, t', N_\tau, N'_\tau$ are fixed while $r, r', M_\tau, M'_\tau$ denote random variables. We define the notation $u\cdot r := \sum_{\tau \in \mathcal{X}_d} u_\tau M_\beta$ for any sequence $u \in \R^{\mathcal{X}_d}$. The characteristic function of $\Pcorrdplus$ can then be written as
	\begin{equation*}
		\widehat\Pcorrdplus :  \R^{\mathcal{X}_d} \times \R^{\mathcal{X}_d} \to \C, \, (u, u') \mapsto \Ecorrdplus\Big[e^{i u \cdot r + i u' \cdot r'}\Big].
	\end{equation*}
	This function is well-defined since for almost all $\tau \in \mathcal{X}_d$, we have $(M_\tau, M'_\tau) = (0,0)$, which avoids the problem of infinite sums in the exponent of $e$. One can easily recover $ \Pcorrdplus $ from $ \widehat\Pcorrdplus $ using the following Fourier-inversion type formula:
	\begin{align}\label{fourier_inversion}
		\Pcorrdplus(t, t') &= \Ecorrdplus\Big[\indicator{r = t} \times \indicator{ r' = t'}\Big] \nonumber \\
		&= \Ecorrdplus\Big[\indicator{\{\forall \tau: \, N_\tau - M_\tau = 0\}} \times \indicator{ \{\forall \tau:\, N'_\tau - M'_\tau = 0\}}\Big] \nonumber \\
		&=  \Ecorrdplus\Big[ \frac{1}{(2\pi)^{p}}\int_{[0, 2\pi]^p} e^{i u \cdot (r-t)} du \, \times\,  \frac{1}{(2\pi)^{p'}} \int_{[0, 2\pi]^{p'}} e^{ i u' \cdot (r' - t')}du' \Big]\nonumber \\
		&= \frac{1}{(2\pi)^{p + p'}}\int_{[0, 2\pi]^p} \int_{[0, 2\pi]^{p'}} e^{-iu\cdot t- iu'\cdot t'} \Pcorrdplus(u, u') du \, du'  
	\end{align}
	Here we have used that the number of non-zero coordinates in $t = (N_\tau)_{\tau \in \mathcal{X}_d}$ and $t'$ equals $p$ and $p'$ respectively. The last equality uses Fubini's Theorem together with the fact $|e^{i u \cdot (r-t)} \Pcorrdplus(u, u')| \leq 1$.
	
	To manipulate the term $\widehat\Pcorrdplus(u, u')$ first recall  that $N_\tau = \Delta_\tau + \sum_{\tau' \in \mathcal{X}_d} \Gamma_{\tau, \tau'}$ and $N'_{\tau'} = \Delta'_{\tau'} + \sum_{\tau \in \mathcal{X}_d} \Gamma_{\tau, \tau'}$ with
	\begin{align*}
		\Delta_\tau \sim \mathrm{Poi}\Big(\lambda s (1-s') \,\mathrm{GW}_d^{(\lambda s)}(\tau)\Big)&, \Delta'_{\tau'} \sim \mathrm{Poi}\Big(\lambda s' (1-s) \, \mathrm{GW}_d^{(\lambda s')}(\tau')\Big), \\ 
		\text{ and } \Gamma_{\tau, \tau'} &\sim \mathrm{Poi}\Big(\lambda s s' \, \Pcorrd(\tau, \tau')\Big)
	\end{align*}
	which are all independent poisson variables. This lets us compute
	\begin{align*}
		\Expe\Big[&e^{i u \cdot r + i u' \cdot r'}\Big] = \Expe\Big[e^{i u \cdot \Delta + \sum_{\tau'} u \cdot \Gamma_{:, \tau'} } e^{i u' \cdot \Delta' + \sum_{\tau} u' \cdot \Gamma_{\tau, :} }\Big]\\
		&= \prod_{\tau \in \mathcal{X}_d} \Expe\Big[e^{i u_\tau \Delta_\tau }\Big]  \prod_{\tau \in \mathcal{X}_d} \Expe\Big[e^{i u'_\tau \Delta'_\tau }\Big] \prod_{\tau, \tau'  \in \mathcal{X}_d} \Expe\Big[e^{i (u_\tau + u'_{\tau'}) \,\Gamma_{\tau, \tau'} }\Big]\\
		&= \exp\Big[\sum_{\tau} \lambda s (1-s') \mathrm{GW}_d^{(\lambda s)}(\tau) (e^{iu_\tau} - 1)\\ 
		& \quad \quad \quad \quad+ \sum_{\tau'} \lambda s' (1-s) \mathrm{GW}_d^{(\lambda s')}(\tau') 	\Big(e^{iu'_{\tau'}} - 1\Big) \\
		& \quad \quad \quad \quad + \sum_{\tau, \tau' } \lambda s s' \,  \Pcorrd(\tau, \tau') \Big(e^{i(u_\tau+u'_{\tau'})} - 1\Big) \Big]\\
		&\overset{(\star)}{=} \exp\Big[\lambda s \sum_\tau \mathrm{GW}_d^{(\lambda s)}(\tau) (e^{iu_\tau} - 1)\Big]  \exp\Big[\lambda s' \sum_\tau \mathrm{GW}_d^{(\lambda s')}(\tau)(e^{iu'_\tau} - 1)\Big] \\
		& \quad \quad \quad \quad \times \exp\Big[\lambda s s' \sum_{\tau, \tau' }  \Pcorrd(\tau, \tau')(e^{iu_\tau} - 1) \Big( e^{iu'_{\tau'}} - 1\Big) \Big]\\
		&= \exp\Big[\lambda s \sum_\tau \mathrm{GW}_d^{(\lambda s)}(\tau)(e^{iu_\tau} - 1)\Big]  \exp\Big[\lambda s' (\tau) \sum_\tau \mathrm{GW}_d^{(\lambda s')} (e^{iu'_\tau} - 1)\Big] \\
		& \quad \quad \quad \quad \times \sum_{m=0}^{\infty} \frac{(\lambda s s')^m}{m!} \Big(\sum_{\tau, \tau' }  \Pcorrd(\tau, \tau')(e^{iu_\tau} - 1) \Big( e^{iu'_{\tau'}} - 1\Big) \Big)^m
	\end{align*}
	To obtain the equality $ (\star) $, one needs to combine all terms with factor $ \lambda s s' $ by using that the marginals of $ \Pcorrd $ fulfill
	\[
	\sum_{\tau, \tau'} \Pcorrd(\tau, \tau') = \sum_\tau \mathrm{GW}_d^{(\lambda s')}(\tau) = \sum_{\tau'} \mathrm{GW}_d^{(\lambda s')}(\tau').
	\]
	To shorten the following explanations, we introduce some notation: Define \[
	A^{(\mu)}_d(u) := \exp\Big[\mu \sum_\tau \mathrm{GW}_d^{(\mu)}(\tau)(e^{iu_\tau} - 1)\Big] \]
	and \[
	B_d(u, u') := \sum_{\tau, \tau' }  \Pcorrd(\tau, \tau')(e^{iu_\tau} - 1) \Big( e^{iu'_{\tau'}} - 1\Big),
	\]
    which lets us write 
	\begin{align*}
		\widehat\Pcorrdplus(u, u') = \Expe\Big[e^{i u \cdot r + i u' \cdot r'}\Big] &= A^{(\lambda s)}_d(u) \, A^{(\lambda s')}_d(u') \sum_{m=0}^{\infty} \frac{(\lambda s s')^m}{m!} \Big(B_d(u, u')\Big)^m\\
		&= \sum_{m=0}^{\infty} \frac{(\lambda s s')^m}{m!} \Big(\sqrt[m]{A^{(\lambda s)}_d(u) \, A^{(\lambda s')}_d(u') }B_d(u, u')\Big)^m.
	\end{align*}
	Note that since $|e^{iu_\tau} - 1| \leq 2$, one can bound $|B_d(u, u')| \leq 4$ as well as 
	\begin{equation}\label{bound_on_Aus}
		\Big| A^{(\lambda s)}_d(u) \, A^{(\lambda s')}_d(u') \Big| \leq e^{2\lambda s + 2\lambda s'} \implies \Big| \sqrt[m]{A^{(\lambda s)}_d(u) \, A^{(\lambda s')}_d(u')} \Big| \leq e^{2\lambda s + 2\lambda s'}.
	\end{equation}
	Consequently we have the following bound which is uniform in $(u, u')$ and summable in $m$:
	\[
	\Big| \frac{(\lambda s s')^m}{m!} \Big(\sqrt[m]{A^{(\lambda s)}_d(u) \, A^{(\lambda s')}_d(u') }B_d(u, u')\Big)^m \Big| \leq \frac{1}{m!}\Big(4 \, ss' \, \lambda e^{2\lambda s + 2\lambda s'}\Big)^m 
	\]
	which lets us apply Fubini's Theorem to obtain
	\begin{equation}
	\begin{aligned}
		&\Pcorrdplus(t, t') \overset{\text{by }(\ref{fourier_inversion})}{=\joinrel=} \frac{1}{(2\pi)^{p + p'}} \int_{[0, 2\pi]^{p+p'}}e^{-iu\cdot t- iu'\cdot t'} \\
		&\qquad \qquad \qquad \qquad \qquad \times \sum_{m=0}^{\infty} \frac{(\lambda s s')^m}{m!} \Big(\sqrt[m]{A^{(\lambda s)}_d(u) \, A^{(\lambda s')}_d(u') }B_d(u, u')\Big)^m du \, du' \\
		&= \sum_{m=0}^{\infty} \frac{(\lambda s s')^m}{m!} \frac{1}{(2\pi)^{p + p'}}\\
		&\qquad \qquad \times \int_{[0, 2\pi]^p} \int_{[0, 2\pi]^{p'}} e^{-iu\cdot t- iu'\cdot t'} A^{(\lambda s)}_d(u) \, A^{(\lambda s')}_d(u') \, \Big( B_d(u, u')\Big)^m du \, du'.
	\end{aligned}\label{Pdcorr_as_series}
	\end{equation}
	As a next step, we develop the $m^\text{th}$ power of the term $B_d(u, u')$. For this, we use the induction hypothesis at depth $d$ under the same form as in $(\ref{goal_induction})$, namely
	\begin{align}\label{ind_hyp_d}
		\Pcorrd(\tau, \tau') &= \mathrm{GW}^{(\lambda s)}_d(\tau) \mathrm{GW}^{(\lambda s')}_d(\tau') \sum_{\beta\in\mathcal{X}_d} \sqrt{ss'}^{|\beta|-1} f^{(\lambda s)}_{d, \beta}(\tau) f^{(\lambda s')}_{d, \beta}(\tau')\nonumber\\
		&=   \sum_{\beta\in\mathcal{X}_d} \sqrt{ss'}^{|\beta|-1} \underbrace{f^{(\lambda s)}_{d, \beta}(\tau)\, \mathrm{GW}^{(\lambda s)}_d(\tau)}_{=: g^{(\lambda s)}_{d, \beta}(\tau)} \underbrace{f^{(\lambda s')}_{d, \beta}(\tau')\, \mathrm{GW}^{(\lambda s')}_d(\tau')}_{= g^{(\lambda s')}_{d, \beta}(\tau')}.
	\end{align}
	In the Fourier-inversion formula (\ref{fourier_inversion}), only finitely many $u_\tau$ and $u'_\tau$ are non-zero. Consequently, the sum $\sum _{\tau, \tau'} $ appearing in $B_d(u, u') $  is finite and we may exchange this summation with the series $\sum_{\beta \in \mathcal{X}_d}$ from equation (\ref{ind_hyp_d}). This lets us compute 
	\begin{align*}
		B_d(u, u') & = \sum_{\beta \in \mathcal{X}_d} \sum_{\tau, \tau' \in \mathcal{X}_d}  \sqrt{ss'}^{|\beta| - 1} g^{(\lambda s)}_{d, \beta}(\tau) \,(e^{iu_\tau} - 1) \, g^{(\lambda s')}_{d, \beta}(\tau') \Big( e^{iu'_{\tau'}} - 1\Big)\\
		&= \sum_{\beta \in \mathcal{X}_d}  \sqrt{ss'}^{|\beta| - 1} \underbrace{\sum_{\tau \in \mathcal{X}_d} g^{(\lambda s)}_{d, \beta}(\tau) \, (e^{iu_\tau} - 1)}_{=: h^{(\lambda s)}_{d, \beta}(u) } \, \underbrace{\sum_{\tau' \in \mathcal{X}_d} g^{(\lambda s')}_{d, \beta}(\tau') \Big( e^{iu'_{\tau'}} - 1\Big)}_{=h^{(\lambda s')}_{d, \beta}(u')}.
	\end{align*} 
	For the $h^{(\mu)}_{d, \beta}(u)$ defined through the last equation, we have the following bound:
	\begin{align}\label{bound_on_ghat}
		\Big| h^{(\mu)}_{d, \beta}(u) \Big| &\leq  \sum_{\tau \in \mathcal{X}_d} \Big| f^{(\mu)}_{d, \beta}(\tau)\Big| \,\mathrm{GW}_d^{(\mu)}(\tau) \, \Big|e^{iu_\tau} - 1\Big|  \nonumber\\
		&\leq 2 \Expe_{\tau \sim \mathrm{GW}^{(\mu)}_d}\Big[\Big| f^{(\mu)}_{d, \beta}(\tau)\Big|\Big] \overset{\text{Jensen}}{\leq} 2 \sqrt{\Expe_{\tau \sim \mathrm{GW}^{(\mu)}_d}\Big[\Big(f^{(\mu)}_{d, \beta}(\tau)\Big)^2 \Big]} \overset{\text{by } (\ref{appendix:orthogonaliy_1_GW})}{=}2.
	\end{align}
	Taking the expression we have obtained for $B_d(u, u')$ to the power $m$ and using the multinomial cauchy product formula, we obtain
	\begin{align*}
		&\Big(B_d(u, u')\Big)^m\\
		& = \sum_{\substack{\gamma = (\gamma_\beta)_{\beta \in \mathcal{X}_d}, \,\\ \sum_\beta \gamma_\beta = m}} \frac{m!}{\prod_\beta \gamma_\beta!}  \prod_{\beta \in \mathcal{X}_d } \Big(\sqrt{ss'}^{\, -1 + |\beta| }\Big)^{\gamma_\beta} h^{(\lambda s)}_{d, \beta}(u) ^{\gamma_\beta}  h^{(\lambda s')}_{d, \beta}(u')^{\gamma_\beta}\\
		& = m! \sum_{\substack{\gamma = (\gamma_\beta)_{\beta \in \mathcal{X}_d},\, \\ \sum_\beta \gamma_\beta = m}} \sqrt{ss'}^{\, -\sum_\beta \gamma_\beta + \sum_{\beta} \gamma_\beta |\beta| } \prod_{\beta \in \mathcal{X}_d } \frac{1}{\sqrt{\gamma_\beta !}}  h^{(\lambda s)}_{d, \beta}(u){\gamma_\beta} \frac{1}{\sqrt{ \gamma_\beta !}} h^{(\lambda s')}_{d, \beta}(u')^{\gamma_\beta}\\
		& = \frac{m!}{(\lambda ss')^m} \sum_{\substack{\gamma = (\gamma_\beta)_{\beta \in \mathcal{X}_d}:\\ \sum_\beta \gamma_\beta = m}} \sqrt{ss'}^{\sum_{\beta} \gamma_\beta |\beta| } \prod_{\beta \in \mathcal{X}_d } \frac{\sqrt{\lambda s}^{\, \gamma_\beta} }{\sqrt{ \gamma_\beta !}} h^{(\lambda s)}_{d, \beta}(u)^{\gamma_\beta}\,  \frac{\sqrt{\lambda s'}^{\, \gamma_\beta} }{\sqrt{ \gamma_\beta !}} h^{(\lambda s')}_{d, \beta}(u')^{\gamma_\beta}.
	\end{align*}
	Inserting this into equation (\ref{Pdcorr_as_series}) yields
	\begin{align}\label{development_sumgamma}
		&\frac{(\lambda ss')^m}{m! (2\pi)^{p + p'}} \int_{[0, 2\pi]^p} \int_{[0, 2\pi]^{p'}}e^{-iu\cdot t- iu'\cdot t'} A^{(\lambda s)}_d(u) A^{(\lambda s')}_d(u') \, \Big( B_d(u, u')\Big)^m du du' \nonumber\\
		&= \sum_{\substack{\gamma = (\gamma_\beta)_{\beta \in \mathcal{X}_d},\, \\ \sum_\beta \gamma_\beta = m}} \sqrt{ss'}^{\sum_{\beta} \gamma_\beta |\beta| } \frac{1}{(2\pi)^{p}}  \int_{[0, 2\pi]^{p}}  e^{- iu\cdot t} A^{(\lambda s)}_d(u)  \prod_{\beta \in \mathcal{X}_d} \frac{\sqrt{\lambda s}^{\, \gamma_\beta} }{\sqrt{ \gamma_\beta !}} h^{(\lambda s)}_{d, \beta}(u)^{\gamma_\beta}du\nonumber\\
		&\quad \quad \quad \times \frac{1}{(2\pi)^{p'}}  \int_{[0, 2\pi]^{p'}} e^{- iu'\cdot t'} A^{(\lambda s')}_d(u')  \prod_{\beta \in \mathcal{X}_d} \frac{\sqrt{\lambda s'}^{\, \gamma_\beta} }{\sqrt{ \gamma_\beta !}} \Big(h^{(\lambda s')}_{d, \beta}(u')\Big)^{\gamma_\beta} du'
	\end{align}
	where swapping the summation in $(\gamma_\beta)_{\beta \in \mathcal{X}_d}$ with the integrals in $(u, u')$ is enabled by Fubini's Theorem and the fact that
	\begin{align*}
		&\sum_{\substack{\gamma = (\gamma_\beta)_{\beta \in \mathcal{X}_d}:\\ \sum_\beta \gamma_\beta = m}} \Big| A^{(\lambda s)}_d(u) \, A^{(\lambda s')}_d(u') \sqrt{ss'}^{\sum_{\beta} \gamma_\beta |\beta| } \prod_{\beta \in \mathcal{X}_d } \frac{\sqrt{\lambda ^2 ss'}^{\, \gamma_\beta} }{ \gamma_\beta !} h^{(\lambda s)}_{d, \beta}(u)^{\gamma_\beta} \, h^{(\lambda s')}_{d, \beta}(u')^{\gamma_\beta}   \Big| \\
		&  \overset{(\ref{bound_on_Aus})}{\leq} \frac{e^{2\lambda (s+s')}}{m!}  \sum_{\substack{\gamma = (\gamma_\beta)_{\beta \in \mathcal{X}_d}:\\ \sum_\beta \gamma_\beta = m}} \frac{m!}{\prod_{\beta \in \mathcal{X}_d} \gamma_\beta!} \prod_{\beta \in \mathcal{X}_d } \sqrt{ss'}^{ \, |\beta|}\Big| \sqrt{\lambda s}\,h^{(\lambda s)}_{d, \beta}(u) \, \sqrt{\lambda s'} h^{(\lambda s')}_{d, \beta}(u') \Big|^{\gamma_\beta} \\
		&  \quad = e^{2\lambda (s+s')}  \frac{\Big(\lambda \sqrt{ss'}\Big)^m}{m!} \Big(\sum_{\beta \in \mathcal{X}_d}  \sqrt{s s'}^{\, |\beta|}\,\Big|h^{(\lambda s)}_{d, \beta}(u)\Big| \,\Big| h^{(\lambda s')}_{d, \beta}(u') \Big|  \Big)^m\\
		&   \overset{(\ref{bound_on_ghat})}{\leq} 4 e^{2\lambda (s+s')}  \frac{\Big(\lambda \sqrt{ss'}\Big)^m}{m!} \sqrt{s s'}^{\, m} \Big(\sum_{n = 1}^\infty \Big|\{\beta \in \mathcal{X}_d \suchthat |\beta| = n\}\Big|  \sqrt{s s'}^{\, n - 1} \Big)^m.
	\end{align*}
	The final expression is finite due to Otter's Proposition \ref{Otter} and $s s' < 1$. It is furthermore uniform in $u, u'$ and integrating over $(u, u') \in [0, 2\pi]^p\times [0, 2\pi]^{p'}$ yields a finite term, which validates the application of Fubini's Theorem in (\ref{development_sumgamma}). 
	
	Now define
	\begin{align*}
		g^{(\mu)}_{d+1, \gamma}(u) &:= \frac{1}{(2\pi)^{p}}  \int_{[0, 2\pi]^{p}} e^{- iu\cdot t} A^{(\mu)}_d(u)  \prod_{\beta \in \mathcal{X}_d} \frac{\sqrt{\mu}^{\, \gamma_\beta} }{\sqrt{ \gamma_\beta !}} \Big(h^{(\mu)}_{d, \beta}(u)\Big)^{\gamma_\beta} du
	\end{align*}
	and combine equation (\ref{development_sumgamma}) with the expression for $\Pcorrdplus(t, t')$ from (\ref{Pdcorr_as_series}) to obtain
	\begin{align*}
		\Pcorrdplus(t, t') &= \sum_{m=1}^{\infty} \sum_{\substack{\gamma = (\gamma_\beta)_{\beta \in \mathcal{X}_d},\, \\ \sum_\beta \gamma_\beta = m}} \sqrt{ss'}^{\sum_{\beta} \gamma_\beta |\beta| } g^{(\lambda s)}_{d+1, \gamma}(u) g^{(\lambda s')}_{d+1, \gamma}(u)
	\end{align*}
	Note that every tuple $(\gamma_\beta)_{\beta \in \mathcal{X}_d}$ with $\sum_\beta \gamma_\beta =m$ corresponds to a unique tree $\gamma \in \mathcal{X}_{d+1} $ with root degree $m$ by the subtree tuple identification. Every such tree has size $|\gamma| = 1 + \sum_{\beta \in \mathcal{X}_d} \gamma_\beta |\beta|$. This observation leads to
	\begin{equation}\label{almost_done}
		\Pcorrdplus(t, t') = \sum_{\gamma \in \mathcal{X}_{d+1}} \sqrt{ss'}^{|\gamma| - 1} g^{(\lambda s)}_{d+1, \gamma}(u) g^{(\lambda s')}_{d+1, \gamma}(u)
	\end{equation}
	This decomposition of $\Pcorrdplus(t, t')$ is already similar to our objective in $ (\ref{goal_induction}) $. Hence, all that remains to do is convert $g^{(\mu)}_{d+1, \gamma}(u)$ into a product between $ \mathrm{GW}_{d+1}^{(\mu)}(\tau) $ and $ f^{(\mu)}_{d+1, \gamma}(\tau) $ where the function $f$ has the desired orthogonality properties. To do so, we start by developing the definition of $ g^{(\mu)}_{d+1, \gamma}(u) $:
	\begin{align*}
		&g^{(\mu)}_{d+1, \gamma}(u) = \frac{1}{(2\pi)^{p}}  \int_{[0, 2\pi]^{p}} e^{- iu\cdot t} A^{(\mu)}_d(u)  \prod_{\beta} \frac{\sqrt{\mu}^{\, \gamma_\beta} }{\sqrt{ \gamma_\beta !}} \Big(h^{(\mu}_{d, \beta}(u)\Big)^{\gamma_\beta} du\\
		&= \frac{1}{\prod_\beta \sqrt{\gamma_\beta!}} \frac{1}{(2\pi)^{p}} \\
        &\quad \times \int_{[0, 2\pi]^{p}} e^{- iu\cdot t} e^{\mu \sum_\tau \mathrm{GW}_d^{\mu}(\tau) \big(e^{iu_\tau} - 1\big)} \prod_{\beta} \Big( \sqrt{\mu} \sum_{\tau \in \mathcal{X}_d} g^{(\lambda s)}_{d, \beta}(\tau) \, (e^{iu_\tau} - 1) \Big) ^{\gamma_\beta} du\\
		&= \frac{1}{\prod_\beta \sqrt{\gamma_\beta!}} \frac{1}{(2\pi)^{p}}  \\
		& \quad \times \int_{[0, 2\pi]^{p}} e^{- iu\cdot t} e^{\mu \sum_\tau \mathrm{GW}_d^{\mu}(\tau) \big(e^{iu_\tau} - 1\big)} \Big(\prod_{\beta } \gamma_\beta! \Big) [x^\gamma] e^{\sum_{\beta } x_\beta \sqrt{\mu} \sum_{\tau } g^{(\mu)}_{d, \beta}(\tau) \, (e^{iu_\tau} - 1) } du.
	\end{align*}
	In the last step, we have introduced some new notation: $x = (x_\beta)_{\beta \in \mathcal{X}_d}$ denotes a tuple of formal variables, $\gamma = (\gamma_\beta)_{\beta \in \mathcal{X}_d}$ is the tuple integers counting the subtrees of $\gamma$. We set $x^\gamma := \prod_\beta x_\beta^{\gamma_\beta}$ and let $[x^\gamma] S$ denote the coefficient of $x^\gamma$ in a formal series $S$ which contains $x$ as a formal variable. Keeping in mind that we are implicitly dealing with a series, we continue computing $g^{(\mu)}_{d+1, \gamma}(u)$ which is equal to
	\begin{align*}
		&\sqrt{\prod_{\beta } \gamma_\beta !} \; \frac{1}{(2\pi)^{p}} \\
        &\qquad \times \int_{[0, 2\pi]^{p}} e^{- iu\cdot t} [x^\gamma] e^{\sum_\tau \big(e^{iu_\tau} - 1\big) \mu \mathrm{GW}_d^{\mu}(\tau) } \, e^{ \sum_{\tau } \big(e^{iu_\tau} - 1\big)  \sum_{\beta } x_\beta \sqrt{\mu} g^{(\mu)}_{d, \beta}(\tau) } du\\
		&= \sqrt{\prod_{\beta } \gamma_\beta !} \; \frac{1}{(2\pi)^{p}}   \int_{[0, 2\pi]^{p}} e^{- iu\cdot t} [x^\gamma]  e^{ \sum_{\tau } \big(e^{iu_\tau} - 1\big) \Big[\mu \mathrm{GW}_d^{(\mu)}(\tau) + \sum_{\beta } x_\beta \sqrt{\mu} g^{(\mu)}_{d, \beta}(\tau) \Big]  } du\\
		&= \sqrt{\prod_{\beta } \gamma_\beta !} \;   [x^\gamma] \frac{1}{(2\pi)^{p}}   \int_{[0, 2\pi]^{p}} e^{- iu\cdot t}e^{ \sum_{\tau } \big(e^{iu_\tau} - 1\big) \Big[\mu \mathrm{GW}_d^{(\mu)}(\tau) + \sum_{\beta } x_\beta \sqrt{\mu} g^{(\mu)}_{d, \beta}(\tau) \Big]  } du.
	\end{align*}
	Note that in the last step, we have swapped the integral with $[x^\gamma]$. This requires the use of Fubini's Theorem to interchange the exponential series with the integral, which is possible due to
	\begin{align*}
		&\sum_{k=0}^{\infty}\frac{1}{k!} \Big(  \int_{[0, 2\pi]^p} \Big| -i u\cdot t + \mu \sum_{\tau } \Big(e^{iu_\tau} - 1\Big) \Big[\mu \mathrm{GW}_d^{(\mu)}(\tau) + \sum_{\beta } x_\beta \sqrt{\mu} g^{(\mu)}_{d, \beta}(\tau) \Big]   \Big| du \Big)^k  \\
		& \leq\sum_{k=0}^{\infty} \frac{1}{k!} \Bigg(  \prod_{\tau \in \mathcal{X}_d}\int_{0}^{2\pi} \vert u_\tau N_\tau \vert \,  du_\tau + 2 \mu \,  \Big\vert  \mu \mathrm{GW}_d^{(\mu)}(\tau) + \sum_{\beta } x_\beta \sqrt{\mu} g^{(\mu)}_{d, \beta}(\tau) \Big\vert  \Bigg)^k \\
		&= \sum_{k=0}^{\infty} \frac{1}{k!} (\text{constant independent of } k)^k < \infty.
	\end{align*}
	Picking up our computations from above once more and introducing the random variable $Z \sim \mathrm{Poi}\Big(\mu \mathrm{GW}_d^{(\mu)}(\tau) + \sum_{\beta } x_\beta \sqrt{\mu} g^{(\mu)}_{d, \beta}(\tau) \Big)$, we obtain
	\begin{align*}
		&g^{(\mu)}_{d+1, \gamma}(u) \\
		&= \sqrt{\prod_{\beta } \gamma_\beta !} \;  [x^\gamma] \prod_{\tau}  \frac{1}{(2\pi)^{p}}   \int_0^{2\pi} e^{- iu_\tau N_\tau} \Expe \Big[e^{i u_\tau Z} \Big] du_\tau\\
		&\overset{\text{(\ref{fourier_inversion})}}{=\joinrel=} \sqrt{\prod_{\beta } \gamma_\beta !} \;  [x^\gamma] \prod_{\tau} \Prob_{ Z \sim \mathrm{Poi}\Big(\mu \mathrm{GW}_d^{(\mu)}(\tau) + \sum_{\beta } x_\beta \sqrt{\mu} f^{(\mu)}_{d, \beta}(\tau)\, \mathrm{GW}^{(\mu)}_d(\tau) \Big) }\Big(Z = N_\tau \Big)\\
		&= \sqrt{\prod_{\beta } \gamma_\beta !} \; [x^\gamma] \prod_{\tau} e^{- \mu \mathrm{GW}_d^{(\mu)}(\tau) -  \sum_{\beta } x_\beta \sqrt{\mu} g^{(\mu)}_{d, \beta}(\tau)}\\
        &\qquad \qquad \times\frac{\Big(\mu \mathrm{GW}_d^{(\mu)}(\tau) + \sum_{\beta } x_\beta \sqrt{\mu} f^{(\mu)}_{d, \beta}(\tau)\, \mathrm{GW}^{(\mu)}_d(\tau) \Big)^{N_\tau}}{N_\tau!}\\
		& = e^{-\mu} \sqrt{\prod_{\beta } \gamma_\beta !} [x^\gamma] e^{-  \sum_{\beta, \tau} x_\beta \sqrt{\mu} g^{(\mu)}_{d, \beta}(\tau)} \prod_{\tau}  \frac{\Big( \mu \mathrm{GW}_d^{(\mu)}(\tau)\Big)^{N_\tau}}{N_\tau!} \\
        & \qquad \qquad \qquad \qquad \qquad \qquad \qquad \qquad \qquad \times\Big(1 + \sum_{\beta } x_\beta \frac{1}{\sqrt{\mu}} f^{(\mu)}_{d, \beta}(\tau)\Big)^{N_\tau}\\
		& = \underbrace{e^{-\sum_\tau \mu \mathrm{GW}_d^{(\mu)}(\tau)} \prod_\tau \frac{\Big(\mu \mathrm{GW}_d^{(\mu)}(\tau)\Big)^{N_\tau} }{N_\tau!}}_{= \mathrm{GW}_{d+1}^{(\mu)}(\tau)} \sqrt{\prod_{\beta } \gamma_\beta !} [x^\gamma] e^{-  \sqrt{\mu} \sum_{\beta, \tau} x_\beta g^{(\mu)}_{d, \beta}(\tau)} \\
        & \qquad \qquad \qquad \qquad \qquad \qquad \qquad \qquad \qquad  \times  \prod_{\tau}    \Big(1 + \sum_{\beta } \frac{x_\beta}{ \sqrt{\mu}} f^{(\mu)}_{d, \beta}(\tau)\Big)^{N_\tau}.
	\end{align*}
	Finally, we define our candidate for the depth-$(d+1)$ family of functions $f^{(\mu)}_{d+1, \gamma}$ as
	\begin{equation}\label{def_f}
		f^{(\mu)}_{d+1, \gamma}(\tau) := \sqrt{\prod_{\beta } \gamma_\beta !} [x^\gamma] e^{-  \sqrt{\mu} \sum_{\beta, \tau} x_\beta g^{(\mu)}_{d, \beta}(\tau)} \prod_{\tau}    \Big(1 + \sum_{\beta } \frac{x_\beta}{ \sqrt{\mu}} f^{(\mu)}_{d, \beta}(\tau)\Big)^{N_\tau}.
	\end{equation}
	Inserting $g^{(\mu)}_{d+1, \gamma}(u) =  \mathrm{GW}_{d+1}^{(\mu)}(\tau) f^{(\mu)}_{d+1, \gamma}(\tau)$ into (\ref{almost_done}) yields the objective defined in (\ref{goal_induction}). What remains to show are the properties (\ref{appendix:trivialtree}), (\ref{appendix:orthogonaliy_1_GW}) and $(\ref{appendix:orthogonal_2_betasum})$.
	
	Since $f^{(\mu)}_{d+1, \gamma}$ only depends on the parameter $\mu$ instead of the asymmetry parameters $s, s'$, its definition in (\ref{def_f}) coincides with the one presented in the original paper \cite{ganassali2022statistical}. Hence, we refer the reader to steps 2.1 and 2.2 in the proof of this refererence's Theorem 4 to conclude our proof of Theorem \ref{thm_diagonalization}.
\end{proof}

\subsection{Proofs of three implications from Theorem \ref{thm_equivalences}}\label{appendix:thm_equivalences}
We recall the theorem about one-sided testability equivalences before proving the three implications from that theorem, which were deferred from Section \ref{section:tree_correl_testing}.
\begin{theorem}
	In the tree correlation testing problem (\ref{testing_problem_2}), the following are equivalent:
	\begin{enumerate}[label=(\alph*)]
		\item There exist one-sided tests $\mathcal{T}_d$ to decide $\Pindd$ vs. $ \Pcorrd $,
		\item There is a sequence of thresholds $\theta_d \to \infty $ such that $\Pindd(L_d > \theta_d) \to 0$ and $\liminf_d \Pcorrd(L_d > \theta_d) > 0$,
		\item The $ \Pind $--martingale $(L_d)_d$ is not uniformly integrable,
		\item $\lambda s s' > 1$ and $KL_\infty = \infty$,
		\item $\lambda s s' > 1$ and $\,  \Pcorr \Big(\liminf_d (\lambda s s')^{-d} \log(L_d) \geq C\Big) \geq 1 - \Prob(\mathrm{Ext}(\mathrm{GW}^{(\lambda s s')})) $ for a constant $C>0$ only depending on $(\lambda, s, s')$ and $\mathrm{Ext}(\mathrm{GW}^{(\lambda s s')})$ denoting the extinction event of the Galton--Watson process with offspring $\mathrm{Poi}(\lambda ss')$.
	\end{enumerate}
\end{theorem}

We start with (c) $\Rightarrow$ (b) before showing (\ref{implication_to_be_shown}) which is part of (b) $\Rightarrow$ (a).   Finally, we move to the most involved implication, (d) $\Rightarrow$ (e).
\begin{proof}[Proof of (c) $\Rightarrow$ (b) in Theorem \ref{thm_equivalences}]
	Recall that we assume $\Eind[L_\infty] < 1$ and want to construct a sequence $\theta_d \to \infty$ fulfilling  
	\begin{equation*}
		\liminf_{d\to\infty} \Pcorr(L_d > \theta_d)  > 0 \quad \text{and} \quad \Pind(L_d > \theta_d) \xrightarrow[d\to\infty]{} 0.
	\end{equation*}
	Let $\theta$ be any continuity point of $ L_\infty \mathrel{\#} \Pind$, the pushforward measure of $\Pind$ under $L_\infty$. We write
	\begin{align*}
		\Eind[L_d] = \Eind[L_d \indicator{L_d > \theta}] + \Eind[L_d \indicator{L_d \leq \theta}] = \Pcorr(L_d > \theta) + \Eind[L_d \indicator{L_d \leq \theta}].
	\end{align*}
	Since $\Eind[L_d] = 1$, this implies
	\begin{equation}\label{eq_liminfEps}
		\liminf_{d\to \infty}  \Pcorr(L_d > \theta) = 1 - \limsup_{d\to \infty} \Eind[L_d \indicator{L_d \leq \theta}] \geq 1 - \Eind[L_\infty]
	\end{equation}
	where in the last step we have used dominated convergence on $L_d \indicator{L_d \leq \theta} \leq \theta \in L^1$ as well as $L_d \indicator{L_d \leq \theta} \to L_\infty \indicator{L_\infty \leq \theta} \leq L_\infty$ since $\theta$ is a continuity point. Setting $\varepsilon := 1 - \Eind[L_\infty] >0$, we can now construct a sequence $\theta_d$ with the desired properties. 
	
	For continuity points $\theta \geq 0$ of $L_\infty \mathrel{\#} \Pind$ define 
	\begin{equation*}
		d(\theta) := \inf\Big\{ k \in \N \suchthat \forall d \geq k: \Pcorr(L_d > \theta) > \varepsilon/2\Big\}.
	\end{equation*}
	Thanks to equation \ref{eq_liminfEps}, this is always a finite quantity which increases with $\theta$ since $\Pcorr(L_d > \bullet)$ is decreasing in $d$. We set 
	\begin{equation*}
		\hat{\theta}_d := \sup\Big\{\theta \in \R_{\geq 0} \suchthat \theta \text{ is a continuity point of }  L_\infty \mathrel{\#} \Pind \text{ and } d(\theta) \leq d \Big\}.
	\end{equation*}
	This yields a sequence of finite real numbers since $\{\theta: d(\theta) \leq d\}$ is bounded: Otherwise there would be $\theta_\ell \to \infty$ with $d(\theta_\ell) \leq d$ meaning that $\Pcorr(L_d > \theta_\ell) > \varepsilon/2$. This is a contradiction to $L_d \in L^1(\Pind)$ because 
	\begin{equation}\label{eq_conv0ForAToInf}
		\Pcorr(L_d > \theta_\ell) = \Eind[L_d \indicator{L_d > \theta_\ell}] \xrightarrow{\ell \to \infty} 0.
	\end{equation}
	The same equation (\ref{eq_conv0ForAToInf}) also implies $\hat{\theta}_d \xrightarrow{d \rightarrow \infty} \infty$ because arbitrarily high values of $d(\theta)$ allow for arbitrarily high values of $\theta$. Finally, choose the sequence $(\theta_d)_d$ in a way that all $\theta_d$ are continuity points of $L_\infty \mathrel{\#} \Pind$ satisfying $\theta_d \in (\hat{\theta}_d-1, \hat{\theta}_d)$. By this construction, we have
	\begin{equation*}
		\liminf_{d\to\infty} \Pcorr(L_d > \theta_d) \geq \frac{\varepsilon}{2} > 0 \quad \text{and} \quad \Pind(L_d > \theta_d) \leq \frac{\Eind[L_d]}{\theta_d} = \frac{1}{\theta_d} \xrightarrow{d\to\infty} 0.
	\end{equation*}
	This is exactly the desired statement.
\end{proof}

\begin{proof}[Proof of equation (\ref{implication_to_be_shown}) from (b) $\Rightarrow$ (a) in Theorem \ref{thm_diagonalization}]
	Recall that our objective is to show
	\begin{equation*}
		(\hat\theta_d)_d \text{ is bounded}\implies \Eind[L_\infty] < 1
	\end{equation*}
	under the assumption $\Pind(\mathcal{T}_d = 1) =: \alpha_d \to 0$.
	
	If $(\hat\theta_d)_d$ is bounded, there exists $M>0$ and a converging subsequence $(\hat{\theta}_{d_k})_k$ such that $\hat{\theta}_{d_k} <M$ for all $k$. Since $L_k$ converges to $L_\infty$ in distribution, the Portmanteau theorem applied to the open set $(M, \infty)$ yields
	\[
	\Pind(L_\infty > M) \leq \liminf_{k\to\infty} \Pind(L_k > M)  \leq \Pind(L_k > \hat{\theta}_{d_k}) \xrightarrow{k\to\infty} 0
	\]
	where the second inequality follows from $ \alpha_d \to 0 $. Consequently, $\Pind(L_\infty \leq M) = 1$.\\
	Define the convex function $g(x) = (M - x)_+$ and note that $g(L_d)$ is a submartingale with values in $[0,M]$. For this reason, $d \mapsto \Eind[g(L_d)]$ is non-decreasing and we can apply the dominated convergence theorem to compute its limit:
	\begin{equation}\label{first}
		\lim_{d\to\infty} \Eind[g(L_d)] = \Eind[g(L_\infty)] = \Eind[(M- L_\infty)_+] = M - \Eind[L_\infty].
	\end{equation}
	Despite the limit random variable being bounded by $M$, we claim that there exists an index $d'$ for which $\Pind(L_{d'} > M) > 0$. This is true since assuming the opposite, i.e., $\Pind(L_d \leq M) = 1$ for all $d$, implies a contradiction to $ \liminf_d \Pcorr(\mathcal{T}_d = 1) > 0$:
	\begin{align*}
		\Pcorr(\mathcal{T}_d = 1) = \Eind[L_d \indicator{\mathcal{T}_d = 1}] \leq M \,  \Pind(\mathcal{T}_d = 1) \leq M\, \alpha_d \xrightarrow{d\to \infty} 0.
	\end{align*}
	At index $d'$, we therefore have
	\begin{equation}\label{second}
		\Eind[g(L_{d'})] = \Eind[(M - L_{d'})_+] > \Eind[M - L_{d'}] = M-1.
	\end{equation}
	Combining (\ref{first}) and (\ref{second}) then yields
	\begin{equation*}
		M-1 < \Eind[g(L_{d'})] \leq M - \Eind[L_\infty] \quad \implies \quad \Eind[L_\infty] < 1
	\end{equation*} 
	which we wanted to show.
\end{proof}

\begin{proof}[Proof of (d) $\Rightarrow$ (e) in Theorem \ref{thm_equivalences}.]
	This implication is the most intricate and necessitates introducing additional notation. For a rooted tree $\tau$, we recall that its number of nodes up to depth $d$ is denoted $|\tau_d|$.  We further write $\mathcal{G}_d(\tau) := \tau_d \setminus \tau_{d-1}$ for the set of nodes in $\tau$ which have exactly depth $d$. Hence, $|\mathcal{G}_d(\tau_*)|$ is the number of $\tau_*$'s nodes at generation $d$. Furthermore, we will use Proposition \ref{prop_explicit_likhood} and the notation introduced for that statement.
	
	Let $(t, t') \sim \Pcorr$ and let $t_*$ be their intersection tree from the sampling process. Call $\sigma_*$ and $\sigma'_*$ the injections of $t_*$ into $t$ and $t'$ respectively.
	The intersection tree's marginal distribution is $t_* \sim \mathrm{GW}_{\lambda s s'}$, and we let $\mu := \lambda s s'$ denote the mean of its offspring distribution. The following lemma describes the asymptotic behaviour of $|\mathcal{G}_d(t_*)|$ compared to $\mu^d$:
	\begin{lemma}\label{lemma_GW_generationmartingale}
		For $t_* \sim \mathrm{GW}_{\lambda s s'}$ set $w_d := |\mathcal{G}_d(t_*)|/\mu^d$. Then, $(w_d)_d$ is a positive martingale whose almost sure limit $w$ satisfies
		\[ 
		\Prob(w > 0) = \Prob(\mathrm{Ext}^c(t_*)).
		\]
	\end{lemma}
	For a proof of this standard result on branching processes, we refer the reader to Lemma 2.13 in \cite{abraham2015introduction}.
	
	As a consequence of Lemma \ref{lemma_GW_generationmartingale}, the martingale $ \mu^{-d}\, |\mathcal{G}_d(t_*)| $ converges almost surely towards a random variable $w$.
	
	From now on, we condition on the event $\mathrm{Ext}^c(t_*) = \{t_* \textit{ survives}\}$. This is possible since $\mu > 1$ implies $\Prob(\mathrm{Ext}^c(t_*)) > 0$. By Lemma \ref{lemma_GW_generationmartingale}, we can therefore assume $w > 0$.
	
	Using Proposition \ref{prop_explicit_likhood}, we have that for all $d, k \in \N$,
	\begin{align*}
		L_{d+k}(t, t') &\geq \underbrace{\prod_{i \in (t_*)_{d-1}} \psi\Big(c_{t_*}(i), \, c_t(\sigma_*(i)), \, c_{t'}(\sigma'_*(i))\Big)}_{=: A_d} \, \underbrace{ \prod_{j \in \mathcal{G}_d(t_*)} L_k(t_{[\sigma_*(j)]}, t'_{[\sigma'_*(j)]})}_{=:B_{d,k}}.
	\end{align*}
	We individually examine $A_d$ and  $B_{d,k}$, starting with the latter. Taking the logarithm, we get
	\begin{align*}
		\log(B_{d,k}) = |\mathcal{G}_d(t_*)| \, \underbrace{\frac{1}{|\mathcal{G}_d(t_*)|} \sum_{i \in \mathcal{G}_d(t_*)} \log\Big(L_k(t_{[\sigma_*(i)]}, t'_{[\sigma'_*(i)]})\Big)}_{\quad =: a_d}.
	\end{align*}
	For $i \neq j$, the random variables $(t_{[\sigma_*(i)]}, t'_{[\sigma'_*(i)]})$ and $(t_{[\sigma_*(j)]}, t'_{[\sigma'_*(j)]})$ are independent since they represent disjoint branches of a GW-tree. They both follow the same law $\Pcorr$, and the logarithms of their likelihoods are in $L^1(\Pcorr)$. Therefore, the law of large numbers is applicable, yielding $a_d \to \Ecorr[\log L_k]$ almost surely. In other words, on an event of probability 1, there exists a sequence $0 < \varepsilon_d \to 0$ such that for all $d$,
	\begin{equation}\label{eq_epsm_majoration}
		\varepsilon_d \geq |\Ecorr[\log(L_k)] - a_d |, \quad \text{that is} \quad a_d \geq \Ecorr[\log(L_k)]| - \varepsilon_d.
	\end{equation}
	Similarly, since $|\mathcal{G}_d(t_*)| \mu^{-d} \to w$ almost surely, there is a sequence $0 < \delta_d \to 0$ with
	\begin{equation}\label{eq_deltm_majoration}
		\delta_d \geq \Big \vert |\mathcal{G}_d(t_*)| \mu^{-d} - w \Big \vert, \quad \text{implying} \quad |\mathcal{G}_d(t_*)| \geq w \mu^d - \delta_d \mu^d \quad \text{for all }d.
	\end{equation} 
	Noting that $ \Ecorr[\log(L_k)] = \mathrm{KL}_k $, we can combine (\ref{eq_epsm_majoration}) with (\ref{eq_deltm_majoration}) to obtain
	\begin{align*}
		B_{d,k} = \exp\left\{|\mathcal{G}_d(t_*)| a_d \right\} &\geq \exp\left\{|\mathcal{G}_d(t_*)| \left(\mathrm{KL}_k - \varepsilon_d\right)\right\} \\ &\geq \exp\left\{\left( w \, \mu^d - \delta_d \mu^d \right) \left(\mathrm{KL}_k - \varepsilon_d\right)\right\}\\
		&\geq \exp\left\{\mathrm{KL}_k \, w \, \mu^d - \mu^d \left( \delta_d \mathrm{KL}_k + \varepsilon_d w \right) \right\}.
	\end{align*}
	Let $n$ be large enough so that $\mathrm{KL}_k > 1$ and fix an element $\omega$ of the probability space $\Omega$ for which all sequences converge. Choosing $d$ sufficiently large, one has $\delta_d(\omega) \leq w(\omega)/4$ and $\varepsilon_d(\omega) \leq \mathrm{KL}_k/4$, yielding
	\[ 
	\mu^d\delta_d(\omega)  \mathrm{KL}_k + \mu^d \varepsilon_d(\omega) w(\omega)  \leq \text{\sfrac{1}{2} } \mathrm{KL}_k w(\omega) \mu^d.
	\]
	This implies that on an event $\mathcal{B}$ of probability 1,
	\begin{equation}\label{B_dk_lower_bound}
		B_{d,k} \geq \exp\left\{\text{\sfrac{1}{2} } \mathrm{KL}_k w \mu^d\right\}\quad \text{ for } d \text{ large enough. }
	\end{equation}
	
Turning to $A_d$ now, we start by introducing the notation
\[
F_i := \log\left( \psi\left(c_{t_*}(i), \, c_t(\sigma_*(i)), \, c_{t'}(\sigma'_*(i))\right)  \right).
\]
As before, we take the logarithm to rewrite
\begin{align*}
	\log(A_d) &= \sum_{i \in (t_*)_{d-1}} F_i = \sum_{g = 0}^{d-1} \, \sum_{i \in \mathcal{G}_g(t_*)} F_i.
\end{align*}
One has the following lemma:
\begin{lemma}\label{lem_Fi_distr}
	The random variables $(F_i)_{i \in (t_*)_d}$ are identically distributed with finite mean and finite variance.
\end{lemma}
\begin{proof}[Proof of Lemma \ref{lem_Fi_distr}]
	We start by noting that the law of $F_i$ does not depend on $i$ since
	\begin{align*}
		F_i =&\log \Big( \psi\big(c_{t_*}(i), \, c_t(\sigma_*(i)), \, c_{t'}(\sigma'_*(i))\big)  \Big)\\
		=&\, \log\left(e^{\lambda s s'} \, \lambda^{-c_{t_*}(i)} \frac{1}{c_{t_*}(i)!} \left(1 - s'\right)^{c_t(\sigma_*(i)) - c_{t_*}(i)}\, \left(1 - s\right)^{c_{t'}(\sigma'_*(i)) - c_{t_*}(i)}\right)\\
		=&\, \lambda s s' -c_{t_*}(i) \log(\lambda) - \log\left(c_{t_*}(i)!\right) + \left(c_t(\sigma_*(i)) - c_{t_*}(i)\right)\log(1 - s') \\
		& \quad \quad + \left(c_{t'}(\sigma'_*(i))-c_{t_*}(i)\right)\log(1 - s).\\
		\overset{\mathrm{law}}{=}& \,  \lambda s s' -\left[X \log(\lambda) + \log\left(X!\right) \right]+ Y \log(1 - s') + Y'\log(1 - s) =: F,
	\end{align*}
	where $X \sim \mathrm{Poi}(\lambda s s')$, $Y \sim \mathrm{Poi}(\lambda s (1-s'))$ and $Y' \sim \mathrm{Poi}(\lambda s' (1-s))$ are independent Poisson variables. This follows from the definition of $(\lambda, s, s')$-augmentations.
	
	To show that $F$ has finite second moment, define the independent random variables
	\[
	W:= \lambda s s' + Y \log(1 - s') + Y'\log(1 - s),
	\qquad
	Z:= X \log(\lambda) + \log(X!),
	\]
	so that $F=W - Z$. Since $Y$ and $Y'$ are Poisson variables, one has $\Expe[W^2]<\infty$ and it remains to check that $\Expe[Z^2]<\infty$. For $x \in \N_{\geq 1}$, one has
	\(
		\log(x!) \leq \log(x^x) = x \log x \leq x^2
	\)
	which leads to 
	\begin{align*}
		\Expe[Z^2] &= \log(\lambda)^2 \Expe[X^2] + 2 \log(\lambda) \Expe[X \log(X!)] + \Expe[ \log(X!)^2] \\
		&\in \mathcal{O}\Big(\Expe[X^2] + \Expe[X^3] + \Expe[X^4]\Big).
	\end{align*}
		Since the Poisson distribution has moments of all orders, this expression is finite. Thus $\Expe[Z^2]<\infty$, leading to $\Expe[F^2]<\infty$. In particular, $F$ has finite variance and mean.
\end{proof}
According to Lemma \ref{lem_Fi_distr}, all $F_i$ are identically distributed and we write $f := \Expe[F]$ for their common mean. Let $\mathcal{G}_g$ be a shorthand for $\mathcal{G}_g(t_*)$ and set
\[
\hat{S}_{d} := \sum_{g = 0}^{d-1} \, \sum_{i \in \mathcal{G}_g} \bigl(F_i - f\bigr)
= \log(A_d) - f\sum_{g = 0}^{d-1} |\mathcal{G}_g|.
\]
We define the sigma-fields $\mathcal{F}_{d} := \sigma\big(F_i \, : \, i \in \mathcal{G}_g,\, g\leq d-1\big)$ and make the following claim:
\begin{equation}\label{eq_Shat_recursion}
	\Expe\big[\hat{S}_{d}^2\big] = \Expe\big[\hat{S}_{d-1}^2\big] + \Var(F)\, \mu^{d-1}.
\end{equation}
To show this equation, use the tower rule of conditional expectation to compute
\begin{align*}
	&\Expe\left[\hat{S}_{d}^2\right]
	= \Expe\big[\Expe[\hat{S}_d^2 \suchthat \mathcal{F}_{d-1}]\big]
	= \Expe\bigg[\Expe\Big[\Big(\hat{S}_{d-1} + \sum_{i \in \mathcal{G}_{d-1}} (F_i - f)\Big)^2 \, \Big| \, \mathcal{F}_{d-1}\Big]\bigg]\\
	&= \Expe\big[\hat{S}_{d-1}^2\big]
	+ \Expe\bigg[\Expe\Big[\Big(\sum_{i \in \mathcal{G}_{d-1}} (F_i - f)\Big)^2 \Big| \mathcal{F}_{d-1}\Big]\bigg] + 2\Expe\bigg[\hat{S}_{d-1}\Expe\Big[\sum_{i \in \mathcal{G}_{d-1}} (F_i - f)\,\Big|\, \mathcal{F}_{d-1}\Big]\bigg].
\end{align*}
Conditionally on $\mathcal{F}_{d-1}$, the random variables $(F_i)_{i\in \mathcal{G}_{d-1}}$ are independent of common law $F$, because the descendants of distinct nodes in generation $d-1$ are generated independently in the Galton--Watson process. Therefore,
\[
\Expe[F_i-f \suchthat \mathcal{F}_{d-1}] = 0
\quad\text{and}\quad
\Expe[(F_i-f)^2 \suchthat \mathcal{F}_{d-1}] = \Var(F) \quad  \text{for all } i\in \mathcal{G}_{d-1}.
\]
The zero-mean and independence of $(F_i - f), (F_j - f)$ for $i \neq j$ implies that mixed terms vanish when we develop the following square:
\[
\Expe\bigg[\Big(\sum_{i \in \mathcal{G}_{d-1}} (F_i - f)\Big)^2 \Big| \mathcal{F}_{d-1}\bigg]
=
\sum_{i \in \mathcal{G}_{d-1}} \Var(F)
=
\Var(F)\, |\mathcal{G}_{d-1}|.
\]
This proves (\ref{eq_Shat_recursion}) which we apply recursively to obtain
\begin{equation}\label{formula_Shatm2}
	\Expe\left[\hat{S}_{d}^2\right]
	= \sum_{g = 1}^d \Var(F)\, \mu^{g-1}
	= \Var(F)\, \frac{\mu^d - 1}{\mu -1}.
\end{equation}
This allows us to establish
\begin{align*}
	\Prob\left(|\hat{S}_d| \geq d \mu^{\frac{d}{2}}\right)
	\leq \frac{\Expe\left[\hat{S}_{d}^2\right]}{d^2\mu^d}
	\leq \frac{1}{d^2}\,\Var(F)\, \dfrac{1 - \mu^{-d}}{\mu - 1}
	\leq \frac{C}{d^2},
\end{align*}
where we have used Markov's inequality, equation (\ref{formula_Shatm2}), Lemma \ref{lem_Fi_distr}, and the fact that $1-\mu^{-d}\leq 1$. The constant $C$ may depend on $\mu$ and $\Var(F)$ but is independent of $d$. Consequently, the probabilities of the events $\{|\hat{S}_d| \geq d \mu^{\frac{d}{2}}\}$ are summable, and by Borel--Cantelli there almost surely exists an $D \in \mathbb{N}$ such that
\begin{equation}\label{eq_A_fluctuation_control}
	\Big| \, \log(A_d) - f\sum_{g = 0}^{d-1} |\mathcal{G}_g| \,\Big| < d \mu^{\frac{d}{2}}
	\qquad \text{for all } d\geq D.
\end{equation}
Recall that Lemma \ref{lemma_GW_generationmartingale} states that almost surely,
\[
\mu^{-g} |\mathcal{G}_g| \xrightarrow[g \to \infty]{} w,
\quad\text{with}\quad
w(\omega)>0 \quad \text{for almost all}\quad \omega \in \mathrm{Ext}^c(t_*).
\]
Fix $\delta \in (0,1)$. Then for almost every $\omega \in \mathrm{Ext}^c(t_*)$ there exists $D'(\omega)\in \mathbb{N}$ such that
\begin{equation}\label{eq_generation_two_sided}
	(1-\delta)w(\omega)\mu^g
	\leq |\mathcal{G}_g|(\omega)
	\leq (1+\delta)w(\omega)\mu^g
	\qquad \text{for all } g\geq D'(\omega).
\end{equation}
Next, derive a lower bound on the drift term $
f\sum_{g=0}^{d-1} |\mathcal{G}_g|(\omega)$:\\
If $f\geq 0$, then trivially
\[
f\sum_{g=0}^{d-1} |\mathcal{G}_g|(\omega)\geq 0.
\]
If $f<0$, then using the upper bound in \eqref{eq_generation_two_sided} for all $d\geq D'(\omega)$,
\begin{align*}
	f\sum_{g=0}^{d-1} |\mathcal{G}_g|(\omega)
	& \geq f\sum_{g=0}^{D'(\omega)-1} |\mathcal{G}_g|(\omega)
	+ f(1+\delta)w(\omega)\sum_{g=D'(\omega)}^{d-1}\mu^g\\
	&\geq - C_1(\omega) -\frac{|f|(1+\delta)}{\mu-1}\,w(\omega)\mu^d
\end{align*}
for a finite constant $C_1(\omega) \geq 0$. Hence,  for all sufficiently large $d$,
\begin{equation*}
	f\sum_{g=0}^{d-1} |\mathcal{G}_g|(\omega)
	\geq -\frac{|f|(1+\delta)}{\mu-1}\,w(\omega)\mu^d - C_1(\omega).
\end{equation*}
Combining this with (\ref{eq_A_fluctuation_control}), we obtain that for almost all $ \omega \in \mathrm{Ext}^c(t_*)$ there exists a finite $D''(\omega)$ such that for all $d\geq D''(\omega)$,
\begin{equation}\label{low}
	\log(A_d)(\omega) \geq -\frac{|f|(1+\delta)}{\mu-1}\,w(\omega)\mu^d - C_1(\omega) - d\mu^{d/2}.
\end{equation}
Since $w(\omega)>0$, by increasing $D''(\omega)$ if needed, we may ensure that
\[
C_1(\omega) + d\mu^{d/2} \leq w(\omega)\mu^d
\qquad\text{for } d\geq D''(\omega).
\]
which, if plugged into (\ref{low}), leads to
\begin{equation}\label{eq_Ad_final_lower_bound}
	\log(A_d)(\omega)
	\geq -\left(\frac{|f|(1+\delta)}{\mu-1}+1\right) w(\omega)\mu^d \qquad\text{for } d\geq D''(\omega).
\end{equation}
Let $\mathcal{A}$ denote the event where (\ref{eq_Ad_final_lower_bound}) holds, which has probability $\Prob(\mathcal{A}) \geq \Prob(\mathrm{Ext}^c(t_*))$. For every $\omega\in\mathcal{A}$ and sufficiently large $d$,
\[
\log(A_d)(\omega)
\geq -C_A\, w(\omega)\mu^d,
\qquad\text{where}\qquad
C_A:=\frac{|f|(1+\delta)}{\mu-1}+1.
\]
Since $C_A$ is deterministic, we can choose $k$ large enough so that$
\frac12 \mathrm{KL}_k > C_A$. 
Combining \eqref{eq_Ad_final_lower_bound} with \eqref{B_dk_lower_bound}, we obtain that on the event $\mathcal A\cap \mathcal B$,
\begin{align*}
	\log\big(L_{d+k}(t,t')\big) \geq \log(A_d)+\log(B_{d,k})\geq \left(-C_A+\frac12 \mathrm{KL}_k\right) w\mu^d
\end{align*}
for all sufficiently large $d$. 
Since $\Prob(\mathcal{B}) = 1$, we conclude 
\begin{align*}
\Pcorr \left(\liminf_{d\to \infty} \mu^{-(d+k)} \log(L_{d+k}) \geq \mu^{-k}(- C_A + \text{\sfrac{1}{2} }\mathrm{KL}_k)\right) &\geq \Prob\left(\mathcal{A}\cap\mathcal{B}\right)
\\&\geq \Prob( \mathrm{Ext}^c(t_*)).
\end{align*}
The result (e) from Theorem \ref{appendix:thm_equivalences} then follows by setting
\[ 
C:= \mu^{-k}(C' + \text{\sfrac{1}{2} }\mathrm{KL}_k) > 0\quad \text{and observing} \quad \Prob\left(\mathrm{Ext}^c(t_*) \right)  = 1 - \Prob(\mathrm{Ext}(\mathrm{GW}_{\lambda s s'})).
\] 
As a final remark, the constant $C$ can be computed from $\mu $, $\Expe[F]$, $k$ and $\mathrm{KL}_k $ which all solely depend on $\lambda, s$ and $ s' $.
\end{proof}

\subsection{Proof of Lemma \ref{gaussian_approx}}\label{appendix:gaussian_approx}
We start by recalling Lemma \ref{gaussian_approx} on Gaussian approximation, which is instrumental for proving (ii) in Theorem \ref{thm:phase_transition}:
\begin{lemma}[copy of Lemma \ref{gaussian_approx}]
	Let $d \in \N$ and recall that $\mathcal{X}_{d+1} = \N^{\mathcal{X}_d}$ by the subtree tuple identification. There exists a pair of functions
	\[
	(y, y'): \mathcal{X}_{d+1}\times \mathcal{X}_{d+1} \to \R^{\mathcal{X}_d} \times  \R^{\mathcal{X}_d}, \, (t, t')  \mapsto \Big(y_\beta (t), y'_\beta(t') \Big)_{\beta \in \mathcal{X}_d}
	\]
	with the property that the above map is affine and bijective. The law of the random vector $(y(t), y'(t')))$ shall be denoted as $\mathcal{L}(y, y')$ if $(t, t') \sim \Pcorrdplus$ and $\mathcal{L}(y) \otimes \mathcal{L}(y')$ in the case where $(t, t') \sim \Pinddplus$.
	
	One has
	\[
	\mathrm{KL}(\Pcorrdplus \, \Vert \, \Pinddplus ) = \mathrm{KL}\Big(\mathcal{L}(y, y') \, \Vert \, \mathcal{L}(y) \otimes \mathcal{L}(y')\Big)
	\]
	and that $\mathcal{L}(y, y')$ converges weakly towards the joint law of a pair of random variables
	\[
	\mathcal{L}(y, y') \xrightarrow[\lambda \to \infty]{w} \mathcal{L}(z, z').
	\]
	The pair $(z, z') = \Big( (z_\beta)_{\beta \in \mathcal{X}_d}, (z'_\beta)_{\beta \in \mathcal{X}_d}\Big)$ is a gaussian vector of infinite dimension with zero mean and covariances defined by
	\[
	\Expe[z_\beta z_\gamma] = \Expe[z'_\beta z'_\gamma] = \indicator{\beta = \gamma}, \quad \Expe[z_\beta z'_\gamma] = \sqrt{ss'}^{|\beta|}\indicator{\beta = \gamma} \quad \text{for all } \beta, \gamma \in \mathcal{X}_d.
	\]
\end{lemma}
\begin{proof}
	The idea for proving the Gaussian approximation lemma is to transform the (integer-valued) subtree tuple representation of a tree $t$ into a closely related tuple of real random variables $y(t)$ with zero mean and normalized variance. Showing weak convergence will then use Levy's Continuity Theorem, applied to infinite-dimensional random vectors.
	
	We begin by letting $t = (N_\tau)_{\tau \in \mathcal{X}_d}$, $t' = (N'_\tau)_{\tau \in \mathcal{X}_d} $ be two trees. For all $\beta \in \mathcal{X}_d$, set
	\begin{align*}
		y_\beta(t)&:= \frac{1}{\sqrt{\lambda s}} \sum_{\tau \in \mathcal{X}_d} f^{(\lambda s)}_{d, \beta}(\tau) \Big(N_\tau - \lambda s \, \mathrm{GW}^{(\lambda s)}_d (\tau) \Big) \quad \text{and}\\
		y'_\beta(t') &:= \frac{1}{\sqrt{\lambda s'}} \sum_{\tau \in \mathcal{X}_d} f^{(\lambda s')}_{d, \beta}(\tau) \Big(N'_\tau - \lambda s' \, \mathrm{GW}^{(\lambda s')}_d (\tau) \Big).
	\end{align*}
	The function $y := (y_\beta)_{\beta \in \mathcal{X}_d} : \mathcal{X}_{d+1} \to \R^{\mathcal{X}_d}$ is clearly affine from the definition. It is furthermore bijective with its inverse map given by $t(y) = (N_\tau(y))_{\tau \in \mathcal{X}_d}$ where
	\[
	N_\tau(y) = \lambda s \, \mathrm{GW}_d^{(\lambda s)}(\tau) + \sqrt{\lambda s} \, \mathrm{GW}_d^{(\lambda s)}(\tau) \sum_{\beta \in \mathcal{X}_d} y_\beta f_{d, \beta}^{(\lambda s)}(\tau).
	\]
	This is indeed the inverse because
	\begin{align*}
		\sum_{\beta \in \mathcal{X}_d} y_\beta f_{d, \beta}^{(\lambda s)}(\tau) &= \frac{1}{\sqrt{\lambda s}}\sum_{\beta \in \mathcal{X}_d} \sum_{\hat \tau \in \mathcal{X}_d} f^{(\lambda s)}_{d, \beta}(\hat \tau) \Big(N_{\hat \tau} - \lambda s \, \mathrm{GW}^{(\lambda s)}_d (\hat \tau) \Big)  f_{d, \beta}^{(\lambda s)}(\tau)\\
		&= \frac{1}{\sqrt{\lambda s}} \sum_{\hat \tau \in \mathcal{X}_d} \Big(N_{\hat \tau} - \lambda s \, \mathrm{GW}^{(\lambda s)}_d (\hat \tau) \Big)  \sum_{\beta \in \mathcal{X}_d} f^{(\lambda s)}_{d, \beta}(\hat \tau) f^{(\lambda s)}_{d, \beta}(\tau)\\
		&\overset{\text{by (\ref{orthogonaliy_2_betasum})}}{=\joinrel=} \frac{1}{\sqrt{\lambda s}} \sum_{\hat \tau \in \mathcal{X}_d} \Big(N_{\hat \tau} - \lambda s \, \mathrm{GW}^{(\lambda s)}_d (\hat \tau) \Big)  \frac{\indicator{\tau = \hat \tau}}{\mathrm{GW}_d^{(\lambda s)}(\tau)}\\
		&= \frac{1}{\sqrt{\lambda s}} \Big(\frac{N_\tau}{\mathrm{GW}_d^{(\lambda s)}} - \lambda s\Big).
	\end{align*}
	The Kullback--Leibler divergence is invariant to affine, bijective transformations (see, for instance, \cite{joram2020statproofbook}) and therefore
	\[
	\mathrm{KL}(\Pcorrdplus \, \Vert \, \Pinddplus ) = \mathrm{KL}\Big(\mathcal{L}(y, y') \, \Vert \, \mathcal{L}(y) \otimes \mathcal{L}(y')\Big)
	\]
	where we recall that $\mathcal{L}(y, y')$ is the law of $(y(t), y'(t')))$  if $(t, t') \sim \Pcorrdplus$, and $\mathcal{L}(y) \otimes \mathcal{L}(y')$ is the law of $(y(t), y'(t')))$ if $(t, t') \sim \Pinddplus$. These notations do not mention $\lambda$ and $d$, but we emphasize that $y$ and $y'$ depend on these parameters.
	
	Next, we will show that a random vector $(y, y') \sim \mathcal{L}(y, y')$ converges in distribution for $\lambda \to \infty$. For this, let  $(z, z') = \Big( (z_\beta)_{\beta \in \mathcal{X}_d}, (z'_\beta)_{\beta \in \mathcal{X}_d}\Big)$ be an infinite-dimensional Gaussian vector with zero expectation and the following covariance structure:
	\[
	\Expe[z_\beta z_\gamma] = \Expe[z'_\beta z'_\gamma] = \indicator{\beta = \gamma}, \quad \Expe[z_\beta z'_\gamma] = \sqrt{ss'}^{|\beta|}\indicator{\beta = \gamma} \quad \text{for all } \beta, \gamma \in \mathcal{X}_d.
	\]
	To show $ (y, y') \xrightarrow[\lambda \to \infty]{(d)} (z, z')$, we start with some clarifications on weak convergence in metric spaces. Note that both $(y, y')$ and $(z, z')$ are elements of $\R^{\mathcal{X}_d} \times \R^{\mathcal{X}_d}$, the space of real-valued functions indexed by trees. According to \cite[Example 1.2]{billingsley2013convergence}, this space can be made into a complete separable metric space by endowing it with the distance
	\[
	\delta\Big((y, y'), (z, z')\Big) := \frac{1}{2} \sum_{\beta \in \mathcal{X}_d} 3^{-|\beta|} \min\Big( \vert y_\beta - z_\beta \vert, \, 1\Big) + \frac{1}{2} \sum_{\beta \in \mathcal{X}_d} 3^{-|\beta|} \min\Big( \vert y'_\beta - z'_\beta \vert, \, 1\Big).
	\]
	This metric is well-defined since for any $(y, y'), (z, z')$, one has
	\[
	\delta\Big((y, y'), (z, z')\Big) \leq \sum_\beta 3^{-|\beta|} = \sum_{n = 1}^\infty \Big| \mathcal{X}_{d}^{(n)}\Big| 3^{-n} \leq \Phi(1/3) \leq \Phi(\alpha)< \infty
	\]
	where finiteness follows from Proposition \ref{Otter}. Since both $(y, y')$ and $(z, z')$ are random variables with values in the metric space of real-valued sequences $(\R^{\mathcal{X}_d} \times \R^{\mathcal{X}_d}, \delta)$, we know from \cite[Example 2.4]{billingsley2013convergence} that the sequences with finitely many non-zero elements form a \textit{convergence-determining class}. This means that to establish the desired convergence in distribution, it suffices to check $(y, y') \to (z, z')$ for all finite-dimensional subsequences. 
	
	To do so, Levy's Continuity Theorem allows us to consider (logarithms of) characteristic functions: For any two sequences $v = (v_\beta)_{\beta \in \mathcal{X}_d}, \, v'= (v'_\beta)_{\beta \in \mathcal{X}_d} \, \in \R^{\mathcal{X}_d}$ with only finitely many non-zero coordinates, our goal is to show
	\begin{equation}\label{goal_after_levy}
		\log\Big( \Expe\Big[e^{i v\cdot y + i v' \cdot y'}\Big] \Big) \xrightarrow[\lambda \to \infty]{} \log\Big( \Expe\Big[e^{i v\cdot z + i v' \cdot z'}\Big] \Big)
	\end{equation}
	where we use the notation $v \cdot y := \sum_\beta v_\beta y_\beta$. \\
	The vector $(z, z')$ being Gaussian, we can immediately compute the right-hand side:
	\begin{align*}
		\log\Big( \Expe\Big[e^{i v\cdot z + i v' \cdot z'}\Big] \Big) &=  - \frac{1}{2} \, (v, v')^\top \, \mathrm{Cov}(z, z') \, (v, v')^\top \\
		&=   - \frac{1}{2} \, (v_{\beta_1}, v'_{\beta_1}, v_{\beta_2}, \dots) \, \mathrm{Cov}(Z)\, (v_{\beta_1}, v'_{\beta_1}, v_{\beta_2}, \dots)^\top \\
		&=   - \frac{1}{2} \, \sum_{\beta \in \mathcal{X}_d} (v_\beta)^2 + (v'_\beta)^2 + 2 \sqrt{ss'}^{|\beta|} \, v_\beta v'_\beta.
	\end{align*}
	The left-hand side is more involved. Recall that $N_\tau = \Delta_\tau + \sum_{\tau' \in \mathcal{X}_d} \Gamma_{\tau, \tau'}$ and $N'_{\tau'} = \Delta'_{\tau'} + \sum_{\tau \in \mathcal{X}_d} \Gamma_{\tau, \tau'}$ with
	\begin{align*}
		\Delta_\tau \sim \mathrm{Poi}\Big(\lambda s (1-s') \,\mathrm{GW}_d^{(\lambda s)}(\tau)\Big)&, \Delta'_{\tau'} \sim \mathrm{Poi}\Big(\lambda s' (1-s) \, \mathrm{GW}_d^{(\lambda s')}(\tau')\Big), \\ 
		\text{ and } \Gamma_{\tau, \tau'} &\sim \mathrm{Poi}\Big(\lambda s s' \, \Pcorrd(\tau, \tau')\Big).
	\end{align*}
	Using the notations $v \cdot f_{d, \bullet}^{(\mu) }(\tau) := \sum_{\beta \in \mathcal{X}_d} v_\beta \, f_{d, \beta}^{(\mu) }(\tau)$ and $\Gamma_{\tau, \bullet} = \sum_{\tau'} \Gamma_{\tau, \tau'}$, compute
	\begin{align*}
		&e^{i v\cdot y + i v' \cdot y'} = \exp\Big[\frac{i}{\sqrt{\lambda s}} \sum_{\tau \in \mathcal{X}_d}  v \cdot f_{d, \bullet}^{(\lambda s)}(\tau)\Big(\Delta_\tau + \Gamma_{\tau, \bullet} - \lambda s \mathrm{GW}^{(\lambda s)}_d(\tau) \Big) \Big]\\
		&\qquad \qquad \qquad  \times \exp \Big[ \frac{i}{\sqrt{\lambda s'}} \sum_{\tau' \in \mathcal{X}_d}  v' \cdot f_{d, \bullet}^{(\lambda s')}(\tau')\Big(\Delta'_{\tau'} + \Gamma_{\bullet, \tau'} - \lambda s' \mathrm{GW}^{(\lambda s')}_d(\tau') \Big)  \Big]\\
		& \quad = \exp\Big[ - i \sqrt{\lambda s} \sum_{\tau }  v \cdot f_{d, \bullet}^{(\lambda s)}(\tau)\mathrm{GW}^{(\lambda s)}_d(\tau) - i \sqrt{\lambda s'} \sum_{\tau' }  v' \cdot f_{d, \bullet}^{(\lambda s')}(\tau)\mathrm{GW}^{(\lambda s')}_d(\tau') \Big]\\
		& \quad \quad \quad \times \prod_{\tau, \tau' } \exp \Big(\frac{i}{\sqrt{\lambda s}}  v\cdot f_{d, \bullet}^{(\lambda s)}(\tau) + \frac{i}{\sqrt{\lambda s'}}  v'\cdot f_{d, \bullet}^{(\lambda s')}(\tau') \Big)^{\Gamma_{\tau, \tau'}} \\
		& \quad \quad \quad \times \prod_{\tau} \exp \Big(\frac{i}{\sqrt{\lambda s}}  v\cdot f_{d, \bullet}^{(\lambda s)}(\tau) \Big)^{\Delta_\tau} \times \prod_{\tau'} \exp \Big(\frac{i}{\sqrt{\lambda s'}}  v'\cdot f_{d, \bullet}^{(\lambda s')}(\tau') \Big)^{\Delta'_{\tau'}}.
	\end{align*}
	Since $\Delta_\tau, \Delta'_{\tau'},$ and $\Gamma_{\tau, \tau'}$ are independent Poisson variables, taking the expectation and then the logarithm on both sides yields
	\begin{align*}
		&\log \Expe\Big[e^{i v\cdot y + i v' \cdot y'}\Big] \\
        &= - i \sqrt{\lambda s} \sum_{\tau }  v \cdot f_{d, \bullet}^{(\lambda s)}(\tau)\mathrm{GW}^{(\lambda s)}_d(\tau) - i \sqrt{\lambda s'} \sum_{\tau' }  v' \cdot f_{d, \bullet}^{(\lambda s')}(\tau)\mathrm{GW}^{(\lambda s')}_d(\tau') \\
		&\quad \quad + \sum_{\tau, \tau'} \lambda s s' \, \Pcorrd(\tau, \tau') \Big( \exp\Big[ \frac{i}{\sqrt{\lambda s}}v\cdot  f_{d, \bullet}^{(\lambda s)}(\tau) + \frac{i}{\sqrt{\lambda s'}} v'\cdot  f_{d, \bullet}^{(\lambda s')}(\tau) \Big] - 1\Big)\\
		&\quad \quad + \sum_{\tau} \lambda s (1-s') \mathrm{GW}^{(\lambda s)}_d (\tau) \Big(\exp\Big[\frac{i }{ \sqrt{\lambda s}}v\cdot  f_{d, \bullet}^{(\lambda s)}(\tau)\Big]  - 1\Big)  \\
		&\quad \quad + \sum_{\tau'} \lambda s' (1-s) \mathrm{GW}^{(\lambda s')}_d (\tau') \Big(\exp\Big[\frac{i}{\sqrt{\lambda s'}} v'\cdot  f_{d, \bullet}^{(\lambda s')}(\tau) \Big]  - 1\Big).
	\end{align*}
	The next step consists in developing all the exponentials as power series. Writing $w(\tau) := v \cdot f_{d, \bullet}^{(\lambda s)}(\tau)$ and $w'(\tau') := v' \cdot f_{d, \bullet}^{(\lambda s')}(\tau')$, one has
	\begin{align*}
		&\log   \Expe\Big[e^{i v\cdot y + i v' \cdot y'}\Big] = \\
        &- i \sqrt{\lambda s} \sum_{\tau }  w(\tau) \, \mathrm{GW}^{(\lambda s)}_d(\tau) - i \sqrt{\lambda s'} \sum_{\tau' }  w'(\tau') \mathrm{GW}^{(\lambda s')}_d(\tau') \\
		&  + \sum_{\tau, \tau'} \lambda s s'  \Pcorrd(\tau, \tau') \Big( \frac{i w(\tau)}{\sqrt{\lambda s}}+ \frac{i  w'(\tau)}{\sqrt{\lambda s'}}  - \frac{w(\tau)^2}{2\lambda s} - \frac{w'(\tau)^2}{2\lambda s'} - \frac{w(\tau)w'(\tau')}{\lambda \sqrt{ss'}} + \mathcal{O}(\lambda^{-\frac{3}{2}})\Big)\\
		& + \sum_{\tau} \lambda s (1-s') \mathrm{GW}^{(\lambda s)}_d (\tau) \Big(\frac{i w(\tau)}{ \sqrt{\lambda s}} - \frac{w(\tau)^2}{2\lambda s} + \mathcal{O}(\lambda^{-\frac{3}{2}})\Big)  \\
		& + \sum_{\tau'} \lambda s' (1-s) \mathrm{GW}^{(\lambda s')}_d (\tau') \Big(\frac{i w'(\tau')}{ \sqrt{\lambda s'}} - \frac{w'(\tau')^2}{2\lambda s'} + \mathcal{O}(\lambda^{-\frac{3}{2}})\Big).
	\end{align*}
	To understand how the right hand side behaves for $\lambda \to \infty$, we start by the terms in $\mathcal{O}(\lambda^{-\sfrac{3}{2}})$. Each of these terms is multiplied by $\lambda$, leaving them still of order $\lambda^{-1/2}$ which converges to $0$. In other words, these terms vanish in the limit.
	
	We move on to extracting all the terms or order $\sqrt\lambda$ which cancel each other out:
	\begin{align*}
		& - i \sqrt{\lambda s} \sum_{\tau }  w(\tau) \, \mathrm{GW}^{(\lambda s)}_d(\tau) - i \sqrt{\lambda s'} \sum_{\tau' }  w'(\tau') \mathrm{GW}^{(\lambda s')}_d(\tau') \\
		&  + i s' \sqrt{\lambda s}\sum_{\tau, \tau'} \Pcorrd(\tau, \tau') w(\tau) + i s \sqrt{\lambda s'}\sum_{\tau, \tau'} \Pcorrd(\tau, \tau') w'(\tau')\\
		& + i (1-s') \sqrt{\lambda s}\sum_{\tau}\mathrm{GW}^{(\lambda s)}_d (\tau) w(\tau) + i (1-s) \sqrt{\lambda s'}\sum_{\tau'}  \mathrm{GW}^{(\lambda s')}_d (\tau') w'(\tau') \,  = \, 0.
	\end{align*}
	For this computation, we have used the identities $\sum_\tau \Pcorrd(\tau, \tau') = \mathrm{GW}^{(\lambda s')}_d (\tau')$ as well as $\sum_{\tau'}\Pcorrd(\tau, \tau') = \mathrm{GW}^{(\lambda s)}_d (\tau)$. Next, the terms of order $\lambda^0 = 1$ are
	\begin{align*}
		&-\frac{s'}{2}\sum_{\tau, \tau'}  \, \Pcorrd(\tau, \tau') w(\tau)^2  -\frac{s}{2}\sum_{\tau, \tau'}  \, \Pcorrd(\tau, \tau') w'(\tau)^2\\
        & \quad \quad \quad -\sqrt{ss'}\sum_{\tau, \tau'}  \, \Pcorrd(\tau, \tau') w(\tau)w'(\tau')\\
		& \quad \quad \quad - \frac{1-s'}{2}\sum_{\tau} \mathrm{GW}^{(\lambda s)}_d (\tau)w(\tau)^2 - \frac{1-s}{2}\sum_{\tau'} \mathrm{GW}^{(\lambda s')}_d (\tau') w'(\tau')^2\\
		&= -\frac{1}{2} \sum_{\tau} \mathrm{GW}^{(\lambda s)}_d (\tau)w(\tau)^2  -  \frac{1}{2}\sum_{\tau'} \mathrm{GW}^{(\lambda s')}_d (\tau') w'(\tau')^2\\
        & \quad \quad \quad -\sqrt{ss'}\sum_{\tau, \tau'}  \Pcorrd(\tau, \tau') w(\tau)w'(\tau') \\
		&\overset{(\star)}{=\joinrel=} - \frac{1}{2}\sum_\beta (v_\beta)^2 - \frac{1}{2}\sum_\beta (v'_\beta)^2 - \sum_\beta v_\beta v'_\beta \sqrt{ss'}^{|\beta| - 1}
	\end{align*}
	The last expression equals $\log \Expe\Big[e^{i v\cdot z + i v' \cdot z'}\Big]  $ which will conclude the proof. All that remains to show is therefore the equality $(\star)$ which requires some additional computations. We start by using the first orthogonality property of $f^{\lambda s}_{d,\beta}$ (\ref{orthogonaliy_1_GW}) to compute
	\begin{align*}
		\sum_\tau &\mathrm{GW}^{(\lambda s)}_d (\tau) w(\tau)^2 = \sum_\tau \mathrm{GW}^{(\lambda s)}_d (\tau) \Big( \sum_\beta v_\beta f_{d, \beta}^{(\lambda s)}(\tau) \Big)^2\\
		& = \sum_{\beta, \gamma} v_\beta v_{\gamma} \sum_\tau \mathrm{GW}^{(\lambda s)}_d (\tau) f_{d, \beta}^{(\lambda s)}(\tau) f_{d, \gamma}^{(\lambda s)}(\tau) = \sum_{\beta, \gamma} v_\beta v_\gamma \indicator{\beta = \gamma} = \sum_\beta (v_\beta)^2.
	\end{align*}
	Similarly, $\sum_{\tau'} \mathrm{GW}^{(\lambda s')}_d (\tau') w'(\tau')^2  = \sum_\beta (v'_\beta)^2$. 
	Next, the equation (\ref{appendix:likhoodratio_bigformula}) under the form 
	\[
	\Pcorrd(\tau, \tau') =  \mathrm{GW}^{(\lambda s)}_{d}(\tau)\,  \mathrm{GW}^{(\lambda s')}_{d}(\tau') \sum_{\delta \in \mathcal{X}_d}  \sqrt{ss'}^{|\delta|-1} f^{(\lambda s)}_{d, \delta}(\tau) f^{(\lambda s')}_{d, \delta}(\tau')
	\]
	implies the following:
	\begin{align*}
		\sum_{\tau, \tau'}  & \, \Pcorrd(\tau, \tau') w(\tau)w'(\tau') \\
		&=\sum_{\tau, \tau'}  \mathrm{GW}^{(\lambda s)}_{d}(\tau) \mathrm{GW}^{(\lambda s')}_{d}(\tau')  \sum_{\delta}\ \sqrt{ss'}^{|\delta| - 1}f^{(\lambda s)}_{d, \delta}(\tau) f^{(\lambda s')}_{d, \delta}(\tau') \\
        &\qquad \qquad \qquad \qquad \qquad \times \sum_{\beta, \gamma} v_\beta f_{d, \beta}^{(\lambda s)}(\tau)\,  v'_{\gamma} f_{d, \gamma}^{(\lambda s')}(\tau')\\
		&= \sum_{\beta, \gamma, \delta} v_\beta v'_\gamma \sqrt{ss'}^{|\delta| - 1} \Big(\sum_\tau \mathrm{GW}^{(\lambda s)}_{d}(\tau) f^{(\lambda s)}_{d, \delta}(\tau) f_{d, \beta}^{(\lambda s)}(\tau)\Big)\\
        &\qquad \qquad \qquad \qquad \qquad \times \Big(\sum_{\tau'} \mathrm{GW}^{(\lambda s')}_{d}(\tau')f^{(\lambda s')}_{d, \delta}(\tau') f_{d, \gamma}^{(\lambda s')}(\tau') \Big)\\
		&= \sum_{\beta, \gamma, \delta} v_\beta v'_\gamma \sqrt{ss'}^{|\delta| - 1}\indicator{\delta = \beta} \indicator{\delta = \gamma} = \frac{1}{\sqrt{ss'}}\sum_\beta v_\beta v'_\beta \sqrt{ss'}^{|\beta|}.
	\end{align*}
	Putting everything together, the equality $(\star)$ in our earlier computation is now immediate. This finishes the proof of  
	\[
	\log \Expe\Big[e^{iv\cdot y + iv'\cdot y'}\Big] \xrightarrow[\lambda \to \infty]{} \log\Expe\Big[e^{iv\cdot z + iv'\cdot z'}\Big] \quad \text{which implies} \quad \mathcal{L}(y, y') \xrightarrow[\lambda \to \infty]{w} \mathcal{L}(z,z').
	\]
\end{proof}

\subsection{Proof of Lemma \ref{Z_moments}}\label{appendix:Z_moments}

We recall the technical lemma about the moments of $Z$ before providing proof of it.

\begin{lemma}[copy of Lemma \ref{Z_moments}]
	The random variable $Z := Z(t, t')$ has the following moments under $\Pind$ and $\Pind$:
	\begin{enumerate}[label=(\alph*)]
		\item $\Einddplus\Big[Z\Big] = 0$,
		\item $\Ecorrdplus\Big[Z\Big] = \lambda s s' \,\Pcorr_d\Big(\mathcal{T}_d = 1\Big)$,
		\item $\Einddplus\Big[Z^4\Big] \leq 36 \, (\lambda^2 s s')^2 \,\Pind_d\Big(\mathcal{T}_d = 1\Big)^2 + 13 \lambda^3 s s' \, \Pind_d\Big(\mathcal{T}_d = 1\Big)$,
		\item $\mathrm{Var}_{d+1}^\text{corr}\Big[Z\Big] \leq \lambda^2 s s'\Big(1 + ss' \Big) \,\Pind_d\Big(\mathcal{T}_d= 1\Big)  +  \Ecorrdplus[Z] $ 
	\end{enumerate}
\end{lemma}

\begin{proof}
	We start by introducing centered versions of the poisson variables $N_{\tau}$ and $N'_{\tau'}$, namely
	\[
	\widetilde{N}_\tau := N_\tau - \lambda s \, \mathrm{GW}_d^{(\lambda s)} (\tau) \quad \text{and} \quad \widetilde{N}'_{\tau'} := N'_{\tau'} - \lambda s' \,\mathrm{GW}_d^{(\lambda s')} (\tau'),
	\]
	which lets us write $Z = Z(t, t') = \sum_{(\tau, \tau'):\, \mathcal{T}_d(\tau, \tau') = 1}  \widetilde{N}_\tau \widetilde{N}'_{\tau'}$. 
	
	What follows are the proofs of points (a), (b), (c), and (d) in their alphabetic order.
	
	\paragraph{Proof of (a). } This is immediate from the independence of  $(N_\tau)_\tau$ and $(N'_{\tau'})_{\tau '}$ under $\Pinddplus$ and the centeredness of $\widetilde{N}_\tau$ and $\widetilde{N}'_{\tau'}$:
	\[
	\Einddplus\Big[Z\Big] = \sum_{\substack{(\tau, \tau') \in \mathcal{X}_d^2,\\ \mathcal{T}_d(\tau, \tau') = 1}} \Einddplus\Big[ \widetilde{N}_\tau \Big] \Einddplus\Big[\widetilde{N}'_{\tau'} \Big] = 0.
	\]
	
	\paragraph{Proof of (b)} Recall that if $\Big((N_\tau)_\tau, (N'_{\tau'})_{\tau'}\Big) \sim\Pcorrdplus$, then we can write $N_\tau = \Delta_\tau + \sum_{\tau' \in \mathcal{X}_d} \Gamma_{\tau, \tau'}$ and $N'_{\tau'} = \Delta'_{\tau'} + \sum_{\tau \in \mathcal{X}_d} \Gamma_{\tau, \tau'}$ where
	\begin{align*}
		\Delta_\tau \sim \mathrm{Poi}\Big(\lambda s (1-s') \,\mathrm{GW}_d^{(\lambda s)}(\tau)\Big)&, \quad \Delta'_{\tau'} \sim \mathrm{Poi}\Big(\lambda s' (1-s) \, \mathrm{GW}_d^{(\lambda s')}(\tau')\Big), \\ 
		\text{ and } \Gamma_{\tau, \tau'} &\sim \mathrm{Poi}\Big(\lambda s s' \, \Pcorrd(\tau, \tau')\Big)
	\end{align*}
	are independent poisson variables. For each one of these variables, we define their centered versions
	\begin{align*}
		\widetilde \Delta_\tau := \Delta_\tau - \lambda s (1-s') \,\mathrm{GW}_d^{(\lambda s)}(\tau)&,\quad \widetilde \Delta'_{\tau'} := \Delta'_{\tau'}  - 	\lambda s' (1-s) \, \mathrm{GW}_d^{(\lambda s)}(\tau), \\ 
		\text{ and } \widetilde\Gamma_{\tau, \tau'} := &\,  \Gamma_{\tau, \tau'} - \lambda s s' \, \Pcorrd(\tau, \tau').
	\end{align*}
	which fulfill $\widetilde N_\tau = \widetilde\Delta_\tau + \sum_{\tau' \in \mathcal{X}_d} \widetilde\Gamma_{\tau, \tau'}$ and $\widetilde N'_{\tau'} = \widetilde \Delta'_{\tau'} + \sum_{\tau \in \mathcal{X}_d} \widetilde \Gamma_{\tau, \tau'}$ because of the marginals $\sum_{\tau'}  \Pcorrd(\tau, \tau') = \mathrm{GW}_d^{(\lambda s)}(\tau)$ and $\sum_{\tau'}  \Pcorrd(\tau, \tau') = \mathrm{GW}_d^{(\lambda s')}(\tau')$.\\
	Exploiting the independence of $\Delta_\tau, \Delta'_{\tau'}$ and $\Gamma_{\tau, \tau'}$, we conclude the proof of (b) with the following computation:
	\begin{align*}
		\Ecorrdplus\Big[Z\Big] &= \sum_{\substack{(\tau, \tau') \in \mathcal{X}_d^2,\\ \mathcal{T}_d(\tau, \tau') = 1}} \Ecorrdplus\Big[ \widetilde{N}_\tau \widetilde{N}'_{\tau'} \Big] = \sum_{\substack{(\tau, \tau') \in \mathcal{X}_d^2,\\ \mathcal{T}_d(\tau, \tau') = 1}} \Ecorrdplus\Big[ \Big( \widetilde{\Gamma}_{\tau, \tau'}\Big)^2 \Big] \\
		&= \lambda s s'  \,\Pcorrd\Big(\mathcal{T}_d(\tau, \tau') = 1\Big).
	\end{align*}
	\paragraph{Proof of (c).} We introduce the shortcut notation $\mathcal{S} := \{(\tau, \tau') \in \mathcal{X}_d^2 \suchthat \mathcal{T}_d(\tau, \tau') = 1\}$ and develop the forth moment of $Z$ as
	\begin{equation}\label{expandZ4}
		\Einddplus\Big[Z^4 \Big] = \Einddplus\Big[ \Big( \sum_{(\tau, \tau') \in \mathcal{S}}  \widetilde{N}_\tau \widetilde{N}'_{\tau'} \Big)^ 4\Big] = \sum_{\substack{(\tau_i, \tau'_i) \in \mathcal{S}, \\ i \in \{1,..., 4\}}} \Einddplus\Big[ \prod_{i = 1}^4 \widetilde{N}_{\tau_i}\Big] \Einddplus\Big[ \prod_{i = 1}^4 \widetilde{N}'_{\tau'_i}\Big]. 
	\end{equation}
	Since $\widetilde{N}_{\tau}$ is independent from $\widetilde{N}_{\theta}$ for $ \tau \neq \theta$ and both random variables are centered, we have that  $ \Einddplus\Big[ \prod_{i = 1}^4 \widetilde{N}_{\tau_i}\Big] = 0 $ in all but the two following cases:
	\begin{itemize}
		\item All four $\tau_i$ are equal, i.e. $|\{\tau_1, \tau_2, \tau_3, \tau_4\}| = 1$, 
		\item The $\tau_i$ are pairwise equal but not all the same, i.e. $|\{\tau_1, \tau_2, \tau_3, \tau_4\}| = 2$.
	\end{itemize}
	The identical observation is true for $\Einddplus\Big[ \prod_{i = 1}^4 \widetilde{N}'_{\tau'_i}\Big]$. Let
	\[
	\mathcal{S}(m, r) := \Big\{ (\tau_i, \tau'_i)_{i=1}^4 \in \mathcal{S}^4 \suchthat |\{\tau_1, \tau_2, \tau_3, \tau_4\}| = m \text{ and } |\{\tau'_1, \tau'_2, \tau'_3, \tau'_4\}|  = r\Big\},
	\]
	and denote $\Sigma(r,m)$ the sum from (\ref{expandZ4}) restricted to summands in $\mathcal{S}(r, m)$ instead of $\mathcal{S}$. This  lets us write $\Einddplus\Big[Z^4 \Big]  = \Sigma(1,1) + \Sigma (1, 2) + \Sigma(2, 1) + \Sigma(2, 2)$ and we are left with examining the terms $\Sigma(r, m)$ one by one.
	
	Recall that for a Poisson variable $X\sim \mathrm{Poi}(\mu)$, the fourth moment of its centered version $\widetilde{X} = X - \mu$ is equal to $\Expe[(\widetilde{X})^4] = 3\mu^2 + \mu$. This lets us compute
	\begin{align*}
		\Sigma(1, 1) & =\sum_{(\tau, \tau') \in \mathcal{S}} \Einddplus\Big[  (\widetilde{N}_{\tau})^4\Big] \Einddplus\Big[ (\widetilde{N}'_{\tau'})^4\Big] \\
		&= \sum_{(\tau, \tau') \in \mathcal{S}} \Big(3(\lambda s)^2 \mathrm{GW}_d^{(\lambda s)}(\tau) ^2 + \lambda s \mathrm{GW}_d^{(\lambda s)}(\tau)\Big) \\
		&\qquad \qquad \quad \times \Big(3(\lambda s')^2 \mathrm{GW}_d^{(\lambda s')}(\tau') ^2 + \lambda s' \mathrm{GW}_d^{(\lambda s')}(\tau')\Big)\\
		&\leq  \lambda^2 ss' \sum_{(\tau, \tau') \in \mathcal{S}} 9 \lambda^2 ss'\,  \mathrm{GW}_d^{(\lambda s)}(\tau) ^2 \mathrm{GW}_d^{(\lambda s')}(\tau') ^2\\
		&\qquad \qquad \qquad \quad+  \Big(3 \lambda s + 3 \lambda s' + 1\Big)\mathrm{GW}_d^{(\lambda s)}(\tau) \mathrm{GW}_d^{(\lambda s')}(\tau')\\
		&\leq 9 (\lambda^2 s s')^2 \Pindd(\mathcal{S})^2 + 7 \lambda^3 ss'\, \Pindd(\mathcal{S})
	\end{align*}
	where we have used $\mathrm{GW}_d^{(\mu)}(\tau) \leq 1$ in the first inequality. The second inequality relies on $ \sum_{i = 0}^\infty p_i^2 \leq (\sum_i p_i)^2$ for $p_i \geq 0$, as well as $\lambda \geq 1$ and
	\[
	\Pindd(\mathcal{S}) = \sum_{(\tau, \tau') \in \mathcal{S}}  \Pindd(\tau, \tau') = \sum_{(\tau, \tau') \in \mathcal{S}}  \mathrm{GW}_d^{(\lambda s)}(\tau) \mathrm{GW}_d^{(\lambda s')}(\tau').
	\]
	Next, we compute
	\begin{align*}
		\Sigma(1,2) &= \sum_{\tau \in \mathcal{X}_d} \Einddplus\Big[  (\widetilde{N}_{\tau})^4\Big] \binom{4}{2} \sum_{\substack{\tau' : (\tau, \tau') \in \mathcal{S}\\ \theta': (\tau, \theta') \in \mathcal{S}}} \Einddplus\Big[ (\widetilde{N}'_{\tau'})^2\Big] \Einddplus\Big[ (\widetilde{N}'_{\theta'})^2\Big] \underbrace{\indicator{ \tau' \neq \theta'}}_{\leq 1}\\
		&\leq 3\sum_{\tau \in \mathcal{X}_d} \Big(3(\lambda s)^2 \mathrm{GW}_d^{(\lambda s)}(\tau) ^2 + \lambda s \mathrm{GW}_d^{(\lambda s)}(\tau)\Big)\\
		&\qquad \qquad  \times\sum_{\substack{\tau' : (\tau, \tau') \in \mathcal{S}\\ \theta': (\tau, \theta') \in \mathcal{S}}} (\lambda s')^2\mathrm{GW}_d^{(\lambda s')}(\tau')\mathrm{GW}_d^{(\lambda s')}(\theta')\\
		&= 9 (\lambda^2 s s')^2\sum_\tau  \mathrm{GW}_d^{(\lambda s)}(\tau) ^2 \sum_{\substack{\tau' : (\tau, \tau') \in \mathcal{S}\\ \theta': (\tau, \theta') \in \mathcal{S}}} \mathrm{GW}_d^{(\lambda s')}(\tau')\mathrm{GW}_d^{(\lambda s')}(\theta')\\
		&\quad \quad  +  3 \lambda^3 s (s')^2 \sum_{(\tau, \tau') \in \mathcal{S}} \mathrm{GW}_d^{(\lambda s)}(\tau)  \mathrm{GW}_d^{(\lambda s')}(\tau') \sum_{\theta': (\tau, \theta') \in \mathcal{S}} \mathrm{GW}_d^{(\lambda s')}(\theta').
	\end{align*}
	Using $\sum_i p_i^2 \leq (\sum_i p_i)$ again, we obtain
	\begin{align*}
		&\sum_\tau  \mathrm{GW}_d^{(\lambda s)}(\tau) ^2 \sum_{\substack{\tau' : (\tau, \tau') \in \mathcal{S}\\ \theta': (\tau, \theta') \in \mathcal{S}}} \mathrm{GW}_d^{(\lambda s')}(\tau')\mathrm{GW}_d^{(\lambda s')}(\theta') \\
		& = \sum_\tau  \sum_{\tau' : (\tau, \tau') \in \mathcal{S}} \mathrm{GW}_d^{(\lambda s)}(\tau) \mathrm{GW}_d^{(\lambda s')}(\tau')  \sum_{\theta' : (\tau, \theta') \in \mathcal{S}} \mathrm{GW}_d^{(\lambda s)}(\tau) \mathrm{GW}_d^{(\lambda s')}(\theta')\\
		&= \sum_\tau  \Big(\sum_{\tau' : (\tau, \tau') \in \mathcal{S}} \mathrm{GW}_d^{(\lambda s)}(\tau) \mathrm{GW}_d^{(\lambda s')}(\tau') \Big)^2 \\
		&\leq  \Big(\sum_{(\tau, \tau') \in \mathcal{S}}  \mathrm{GW}_d^{(\lambda s)}(\tau) \mathrm{GW}_d^{(\lambda s')}(\tau') \Big)^2 =\Pind_d(\mathcal{S})^2.
	\end{align*}
	Plugging this into our earlier expression for $\Sigma(1, 2)$ and using $s' \leq 1$ as well as the inequality $\sum_{\theta': (\tau, \theta') \in \mathcal{S}} \mathrm{GW}_d^{(\lambda s')}(\theta') \leq 1$, one obtains
	\[
	\Sigma(1, 2) \leq 9 (\lambda^2 s s')^2\Pind_d(\mathcal{S})^2 + 3 \lambda^3 s s' \, \Pind_d(\mathcal{S})
	\]
	Since this last term is invariant to swapping $s$ and $s'$, an analog computation yields the same upper bound for $\Sigma(2, 1)$. 
	
	The final term to examine is $\Sigma(2,2)$. Since every combination of $\tau_i$'s we sum over fulfills $\{\tau_i \suchthat i \in [4]\} = \{\tau, \theta\}$ for $\tau \neq \theta$, there are $ \binom{4}{2} = 3$ choices for $(\tau, \theta)$. The same thing holds for the summation over $\tau'_i$. Consequently,
	\begin{align*}
		&\Sigma(2, 2) \\
		&= \binom{4}{2} \sum_{\substack{\tau, \theta \in \mathcal{X}_d,\\ \tau \neq \theta}} \Einddplus\Big[  (\widetilde{N}_{\tau})^2\Big] \Einddplus\Big[  (\widetilde{N}_{\theta})^2\Big] \binom{4}{2}\sum_{\substack{\tau' : (\tau, \tau') \in \mathcal{S}\\ \theta': (\tau, \theta') \in \mathcal{S},\\ \tau' \neq \theta'}} \Einddplus\Big[ (\widetilde{N}'_{\tau'})^2\Big] \Einddplus\Big[ (\widetilde{N}'_{\theta'})^2\Big] \\
		&\leq 9 \sum_{\substack{\tau, \theta \in \mathcal{X}_d}} 	(\lambda s)^2\mathrm{GW}_d^{(\lambda s)}(\tau)\mathrm{GW}_d^{(\lambda s)}(\theta) \sum_{\substack{\tau' : (\tau, \tau') \in \mathcal{S}\\ \theta': (\tau, \theta') \in \mathcal{S}}} 	(\lambda s')^2\mathrm{GW}_d^{(\lambda s')}(\tau')\mathrm{GW}_d^{(\lambda s')}(\theta') \\
		&= 9 \, (\lambda^2 s s')^2 \sum_{(\tau, \tau') \in \mathcal{S}}  \mathrm{GW}_d^{(\lambda s)}(\tau) \mathrm{GW}_d^{(\lambda s')}(\tau') \sum_{(\theta, \theta') \in \mathcal{S}}  \mathrm{GW}_d^{(\lambda s')}(\theta'     \mathrm{GW}_d^{(\lambda s)}(\theta)\\
		&= 9 \, (\lambda^2 s s')^2 \,\Pind_d(\mathcal{S})^2. 
	\end{align*}
	Putting all things together, we obtain the desired inequality:
	\[
	\Einddplus\Big[Z^4 \Big]  = \sum_{r,m = 1}^{2}\Sigma(r, m) \leq 36 \, (\lambda^2 s s')^2 \,\Pind_d(\mathcal{S})^2 + 13 \lambda^3 s s' \, \Pind_d(\mathcal{S})
	\]

	\paragraph{Proof of (d).} Reusing the notations from the proof of (b), we can write
	\begin{align*}
		&\Ecorrdplus\Big[Z^2\Big] \\
		& =\hspace{-8pt} \sum_{\substack{(\tau, \tau') \in \mathcal{S},\\ (\theta, \theta') \in \mathcal{S}}}  \hspace{-8pt}\Ecorrdplus\Big[\Big( \widetilde\Delta_\tau + \textstyle \sum_{\beta'} \widetilde\Gamma_{\tau, \beta'}\Big) \Big(\widetilde \Delta'_{\tau'} + \textstyle\sum_{\beta} \widetilde \Gamma_{\beta, \tau'}\Big)\Big( \widetilde\Delta_\theta + \textstyle\sum_{\gamma'} \widetilde\Gamma_{\theta, \gamma'}\Big) \Big(\widetilde \Delta'_{\theta'} + \textstyle \sum_{\gamma} \widetilde \Gamma_{\gamma, \theta'}\Big)\Big].
	\end{align*}
	The product inside the expectation expands into a polynomial in the variables $( \widetilde \Delta, \widetilde \Delta', \Gamma )$ indexed accordingly by trees. Since these random variables are independent and centered, the only monomes not being canceled by the expectation are of the following forms:
	\begin{itemize}
		\item $\widetilde{\Delta}_\tau^2 \Big(\widetilde{\Delta}'_{\tau'}\Big)^2$; denote the sum over these monomes as $\Sigma(2, 2, 0)$,
		\item $\widetilde{\Delta}_\tau^2 \Big(\textstyle\sum_{\beta} \widetilde \Gamma_{\beta, \tau'}\Big)^2$ ; denote the sum over these monomes as $\Sigma(2, 0, 2)$,
		\item $\Big(\widetilde{\Delta}'_{\tau'}\Big)^2 \Big(\textstyle\sum_{\gamma'} \widetilde\Gamma_{\tau, \gamma'}\Big)^2$;  denote the sum over these monomes as $\Sigma(0, 2, 2)$,
		\item $\Big(\textstyle\sum_{\beta} \widetilde \Gamma_{\beta, \tau'}\Big)^2 \Big(\textstyle\sum_{\gamma'} \widetilde\Gamma_{\tau, \gamma'}\Big)^2$; denote the sum over these monomes as $\Sigma(0, 0, 4)$.
	\end{itemize}
	As in the proof of (c), we compute an upper bound of these terms one by one, starting with
	\begin{align*}
		\Sigma(2,2,0) & = \sum_{(\tau, \tau') \in \mathcal{S}} \Ecorrdplus\Big[\widetilde{\Delta}_\tau^2 \Big] \Ecorrdplus\Big[ \Big(\widetilde{\Delta}'_{\tau'}\Big)^2\Big] \\
		&= \lambda^2 ss' (1-s) (1-s') \underbrace{\sum_{(\tau, \tau') \in \mathcal{S}} \mathrm{GW}_d^{(\lambda s)}(\tau) \, \mathrm{GW}_d^{(\lambda s')}(\tau')}_{=\Pind_d(\mathcal{S})}.
	\end{align*}
	Next,
	\begin{align*}
		\Sigma(2,0,2) & = \sum_{(\tau, \tau') \in \mathcal{S}} \Ecorrdplus\Big[\widetilde{\Delta}_\tau^2 \Big] \Ecorrdplus\Big[ \Big(\textstyle\sum_{\beta} \widetilde \Gamma_{\beta, \tau'}\Big)^2\Big] =  \lambda^2 s^2 s' (1-s') \,\Pind_d(\mathcal{S})
	\end{align*}
	and similarly $\Sigma(0, 2, 2) = \lambda^2 (s')^2 s (1-s) \,\Pind_d(\mathcal{S})$. The more involved term is
	\begin{align*}
		&\Sigma(0,0,4) = \sum_{\substack{(\tau, \tau') \in \mathcal{S},\\ (\theta, \theta') \in \mathcal{S} }} \, \, \sum_{\substack{ \beta, \beta' \in \mathcal{X}_d, \\ \gamma, \gamma' \in \mathcal{X}_d }} \Ecorrdplus\Big[  \widetilde \Gamma_{\tau, \beta'} \, \widetilde \Gamma_{\beta, \tau'}\, \widetilde \Gamma_{\theta, \gamma'}\, \widetilde \Gamma_{\gamma, \theta'}\Big]\\
		&= \sum_{ (\tau, \tau') \in \mathcal{S} } \Ecorrdplus[  \widetilde \Gamma_{\tau, \tau'}^4 ] \, +\,  \sum_{ (\tau, \tau') \in \mathcal{S} }\Ecorrdplus[  \widetilde \Gamma_{\tau, \tau'}^2 ] \sum_{\substack{(\theta, \theta') \in \mathcal{S},\\ (\theta, \theta') \neq (\tau, \tau')}}  \Ecorrdplus[  \widetilde \Gamma_{\theta, \theta'}^2 ]\\
		& \quad+  \sum_{\substack{(\tau, \tau') \in \mathcal{S} \\
				\beta, \beta' \in \mathcal{X}_d, \\ (\beta, \tau') \neq (\tau, \beta')}}   \Ecorrdplus[  \widetilde \Gamma_{\tau, \beta'}^2 ]  \Ecorrdplus[  \widetilde \Gamma_{\beta, \tau'}^2 ]  +\sum_{\substack{(\tau, \tau') \in \mathcal{S},\\ (\theta, \theta') \in \mathcal{S}, \\
				(\tau, \theta') \neq (\theta, \tau')}} \Ecorrdplus[  \widetilde \Gamma_{\tau, \theta'}^2 ] \Ecorrdplus[  \widetilde \Gamma_{\theta, \tau'}^2 ].
	\end{align*}
	Recalling $ \Gamma_{\tau, \tau'} \sim \mathrm{Poi}\Big( \lambda s s' \, \Pcorrd(\tau, \tau') \Big) $, we compute an upper bound the summands one by one. First, one has
	\begin{align*}
		\sum_{ (\tau, \tau') \in \mathcal{S} } \Ecorrdplus[  \widetilde \Gamma_{\tau, \tau'}^4 ] &= \sum_{ (\tau, \tau') \in \mathcal{S} } 3 (\lambda ss')^2 \Pcorrd(\tau, \tau') ^2 + \lambda s s' \, \Pcorrd(\tau, \tau')\\
		&= (\lambda s s') \, \Pcorrd(\mathcal{S}) +  3 (\lambda ss')^2  \sum_{ (\tau, \tau') \in \mathcal{S} }\Pcorrd(\tau, \tau') ^2,
	\end{align*}
	and 
	\begin{align*}
		&\sum_{ (\tau, \tau') \in \mathcal{S} }\Ecorrdplus[  \widetilde \Gamma_{\tau, \tau'}^2 ] \sum_{\substack{(\theta, \theta') \in \mathcal{S},\\ (\theta, \theta') \neq (\tau, \tau')}}  \Ecorrdplus[  \widetilde \Gamma_{\theta, \theta'}^2 ]\\
		&\qquad \qquad\qquad= (\lambda s s')^2\sum_{ (\tau, \tau') \in \mathcal{S}} \Pcorrd(\tau, \tau') \sum_{\substack{(\theta, \theta') \in \mathcal{S},\\ (\theta, \theta') \neq (\tau, \tau')}}  \Pcorrd(\theta, \theta')\\
		&\qquad \qquad\qquad=  (\lambda s s')^2 \Big(\Pcorrd(\mathcal{S})^2  - \sum_{ (\tau, \tau') \in \mathcal{S}} \Pcorrd(\tau, \tau')^2\Big).
	\end{align*}
	Next,
	\begin{align*}
		\sum_{ (\tau, \tau') \in \mathcal{S} }& \,  \, \sum_{\beta, \beta' \in \mathcal{X}_d}  \Ecorrdplus[  \widetilde \Gamma_{\tau, \beta'}^2 ]  \Ecorrdplus[  \widetilde \Gamma_{\beta, \tau'}^2 ]\,  \indicator{(\beta, \tau') \neq (\tau, \beta')} \\
		&= (\lambda ss')^2 \sum_{ (\tau, \tau') \in \mathcal{S}} \sum_{\beta \in \mathcal{X}_d}\Pcorr_d(\beta, \tau') \sum_{\beta' \in \mathcal{X}_d}\Pcorr_d(\tau, \beta') \indicator{(\beta, \tau') \neq (\tau, \beta')}\\
		&= (\lambda ss')^2 \sum_{ (\tau, \tau') \in \mathcal{S}} \sum_{\beta \in \mathcal{X}_d}\Pcorr_d(\beta, \tau') \Big( -\Pcorr_d(\beta, \tau') + \textstyle\sum_{\beta'}\Pcorr_d(\tau, \beta') \Big)\\
		&= - (\lambda ss')^2 \sum_{ (\tau, \tau') \in \mathcal{S}} \sum_{\beta \in \mathcal{X}_d}\Pcorr_d(\beta, \tau')^2\\
		& \qquad \qquad \qquad \qquad + (\lambda ss')^2 \sum_{ (\tau, \tau') \in \mathcal{S}} \sum_{\beta \in \mathcal{X}_d}\Pcorr_d(\beta, \tau') \mathrm{GW}^{(\lambda s')}_d(\tau')\\
		&\leq -  (\lambda ss')^2 \sum_{ (\beta, \tau') \in \mathcal{S}}\Pcorrd(\beta, \tau')^2 + (\lambda ss')^2 \sum_{ (\tau, \tau') \in \mathcal{S}} \mathrm{GW}^{(\lambda s')}_d(\tau')  \, ^{(\lambda s)}_d(\tau)  \\
		&= (\lambda ss')^2 \Big(\Pind_d(\mathcal{S}) - \sum_{ (\tau, \tau') \in \mathcal{S}} \Pcorrd(\tau, \tau')^2\Big)
	\end{align*}
	where we have used $\sum_{ (\tau, \tau') \in \mathcal{S}} \sum_{\beta \in \mathcal{X}_d}\Pcorr_d(\beta, \tau')^2 \geq \sum_{ (\beta, \tau') \in \mathcal{S}}\Pcorr_d(\beta, \tau')^2 $. This inequality can be recycled for the last summand: Using $\mathcal{S} \subset \mathcal{X}_d^2$ for a first upper bound, followed by the above inequality where $\beta$ is replaced by $\theta$, one obtains
	\begin{align*}
		\sum_{ (\tau, \tau') \in \mathcal{S} } \, \sum_{(\theta, \theta') \in \mathcal{S}} &\Ecorrdplus[  \widetilde \Gamma_{\tau, \theta'}^2 ]  \Ecorrdplus[  \widetilde \Gamma_{\theta, \tau'}^2 ]\,  \indicator{(\tau, \theta') \neq (\theta, \tau')} \\
		&\leq \sum_{ (\tau, \tau') \in \mathcal{S} } \,  \, \sum_{\theta, \theta' \in \mathcal{X}_d}  \Ecorrdplus[  \widetilde \Gamma_{\tau, \theta'}^2 ]  \Ecorrdplus[  \widetilde \Gamma_{\theta, \tau'}^2 ]\,  \indicator{(\theta, \tau') \neq (\tau, \theta')} \\
		&\leq	(\lambda ss')^2 \Big(\Pind_d(\mathcal{S}) - \sum_{ (\tau, \tau') \in \mathcal{S}} \Pcorrd(\tau, \tau')^2\Big)
	\end{align*}
	Putting all the information on the summands together, we obtain
	\begin{align*}
		\Sigma(0,0,4) & \leq  3 (\lambda ss')^2  \sum_{ (\tau, \tau') \in \mathcal{S} }\Pcorrd(\tau, \tau') ^2  + (\lambda s s')^2 \Big(\Pcorrd(\mathcal{S})^2  - \hspace{-10pt} \sum_{ (\tau, \tau') \in \mathcal{S}} \hspace{-8pt} \Pcorrd(\tau, \tau')^2\Big)\\
		&\quad \quad \quad + 2 (\lambda ss')^2 \Big(\Pind_d(\mathcal{S}) - \sum_{ (\tau, \tau') \in \mathcal{S}} \Pcorrd(\tau, \tau')^2\Big) + (\lambda s s') \, \Pcorrd(\mathcal{S}) \\
		&= (\lambda ss')^2 \, \Pcorrd(\mathcal{S})^2 + 2 (\lambda ss')^2\, \Pind_d(\mathcal{S}) +  (\lambda s s') \, \Pcorrd(\mathcal{S})\\
		&= 2 (\lambda ss')^2\, \Pind_d(\mathcal{S}) + \Ecorrdplus[Z]^2 +  \Ecorrdplus[Z]
	\end{align*}
	where the last equality follows from (b). This lets us conclude
	\begin{align*}
		\Ecorrdplus\Big[Z^2\Big] &= \Sigma(2,2,0) + \Sigma(2,0,2) + \Sigma(0,2,2) + \Sigma(0,0,4)\\
		&\leq \lambda^2 s s'\Big(1 + ss' \Big) \,\Pind_d(\mathcal{S})  + \Ecorrdplus[Z]^2 +  \Ecorrdplus[Z]
	\end{align*}
	where the factor multiplying $\Pind_d(\mathcal{S}) $ is obtained from
	\[
	\lambda ^2 s s' (1-s) (1-s') + \lambda^2 (s')^2 s (1-s) + \lambda^2 s^2 s' (1-s') + 2 \lambda^2 s^2(s')^2 = \lambda^2 s s'\Big(1 + ss' \Big). 
	\]
	and (d) follows from
	\[
	\mathrm{Var}_{d+1}^\text{corr}\Big[Z\Big]  = \Ecorrdplus\Big[Z^2\Big] - \Ecorrdplus\Big[Z\Big]^2 \leq \lambda^2 s s'\Big(1 + ss' \Big) \,\Pind_d(\mathcal{S})  +  \Ecorrdplus[Z].
	\]
\end{proof}

%------
% Insert acknowledgments and information
% regarding funding at the end of the last
% section, i.e., right before the bibliography.
%------

\begin{ack}
	We thank Luca Ganassali for insightful and enriching discussions, which helped us clarify the most intricate details. We also express our gratitude to the two anonymous reviewers for their very detailed and knowledgeable feedback. 
\end{ack}

\begin{funding}
	This work was partially supported by the French government under management of
	Agence Nationale de la Recherche as part of the “Investissements d’avenir” program, reference ANR19-P3IA-0001 (PRAIRIE 3IA Institute).
\end{funding}

%------
% Insert the bibliography.
%------

\bibliographystyle{emss}  % Uses "emss.bst"
\bibliography{references}  % Loads "references.bib"

\end{document}